%% file: DkSH-Curve.tex
\documentclass[11pt]{article}
\usepackage[utf8]{inputenc}
\usepackage{style}
\usepackage{framed}
\usepackage{nicefrac}
\usepackage[parfill]{parskip}
\usepackage{setspace}
\usepackage{hyperref}
\usepackage{mdwlist}
\usepackage{mdframed}
\usepackage[ruled,linesnumbered]{algorithm2e}

\input{macros}

\newcommand{\bphi}{\bar{\phi}}
\newcommand{\what}{\widetilde}
\renewcommand{\cP}{\mathcal{P}}

\newcommand{\Inf}[2]{{\sf Inf}_{#1}\left[#2\right]}

\newcommand{\ce}{{\sf e}}
\allowdisplaybreaks

\def\DEBUG{true}
\usepackage[colorinlistoftodos,textsize=tiny,textwidth=2cm,color=red!25!white,obeyFinal]{todonotes}

\ifdefined\DEBUG
	
	\newcommand{\el}[1]{\todo[color=blue!25!white]{EL: #1}\xspace}
\else
	
	\newcommand{\el}[1]{}
\fi

\title{On Lifting Integrality Gaps to SSEH Hardness \\ for Globally Constrained CSPs}
\author{ Suprovat Ghoshal \\ Northwestern University \& TTIC \\ suprovat.ghoshal@northwestern.edu
	\and Euiwoong Lee \\ University of Michigan \\ euiwoong@umich.edu}
\begin{document}
\begin{titlepage}
	\maketitle
	\begin{abstract}
		A $\mu$-constrained Boolean {\sc Max-CSP}$(\psi)$ instance is a Boolean Max-CSP instance on predicate $\psi:\{0,1\}^r \to \{0,1\}$ where the objective is to find a labeling of relative weight {\em exactly} $\mu$ that maximizes the fraction of satisfied constraints. In this work, we study the approximability of constrained Boolean Max-CSPs via SDP hierarchies by relating the integrality gap of {\sc Max-CSP}$(\psi)$ to its $\mu$-dependent approximation curve. \\
		
		Formally, assuming the Small-Set Expansion Hypothesis, we show that it is \NP-hard to approximate $\mu$-constrained instances of {\sc Max-CSP}($\psi$) up to factor ${\sf Gap}_{\ell,\mu}(\psi)/\log(1/\mu)^2$ (ignoring factors depending on $r$) for any $\ell \geq \ell(\mu,r)$. Here, ${\sf Gap}_{\ell,\mu}(\psi)$ is the optimal integrality gap of $\ell$-round Lasserre relaxation for $\mu$-constrained {\sc Max-CSP}($\psi$) instances. 
		
		Our results are derived by combining the framework of Raghavendra [STOC 2008] along with more recent advances in rounding Lasserre relaxations and reductions from the Small-Set Expansion (SSE) problem. A crucial component of our reduction is a novel way of composing generic bias-dependent dictatorship tests with SSE, which could be of independent interest.
	\end{abstract}	
	
\end{titlepage}

\tableofcontents

\input{intro}

\input{overview-lift}

\input{prelims-lift}

\input{pre-proc}

\input{redn-lift}
\input{completeness-lift}

\input{soundness-lift}
\input{rag-round}

\bibliographystyle{alpha}
\bibliography{main}

\appendix
\input{appendix-lift}

\end{document}

%% file: macros.tex
\usepackage[utf8]{inputenc}
\usepackage{style}
\usepackage{mdframed}
\usepackage{mdwlist}
\usepackage{nicefrac}
\usepackage{graphicx}
\usepackage[parfill]{parskip}
\usepackage{setspace}
\usepackage[colorlinks]{hyperref}
\usepackage[ruled,linesnumbered]{algorithm2e}
\usepackage{amsmath}
\usepackage{amssymb}
\usepackage{complexity}

\newcommand{\Ex}{\E}
\renewcommand{\cP}{{\cal P}}

\allowdisplaybreaks

\newcommand{\defeq}{\overset{\rm def}{=}}
\usepackage{microtype}

\usepackage{thm-restate}

\usepackage{boxedminipage}

\newcommand{\smallsetexpansion}{{\sc SmallSetExpansion}}

\newcommand{\uniquegames}{{\sc Unique Games}}

\newcommand{\cH}{\mathcal{H}}

\renewcommand{\cX}{\mathcal{X}}
\renewcommand{\cS}{\mathcal{S}}

\newcommand{\cE}{{\mathcal{E}}}

\newcommand{\wh}[1]{\widehat{#1}}
\newcommand{\Langle}{\left\langle}
\newcommand{\Rangle}{\right\rangle}
\newcommand{\non}{\nonumber}
\newcommand{\dksh}{\ensuremath{{\sf D}k{\sf SH}}}
\newcommand{\dks}{\ensuremath{{\sf D}k{\sf S}}}
\newcommand{\cM}{\mathcal{M}}
\newcommand{\uphi}{u_\emptyset}

\newcommand{\Pom}{{(\Omega_i)}}
\newcommand{\Pmu}{{(\mu_I)}}
\newcommand{\Pb}{{(\beta)}}

\newcommand{\sT}{\mathcal{T}}

%% file: intro.tex
\section{Introduction}

Maximum Constraint Satisfaction Problems ({\sc Max-CSP}s) are some of the most commonly studied optimization problems in theoretical computer science. An instance $\Psi(H = (V,E),\{\Pi\}_e)$ of a $r$-ary Boolean {\sc Max-CSP} is identified with a $r$-ary constraint (hyper)graph $H = (V,E)$, a constraint set $\{\Pi_e\}_e$, where for every edge $e$, $\Pi_e \subseteq \{0,1\}^r$ defines the set of ``accepting'' assignments to the edge. The overall objective of the {\sc Max-CSP} problem is to find a Boolean labeling of the vertices that maximizes the fraction of satisfied constraints. The generality of the definition allows it to express a wide array of combinatorial optimization problems such as {\sc Max-Cut}, {\sc Max-SAT}, {\sc Max-$3$-Lin}, {\sc Label Cover}, among many others as special cases; each of these problems are of fundamental interest on their own, and have their own dedicated line of works that explore its various aspects (see \cite{MM17CSP} and references therein for a comprehensive overview of such results). 

An extensively studied question in the context of each of these problems, and for {\sc Max-CSP}s in general, is that of proving tight bounds for efficiently achievable approximation guarantees. While in the beginning, individual classes of CSPs seem to require their own problem-specific techniques for algorithms and hardness, over a steady sequence of works, the shape of a more unified framework for understanding CSPs began to emerge; upper bounds were facilitated using advances in understanding the power of semidefinite programming \cite{GW94,KMS98,CMM06}, whereas lower bounds were being facilitated by dictatorship test-based gadget reductions \cite{Has01,KKMO07}). These developments eventually culminated in the work of Raghavendra~\cite{Rag08}, which connected the two directions by showing that the limitations of SDP based algorithms are fundamentally connected to the best efficiently achievable approximation factors. Specifically, \cite{Rag08} showed that assuming the Unique Games Conjecture (UGC)~\cite{Khot02a}, for every {\sc Max-CSP} there exists a canonical SDP relaxation whose integrality gap matches the best possible hardness for the problem under UGC. The soundness analysis of its reduction automatically yielded the optimal rounding schemes for the SDP relaxation, thus further strengthening the connection between SDPs and UGC-based hardness reductions. Following \cite{Rag08}, there have been several subsequent works which establish analogous connections for the settings of graph partitioning problems~\cite{MNRS08-multicut}, Ordering CSPs~\cite{GHMRC11-OrderingCSP}, Strict CSPs~\cite{KMTV11}, and Tree CSPs~\cite{CM22}.   

Despite the ubiquity of \cite{Rag08}'s framework, there are still several natural classes of optimization problems for which an analogous unified theory of optimal algorithms and hardness has remained elusive. A particularly well-studied class of such problems are CSPs with global cardinality constraints. Formally, given a predicate $\psi:\{0,1\}^r \to \{0,1\}$, a $\mu$-constrained {\sc Max-CSP}$(\psi)$ instance is a Boolean {\sc Max-CSP} where the edge constraints are identified by $\psi$, and the objective is to find an assignment of relative weight {\em exactly} $\mu$ that satisfies the maximum fraction of constraints. Such CSPs and its variants again express several well-studied problems such as {\sc Densest-$k$-Subgraph}, {\sc Max-Bisection}, {\sc Max $k$-Coverage}, {\sc Small-Set Expansion}, among many others. 

While there is extensive literature that study algorithms and hardness for these problems, these results often employ problem-specific techniques, and as such, a general unified framework of these problems, as in \cite{Rag08}, is far from being realized. In particular, in the quest to develop such a theory for globally constrained CSPs, one is faced with following immediate challenges:

{\bf Question 1}: {\em What is a natural algorithmic framework that is amenable towards deriving optimal approximation guarantees for globally constrained CSPs?} 

{\bf Question 2}: {\em Is there a way to lift lower bounds for the class of optimal algorithms to (possibly conditional) \NP-hardness lower bounds?} 

Towards answering the above, it is useful to understand why the techniques of \cite{Rag08} don't immediately apply to the setting of globally constrained CSPs. The key observation here is that in the setting of Max-CSPs, one can usually reduce the task of solving the CSP instance to that of solving a distribution over constant sized `local' instances. This viewpoint is especially useful in the context of both algorithmic and hardness frameworks. In particular, analyses of SDP based approximation algorithms often view the CSP as a distribution over constraints, and derive the approximation guarantee by reducing the task to that of analyzing the performance of the algorithm on individual constraints. On other hand, for establishing hardness results, a usual approach is to first construct a toy instance (aka a dictatorship test) that models the conceptual difficulty of the problem, and then one can create a large hard instance of the same by embedding copies of the smaller instance along the edges of a hard CSP such as Unique Games. However, back in the setting of globally constrained Max-CSPs, these local techniques often fail to capture the additional complexity introduced due to the `global' constraint(s), and yield bounds that are far from optimal. 

Fortunately, it turns out that there are more recent advances in the theory of approximation which can provide us with candidate guesses for the answers to these questions. Firstly, there is promising evidence which suggests that algorithmic frameworks based on higher-order SDP hierarchies (such as the {\em Sum-of-Squares} (SoS) hierarchy) are better suited for addressing local and global constraints simultaneously. In particular, the work of Raghavendra and Tan~\cite{RT12} proposed a systematic SoS-based framework for solving constrained CSPs, variations of which have been used to establish near-optimal approximation guarantees for several problems such as {\sc Max/Min-Bisection}, {\sc Balanced Separator}~\cite{RT12,ABG12}, {\sc Max-$k$-VC}, {\sc CC Max-Cut} ~\cite{AS19}. For the second question, a candidate staging ground for proving hardness results for globally constrained Boolean Max-CSPs is the {\em Small Set Expansion Hypothesis (SSEH)}. Introduced in the work of Raghavendra and Steurer~\cite{RS09}, the Small-Set Expansion problem has been useful in establishing tight bounds for several natural problems such as {\sc Balanced Separator}, {\sc Min-Bisection}~\cite{RST12}, {\sc Max-Biclique}, {\sc Min-$k$-Cut}~\cite{Man18}, {\sc Densest-$k$-SubHypergraph} \cite{GL22}. 

In this work, we make progress towards bridging the gap between upper and lower bounds for globally constrained CSPs. In particular, we extend the techniques of \cite{Rag08} and develop a systematic way of lifting SoS integrality gaps to SSEH hardness. We detail the contributions of this work in the remainder of this section.

\subsection{Our Results}

Let $\psi:\{0,1\}^r \to \{0,1\}$ be a $r$-ary Boolean predicate. Let $G_{\rm gap} = (V_{\rm gap},E_{\rm gap},\tilde{w},w)$ be a weighted constraint hypergraph (where the vertex and edge weights are given by $\tilde{w}$ and $w$ respectively) for the {\sc Max-CSP} problem with predicate $\psi$. We will assume that the edge weights (vertex weights) sum to one, and hence, they define a distribution over the set of edges (vertices). We will use $i \sim G_{\rm gap}$ and $e \sim E_{\rm gap}$ to denote a draw of a vertex and an edge from the corresponding distributions.

In this work, we deal with Lasserre SDP relaxations for the {\sc Max-CSP}$(\psi)$ instance on $G_{\rm gap}$. Formally, the $\ell$-round Lasserre relaxation for $G_{\rm gap}$ -- denoted by ${\sf Lass}_{\mu,\ell}(G_{\rm gap})$ -- is described in the Figure \ref{fig:lass}. The convex relaxation described is the $\ell$-round Lasserre lifting of the basic SDP. For every subset of variables $A$ of size at most $\ell$ and any partial labeling $\alpha$ of the vertices in $A$, it introduces a vector variable $v_{A,\alpha}$ which is meant to indicate whether the vertices in $A$ are assigned the label $\alpha$. Furthermore, it also introduces a local distribution $\theta_A$ over partial assignments to the vertices in $A$, for every subset of $A$ of size at most $\ell$. Finally it enforces inner product constraints which ensure that the $\ell^{th}$-order moment matrix (defined using the local distributions) is well-defined and positive semidefinite. 

\begin{figure}[ht!]
	\begin{mdframed}
		\begin{eqnarray*}
			\textnormal{Maximize} & \Ex_{e \sim E_{\rm gap}} \Pr_{X_e \sim \theta_e}\left[\psi(X_e) = 1\right] & \\
			\textnormal{Subject to} & \Ex_{i \sim G_{\rm gap}}\Pr_{X_i \sim \theta_i}\left[X_i = 1\right] = \mu &\\
			& \langle u_{A,\alpha},u_{B,\beta} \rangle = \Pr_{\theta_{A \cup B}}\left[X_A = \alpha, X_B = \beta\right] & 
			\forall~A,B \subseteq [\ell] \textnormal{such that} \\
			& & |A|,|B| \leq \ell/2, \\
			& & \alpha \in \{0,1\}^A, \beta \in \{0,1\}^B
		\end{eqnarray*}
	\end{mdframed}
	\caption{${\sf Lass}_{\mu,\ell}(G_{\rm gap})$}
	\label{fig:lass}
\end{figure} 

Our main result here is that integrality gaps for the above relaxation (with $\ell \geq \ell(\mu,r)$) can be translated to almost matching Small-Set Expansion hardness. In order to state our result, we need to formally define the notion of a gap instance.

\begin{definition}[$(\ell,\mu,c,s,\gamma)$-gap instance]					\label{defn:gap-inst}
	An $(\ell,\mu,c,s,\gamma)$-gap instance $(G_{\rm gap},\theta := \{\theta_S\}_{|S| \leq \ell})$ is characterized by a weighted constraint hypergraph $G_{\rm gap} = (V_{\rm gap},E_{\rm gap},\tilde{w},w)$ and a valid feasible solution $\theta = \{\theta_S\}_S$ to the $\ell$-round Lasserre lifting ${\sf Lass}_{\mu,\ell}(G_{\rm gap})$ which satisfies the following conditions:
	\begin{itemize}
		\item {\bf Bias Constraint} $\Ex_{i \sim G_{\rm gap}}\Pr_{X_i \sim \theta} \left[X_i = 1\right] = \mu$, where $i \sim G_{\rm gap}$ denotes the random draw of vertex according to the distribution induced by the vertex weight function $\tilde{w}$.
		\item {\bf Completeness}. The set of local distributions satisfy
		\[
		\Ex_{e \sim E_{\rm gap}}\Pr_{X_e \sim \theta_e}\left[\psi(X_e) = 1\right] = c.
		\]
		\item {\bf Robust Soundness}. The CSP $G_{\rm gap}$ satisfies ${\sf Opt}_{\mu'}(G_{\rm gap}) \leq s$ for every $\mu' \in \mu(1 \pm \sqrt{\gamma})$, where ${\sf Opt}_{\mu'}(G_{\rm gap})$ denotes the optimal $\mu'$-constrained value of the instance $G_{\rm gap}$.
	\end{itemize}
\end{definition}

In particular, a gap instance for a $\mu$-constrained {\sc Max-CSP}$(\psi)$ corresponds to an instance with $\ell$-round Lasserre SDP value $c$ -- which is witnessed by a $\ell$-round Lasserre solution $\theta$ --  and optimal $\mu$-constrained value $s$. Equipped with the above definition, we are now ready to state the main result of this work in the following theorem.

\begin{theorem}				\label{thm:main}
	The following holds assuming the SSEH. Fix a predicate $\psi:\{0,1\}^r \to \{0,1\}$. Let $\mu \in (0,1/2)$. Suppose there exists a $(\ell,\mu,c,s,\gamma)$-gap instance $(G_{\rm gap} = (V_{\rm gap},E_{\rm gap},\tilde{w},w),\theta)$ for {\sc Max-CSP}$(\psi)$ as in Definition \ref{defn:gap-inst} such that $\ell \geq 1/\gamma^8 + r$ and $\mu^{10r} \geq \gamma \geq |V_{\rm gap}|^{-1/C}$, where $C > 0$ is a large fixed constant. Then given a $\mu$-constrained {\sc Max-CSP}$(\psi)$ instance $\cH = (V_\cH,E_\cH,\tilde{w}_{\cH},w_\cH)$, it is \NP-hard to distinguish between the following two cases:
	\[
	{\bf YES~Case}:~{\sf Opt}_\mu(\cH) \gtrsim_r \frac{c}{\log(1/\mu)^2}
	\ \ \ \ \ \ \ \ \textnormal{and} \ \ \ \ \ \ \ \
	{\bf NO~Case}:~{\sf Opt}_{\mu}(\cH) \lesssim_r s.
	\]
	Here the $\lesssim_r$ and $\gtrsim_r$ notations hide multiplicative factors that depend only on $r$.
\end{theorem}

The above theorem states that an $s/c$-integrality gap for $\ell(\mu,r)$-round Lasserre relaxation can be translated to an $\Omega_r((s/c) \cdot \log(1/\mu)^2)$-factor SSEH hardness. At a high level, this implies that hard-to-round instances for globally constrained CSPs approximately translate to conditional \NP-hardness i.e., the integrality gaps for higher-order convex relaxations imply fundamental computational bottlenecks in approximating globally constrained CSPs. Whereas previous works such as~\cite{RT12,ABG12,BansalSticky20} illustrate the efficacy of Lasserre based algorithms in approximating constrained {\sc Max-CSP}s (albeit for specific cases), we provide a connection in the other direction by lifting Lasserre gaps to SSEH based hardness for every constrained {\sc Max-CSP}. Theorem \ref{thm:main} illustrates the power of Lasserre hierarchy and SSEH in understanding the approximability of problems with global constraints, where their basic counterparts, namely SDPs and UGC, seem to provide relatively limited insights. 


\begin{remark}[Lifting Degree-$2$ SoS gaps]
	An obvious interesting question here is whether basic SDP integrality gaps can be lifted to (conditional) \NP-hardness as well. We believe that an optimal gaps-to-hardness lifting theorem (which does not lose any constants) would require higher levels of SoS, since recent works on cardinality constrained problems such as Max/Min-Bisection~\cite{RT12,ABG12}, Balanced Max $2$-SAT~\cite{ABG12}, Max-$k$-VC~\cite{AS19} all crucially use properties of higher-level SoS-relaxations. However, there are individual instances of globally constrained problems such as Small-Set Expansion and Max-$k$-Coverage for which SDP~\cite{RST10a} and LP relaxations\footnote{For Max-$k$-Coverage, it is folklore that independent rounding on LP relaxations yields a $(1 - 1/e)$-approximation algorithm, which is tight~\cite{FeigeSetCover}.} (for the latter) are known to be optimal, and it is likely that the SDP gap instances of these problems themselves can be losslessly lifted to \NP-hardness. 
\end{remark}

\subsection{Related Work}

{\bf Integrality Gaps and Hardness}. The connection between SDP integrality gaps and hardness was formally first established in the work of Austrin \cite{Aus10}, who showed matching bounds for $2$-CSPs under certain assumption on the hardest-to-round distributions. Following \cite{Rag08}, several works have extended this to various settings such as Ordering CSPs~\cite{GHMRC11-OrderingCSP}, Strict CSPs~\cite{KMTV11}, graph partitioning problems~\cite{MNRS08-multicut}, and Tree CSPs~\cite{CM22}. Khot and Saket~\cite{KS15} showed that LP gaps for {\sc Max-CSP}s can be lifted to Unique Games hardness while losing a factor of $O(\log k)$, where $k$ is the label set size. There have also been several works which lift previously known integrality gaps to explicit Unique Games based lower bounds, for e.g. see \cite{GSS15,Lee17,BHPZ21}.

{\bf Lasserre Hierarchy}. The Lasserre -- aka Sum-of-Squares (SoS) -- hierarchy has been studied extensively in the context of approximation algorithms for CSPs, and is widely believed to be a candidate meta-algorithm for refuting the Unique Games Conjecture. Works such as \cite{BRS11,GS11,AJT19} exhibit bounds against the expansion profile of hard instances of Unique Games and related CSPs by showing that instances with small threshold rank can be efficiently solved using SoS. On the other hand, there have been several works~\cite{BBHKSZ12,OZ13,KOTZ14,BBKSS21} which show that SoS can efficiently refute integrality gap instances of several fundamental problems which fool the basic SDP relaxation. We refer interested readers to \cite{BS-SoSSurvey,FKT-SoSSurvey} for an overview of related results.

{\bf Globally Constrained CSPs}. There have been several lines of works that study specific globally constrained CSPs such as Densest-$k$-Subgraph problem~\cite{FS97,BCCFV10,Man17}, Max-Bisection~\cite{RT12,ABG12}, Max $k$-Vertex Coverage~\cite{RT12,AS19}. \cite{Gh22} studied the constrained variant of the Homogeneous Max-$3$-Lin problem from the context of approximation resistance and established nearly-tight bounds in several regimes. There have been several works which propose general purpose algorithmic frameworks using Lasserre hierarchy in several settings, for e.g. see \cite{RT12,GS11,BansalSticky20} and references therein.  

Of particular relevance to the current work is that of Ghoshal and Lee~\cite{GL22} who studied Biased CSPs, where the objective is to find a labeling of relative weight {\em at most} $\mu$ that satisfies the maximum number of constraints. They established tight bounds for the bias-approximation curve of every Biased CSP of constant arity by expressing it as a function of the bias-approximation curves of the Densest-$k$-SubHypergraph problems. We point out that their results do not apply to this setting due to the following reasons (i) in this work we study CSPs with bounded weights (i.e., where the vertex weights at most inverse polynomial in the instance size -- say, bounded by $1/n^{1/100}$) whereas \cite{GL22} studied CSPs where the vertex weights can be unbounded and (ii) the feasible labelings are constrained to have relative weight exactly $\mu$ as opposed to at most $\mu$. Due to these differences, there exist predicates for which the approximation curve of the ``equals'' version studied in this work is distinctly different from the curve for the version studied in \cite{GL22}. For illustration, we provide such an example in Section \ref{sec:examp}.

%% file: overview-lift.tex
\section{Overview and Techniques}

Our approach towards establishing Theorem \ref{thm:main} is based on the following principle: we want to use the ``hard-to-round'' integrality gap instance to design a dictatorship test which, when combined with an appropriate outer verifier, will yield a similar (conditional) hardness. This is a well-understood process that can be distilled into two clear objectives:

\begin{itemize}
	\item[(i)] Given a $(c,s)$-integrality gap instance, construct a ``bias-dependent'' $(c,s)$-dictatorship test\footnote{A {\em bias-dependent} dictatorship test is a dictatorship test where the completeness-soundness guarantees of the test only apply when the input long code table $f$ has fixed relative weight, see Figure \ref{fig:test-prop} for a more formal description of these properties.} $\cT_{\rm dict}$ for the same predicate.
	\item[(ii)] Compose $\cT_{\rm dict}$ with \smallsetexpansion~as the outer verifier.
\end{itemize}

There have been several works which use the above framework (for {\sc Max-CSP}s) to establish (often tight) connections between convex programming relaxations and UGC based hardness. At a high level, Theorem \ref{thm:main} is based on \cite{Rag08} (and in part, \cite{RT12}) -- however extending their framework to the setting of globally constrained CSPs is challenging and will require several new ideas. 

\subsection{\cite{Rag08}'s approach}

We first give a brief account of \cite{Rag08}'s framework, which will be useful towards highlighting the key bottlenecks that need to be addressed in the setting of constrained CSPs. As mentioned above, \cite{Rag08} also reduces the task into establishing (i) and (ii). Towards establishing (i), \cite{Rag08} considers the following natural dictatorship test.\footnote{To keep the description simple and consistent, here we describe \cite{Rag08}'s construction for the setting of Boolean CSPs -- however, their setting and results are far more general, we refer interested readers to \cite{Rag08} for details.}

\begin{figure}[ht!]
	\begin{mdframed}
		{\bf Setting}. Let $(G_{\rm gap} = (V_{\rm gap},E_{\rm gap},w),\theta)$ be a $(c,s)$-gap instance\footnote{A $(c,s)$-gap instance refers to a gap instance for the SDP relaxation corresponding to {\sc Max-CSP}$(\psi)$.} for {\sc Max-CSP}$(\psi)$.\\
		{\bf Input}. An assignment\footnote{Typically, $r$ and $|V_{\rm gap}|$ are treated as problem-specific constants, and then $R$ is chosen large enough as a function of $r$ and $|V_{\rm gap}|$.} $f:\{0,1\}^R \to \{0,1\}$. \\
		{\bf Test}. 
		\begin{enumerate}
			\item Sample edge $e \sim E_{\rm gap}$ according to the distribution given by the weight function $w$.
			\item For each $j \in [R]$, independently sample $(x_1(j),\ldots,x_r(j)) \sim \theta_e$.
			\item For each $i \in e$, sample $x'_i \underset{1 - \eta}{\sim}  x_i$ independently.
			\item Accept if and only if
			\[
			\psi\Big(f(x'_1),\ldots,f(x'_r)\Big) = 1.
			\]
		\end{enumerate}
	\end{mdframed}
\caption{Test $\cT_{\rm UGC}$}
\label{fig:test-rag}
\end{figure}

Then they proceed to show that the above test has completeness $c$ and soundness $s$. As is usual, arguing completeness is simple -- one can show that any dictator function $f = \chi_i$ passes the test with probability at least $c - o(1)$. On the other hand, arguing that the soundness of the test is at most $s := {\sf Opt}(G_{\rm gap})$ is the more challenging direction: this is established by providing a rounding algorithm ${\sf Round}$
which shows that non-influential functions can be used to round off solutions with value matching the acceptance probability -- this is stated formally in the following theorem.
\begin{theorem}[\cite{Rag08}]					\label{thm:sound-inf}
	There exists a randomized algorithm  ${\sf Round}$ with the following property. Given a non-influential long code assignment $f$, it can use $f$ to round off a labeling $\sigma:={\sf Round}(f)$ such that 
	\[
	\Ex_{\sigma}\left[{\sf Val}_\sigma(G_{\rm gap})\right] \approx \Pr\left[\cT_{\rm UGC} \mbox{ Accepts } f\right],
	\]
	where ${\sf Val}_{\sigma}(G_{\rm gap})$ denotes the weight of edges in $G_{\rm gap}$ satisfied by the labeling $\sigma$.	
\end{theorem}

In particular, since ${\sf Val}_{\sigma}(G_{\rm gap})$ is always upper bounded by ${\sf Opt}(G_{\rm gap})$, the above implies that  
\[
\Pr\left[\cT_{\rm UGC} \mbox{ accepts } f\right] \approx \Ex_{\sigma}\left[{\sf Val}_\sigma(G_{\rm gap}) \right] \leq {\sf Opt}(G_{\rm gap}) = s,
\] 
thus establishing soundness, and consequently (i). Now given the dictatorship test from (i), using \uniquegames~instances as the outer verifier, the composition step (ii) follows almost immediately using standard techniques~\cite{KKMO07}.

{\bf \cite{Rag08} for $\mu$-constrained CSPs?}. In extending the above to the setting of $\mu$-constrained CSPs, several immediate issues arise, some of which have been addressed in previous works. To begin with, a first step is to understand what kind of convex relaxations can be used to construct dictatorship tests for $\mu$-constrained CSPs? Since the local distributions completely specify the dictatorship test, the above is akin to asking what kind of properties of local distributions would allow one to design algorithms that can round-off non-influential functions to {\em globally-feasible} integral solutions. Towards this, Raghavendra and Tan~\cite{RT12} showed that local distributions with small average covariance suffice for this purpose; in particular they showed that given a gap instance for a $\mu$-constrained CSP instance with completeness and soundness parameters $c$ and $s$, such that the local distributions witnessing $c$ have average covariance $o(1)$, there exists a bias-dependent dictatorship test $\cT_\mu$ (on some domain $\Omega^R$) with the following properties (Figure \ref{fig:test-prop}):
\begin{figure}[ht!]
	\begin{mdframed} 			
\begin{itemize}			
	\item {\bf Completeness}. If $f:\Omega^R \to [0,1]$ is a dictator, then $f$ passes the test with probability at least $c - o(1)$, and the relative weight of $f$ w.r.t. the distribution over the query indices in $\cT_\mu$ is $\mu$.
	\item {\bf Soundness} If $f:\Omega^R \to [0,1]$ has no influential coordinates, and has relative weight $\mu$ under the test distribution $\cT_{\mu}$, then the test accepts with probability at most $s + o(1)$.
\end{itemize}
	\end{mdframed}
	\caption{Bias Dependent Dictatorship Test}
	\label{fig:test-prop}
\end{figure}

We point out to the readers the construction of $\cT_\mu$ using the local distributions itself is identical to that of \cite{Rag08}. On the other hand, its analysis requires additional work -- in particular, in order to establish an analogue of Theorem \ref{thm:sound-inf}, \cite{RT12} uses the small average-covariance guarantee to show that the rounded solution would have relative weight close to $\mu$ with high probability, thus allowing the soundness analysis to relate the acceptance probability to the optimal value of the $\mu$-constrained CSP. 

Another key question here is to understand what choice of outer verifier would be amenable to lifting local $\mu$-constrained dictatorship tests to (conditional) \NP-hardness. As discussed before, for the setting of globally constrained CSPs, we require outer verifiers with stronger mixing properties. Specifically, we would want the mixing properties to ensure that any globally feasible long code table is also locally feasible for most instantiations of the local tests, so that one cannot create globally feasible assignments in the reduction that can cheat in a significant fraction of the local tests. 

\begin{figure}[ht!]
	\begin{mdframed}
		\begin{center}
			{\bf \underline{ Mixing Properties for Bias Dependent Tests} }
		\end{center} 
		\vspace{0.5cm}
		More formally, suppose we want to compose the dictatorship test with an outer verifier $\Phi$. Then in the full reduction, the space of assignments to the instance (say we denote it by $\cH$) usually consists of $f:= \{f_A\}_{A \in \Phi}$, where each $f_A:\Omega^R \to \{0,1\}$ is the local assignment corresponding to a vertex $A \in \Phi$. Now, as is usual in the dictatorship test style reductions, one can reduce the analysis of the full reduction to that of analyzing the local tests, i.e.,
		\[
		{\sf Val}_\cH(f) :=  \Pr_A\Pr_{B \sim M(A)}\left[\cT_\mu~\textnormal{ accepts }~f_B\right],
		\]
		where $M$ is a suitable averaging operator on the space of vertices of $\Phi$. Now the main challenge here is to ensure that for most choices of $A$, the averaged local assignment $\tilde{f}_B = \Ex_{B \sim M(A)} f_B$ satisfies $\Ex_x \tilde{f}_A(x) \approx \Ex_A\Ex_x \tilde{f}_A(x) = \mu$, so that we can leverage the completeness and soundness guarantees of the test (Figure \ref{fig:test-prop}) to argue the completeness and soundness of the full reduction.
	\end{mdframed}
	\caption{}
	\label{fig:mix}
\end{figure}

While several constructions exist that provide various trade-offs between such mixing properties and the PCP sizes (e.g., Mixing Label Cover\cite{HK04}, Quasi-random PCP~\cite{Khot06}, Birthday Repetition~\cite{MR17}), a relatively convenient choice is the \smallsetexpansion~problem, which is more well-suited for dictatorship test-based reductions for deriving conditional \NP-hardness results.

At this point, while it may appear that one has all the necessary ingredients for proving Theorem \ref{thm:main}, it turns out that one still needs to address several key issues which we outline below:  

{\bf Composition with SSE}. Unfortunately, it turns out that one of the simplest steps in \cite{Rag08} i.e., {\em composition}, is also the trickiest step in this framework due to our choice of \smallsetexpansion~as the outer verifier. The reason for this is, unlike \uniquegames, there is an absence of generic techniques that can compose dictatorship tests for CSPs with \smallsetexpansion~as is in a black-box way. This is evident in that only a handful of previous works~\cite{RST12,LRV13,Man17,GL22} have successfully used dictatorship-test-gadget-based reductions for showing SSEH-based hardness results and in particular, heavily rely on the techniques of \cite{RST12} to derive the mixing properties required for the soundness analysis of the full reduction. 

{\bf Variance Blow-up in Soundness Analysis}. A more subtle issue lies in the analysis of \cite{RT12}'s analogue of Theorem \ref{thm:sound-inf}; in particular, it can only guarantee that if the average covariance of the local distribution is $\eta$, then the rounded solution has relative weight bounded in $\mu ( 1 \pm 2^{O(R)}\eta)$, where $R$ is the dimension of the cube corresponding to the long code table. Since the cube-dimension of long code typically depends inversely on the volume parameter of the \smallsetexpansion~instance, which is often required to be an extremely small constant as a function of the other reduction parameters, the weight bound becomes unusable in the context of the reduction.

Handling the above issues are the key contributions of this work, we expand on these issues and our techniques for handling them in the remaining sections.

\subsection{Understanding Composition with SSE}

For the purpose of illustration, we will describe a ``SSE composable'' test from a gap instance for the {\sc Densest-$k$-SubHypergraph} ($\dksh$) problem. Recall that in the $\dksh_r$~problem, we are given a hypergraph $H = (V,E,w)$ of arity $r$, and the objective is to find a subset of $k$ vertices that induces the maximum weight of hyperedges. Equivalently, in the terminology used in this work, we can phrase it as a $\mu = k/n$-constrained {\sc Max-CSP} instance with the ${\sf AND}$ predicate of arity $r$. Our first attempt at a test for the $\dksh_r$~problem will be an adaptation of the test in Figure \ref{fig:test-rag} to the setting of \cite{RST12}. The key difference between the test from Figure \ref{fig:test-0} and the test employed in our reduction is the following: similar to \cite{RST12}, in order to break the local gadget structure and facilitate mixing, we will employ a ``folding'' operator that acts in a lifted space. Specifically, instead of defining our long codes to be on the intended probability space $L_2((\Omega^R,\gamma^R))$ (for e.g., $(\Omega,\gamma) = \{0,1\}_\mu$ in Figure \ref{fig:test-rag}), we will define the long codes to be on the lifted space $L_2((\Omega^R,\gamma^R) \otimes \{\bot,\top\}^R_\beta)$, where $\{\bot,\top\}_\beta$ is the distribution on $\{\bot,\top\}$ which assigns measure $\beta$ on `$\top$.' This lifted space allows one to define a family of stochastic folding operators $\{M_z\}_{z \in \{\bot,\top\}^R}$, which are defined as follows. 

\begin{definition}[Noise operator with leakage~\cite{RST12}]
	For any probability space $L_2(\Omega^R,\gamma^R)$, and $z \in \{\bot,\top\}^R$, we define the stochastic functional $M_z$ as follows. For any $x \in \Omega^R$, we sample $x' \sim M_z(x)$ as follows. Do the following for every $j \in [R]$:
	\begin{itemize}
		\item If $z(j) = \top$, then set $x'(j) = x(j)$.
		\item If $z(j) = \bot$, then sample $x'(j) \sim \gamma$ independently.
	\end{itemize}
\end{definition}

The above noise operator has several interesting properties that are useful for breaking the local gadget structure of the reduction. For instance, observe that for a fixed choice of $z \in \{\bot,\top\}^R$, the operator $M_z$ folds the probability space $(\Omega^R,\gamma^R)$ along the coordinates where $z(i) = \bot$. Hence, for a randomly sampled $z \sim \{\bot,\top\}^R_\beta$, the corresponding noise operator $M_z$ behaves like the $\beta$-correlated noise operator $T_\beta$ on $(\Omega^R,\gamma^R)$ (which in turn ensures the mixing properties (Fig. \ref{fig:mix}). Then following \cite{RST12}, we incorporate these operators into \cite{Rag08}'s test which yields the following test (described in Figure \ref{fig:test-0}): 

\begin{figure}[ht!]
	\begin{mdframed}
		Let $(G_{\rm gap},\theta,\tilde{w},w)$ be $(c,s,\ell,\mu,\gamma)$ gap instance for $\dksh$. \\ 
		{\bf Input}. Long code $f:\{0,1\}^R \times \{\bot,\top\}^R \to \{0,1\}$ satisfying
		\begin{equation}				\label{eqn:bias-0}
			\Ex_{i \sim G_{\rm gap}}\Ex_{x \sim \{0,1\}^R_{\mu_i}}\Ex_{z \sim \{\bot,\top\}^R_\beta}\left[f(x,z)\right] = \mu.
		\end{equation}
		
		{\bf Setup}. For every edge $e$, and $i \in e$, $x_i$ and $z_i$ are $\{0,1\}^R$ and $\{\bot,\top\}^R$-valued vector random variables. Furthermore, let $\cD^R_e$ be the joint distribution (to be determined later) on variables $(x_i,z_i)_{i \in e}$.	 \\
		{\bf Test}.
		\begin{enumerate}
			\item Sample edge $e \sim E$.
			\item {\bf Long code Step}.
			\begin{itemize}
				\item For every $j \in [R]$, sample $(x_i(j),z_i(j))_{i \in e}$ from the joint distribution $\cD^R_e$. 
				\item For every $i \in e$, sample $(1 - \eta)$-correlated copies $(\tilde{x}_i,z'_i) \underset{1 - \eta}{\sim} (x_i,z_i)$. 
			\end{itemize}
			\item {\bf Folding Step}.
			\begin{itemize}
				\item For every $i \in [r]$, sample $x'_i \sim M_{z'_i}(\tilde{x}_i)$.
			\end{itemize}
			\item Accept if and only if
			\[
			f(x'_i,z'_i) = 1 ~~~~~ \forall i \in [R].
			\]
		\end{enumerate}
	\end{mdframed}
	\caption{Test Framework for SSEH Reduction}
	\label{fig:test-0}
\end{figure}

Informally, the test in Figure \ref{fig:test-0} first samples $(x_i,z_i)_{i \in e}$ from a distribution which is intended to enforce the checks corresponding to the test in Figure \ref{fig:test-rag}, following which it re-randomizes, and then folds each $x_i$ vector variables via the $M_{z_i}$ operator\footnote{In the actual reduction, the folding step also folds along the space of vertices of the outer verifiers (See Footnote \ref{foot:mix}).}. Clearly, the main factor determining the properties of this test is the family of joint distributions $\{\cD^R_e\}_{e \in E}$ on the $(x_i,z_i)$ variables, which has to be designed carefully to ensure several properties which we briefly describe below. 

The key issue here is that while the choice of the distribution of the $x_i$ variables is clear, apriori, the distribution of the $z_i$ variables is not immediate as these are auxiliary variables that are introduced to ensure the mixing properties required for the reduction (Figure \ref{fig:mix}), and as such, are not immediately relatable to the SDP solution of the gap instance. Ideally, we would want to introduce the $z_i$ variables to the test in a way such that they facilitate the mixing properties, and then we would like to carry out the completeness and soundness analysis of the test as before just using the properties of the distribution of the $x_i$ variables. Unfortunately, this ``modular'' scenario is somewhat of a pipe dream as the $z_i$ variables end up affecting the completeness and soundness properties in unexpectedly non-trivial ways. To illustrate this, let us try out a couple of elementary approaches and see why they fail. Fix an edge $e = [r]$, and let $z_1,\ldots,z_r$ be the corresponding variables from the test distribution conditioned on the fixing of $e$. 
\begin{itemize}
	\item Suppose $z_1,\ldots,z_r$ are sampled completely independently.  Then because of the folding step, the resulting $(\tilde{x}_i,{z}'_i)$-variables will be almost independent for different choices of $i \in [r]$. In that case, dictator functions would no longer be able to exploit the correlations along a single coordinate, thus resulting in poor completeness parameters.
	\item On the other hand, let us consider the case where $z_i$'s are identical random variables (i.e., completely correlated). Then it is possible to construct long code assignments which depend only on the $z$-variable component (i.e., $f(x,z) = h(z)$ for some function $h$) that does strictly better than the intended soundness of the basic test, which actually relies on the correlation structure of the $(x_i)_i$-variables. 
\end{itemize}  

Another bottleneck is that the soundness analysis of \cite{Rag08,RT12} -- and Theorem \ref{thm:sound-inf} in particular -- crucially relies on the fact that under the test distribution, the ensemble of $(x_i)_i$ variables satisfy the following property: the covariance structure of the variables along any coordinate $j \in [R]$ is identical to that of the corresponding local variables $(X_i)_i$ under the SDP solution $\theta$. On the other hand, the inclusion of the additional $(z_i)_i$ variables in our setting makes the correlation structure of the resulting set of variables $(x_i,z_i)_{i \in V_{\rm gap}}$ incomparable of correlation structure of the vector solution, and rules out the possibility of carrying out \cite{Rag08}'s soundness analysis in the lifted space as is.

\subsection{Our Approach: Weak Coupling of $z$ Variables}

Our approach towards addressing the above is that we devise a way of sampling the $(x_i,z_i)$ variables in a {\em  coupled way} such that the correlation structure of the $z_i$ variables does not overwhelm the $x_i$ variables, while still being correlated enough to guarantee that (i) the test has the required completeness/soundness properties (up to some loss) under the resulting distribution (Figure \ref{fig:test-prop}), and (ii) the corresponding folding operators $\{M_z\}_z$ have the desired mixing properties (Figure \ref{fig:mix}). In particular, our construction of the joint distribution will have the following key property: if $f$ is non-influential, then we can {\em decouple} the $(z_i)_i$ variables and ``average them out'' so that the soundness analysis reduces to the setting where the resulting function is just dependent on the $x_i$ variables, thus enabling the use of \cite{Rag08,RT12}'s rounding argument to conclude the soundness analysis. 

Towards stating our idea more formally, let us begin by defining the collection of distributions $(\cD^R_e)_{e \in E_{\rm gap}}$:
\begin{definition}[Distribution $\cD^R_e$]							\label{defn:dist}
For any edge $e = (i_1,\ldots,i_r) \in E_{\rm gap}$, let $\cD^R_e$ be the following joint distribution over variables $(x_i,z_i)_{i \in e}$. Firstly, $(x_{i_1},\ldots,x_{i_r})$ are jointly distributed $R$-dimensional variables such that $(x_{i_1}(j),\ldots,x_{i_r}(j)) \sim \theta_e$ for every $j \in [R]$ (as in Figure \ref{fig:test-rag}). Furthermore, $(z_1,\ldots,z_r)$ are $r$ vector-valued random variables of dimension $R$, such that for every coordinate $j \in [R]$, the variables $(z_i(j))_{i \in e}$ is distributed (independent of $(x_i)_{i \in e}$) as follows. Sample a common assignment $z = (z(j))_{j \in [R]} \sim \{\bot,\top\}^R_\beta$, and then do the following for every $j \in [R]$ independently:
\begin{itemize}
	\item W.p. $\rho$, set every variable $(z_i(j))_{i \in e}$ to $z(j)$.
	\item W.p. $1 - \rho$, we sample $z_i(j) \sim \{\bot,\top\}_\beta$ independently for each $i \in e$.
\end{itemize}
\end{definition}

The following lemma (stated for $e = [r]$) says that if $\rho$ is small enough as a function of $\mu$ and $r$, and the functions $f_{1}(x_{1},z_{1}),\ldots,f_{r}(x_{r},z_{r})$ have small influences, then under the above distribution, the product of the functions $\prod_{i \in e} f_i(x_i,z_i)$ would behave near-identically to the distribution where $z_i$'s are fully independent. 

\begin{lemma}[Informal version of Lemma \ref{lem:prod}]				\label{lem:prod-inf}
	Let $r \geq 2$ and $\mu \in (0,1)$ be small enough. Let $f_1,\ldots,f_r$ be functions, such that $f_i \in L_2((\{0,1\}_{\mu_i} \otimes \{\bot,\top\}_\beta)^R)$ has its influences bounded by $\tau(\mu,r)$, where $\tau(\mu,r)$ depends only on $\mu$ and $r$. Furthermore, suppose $\rho \leq 1/(4r^2\log (1/\mu)^2)$. Then,
	\[
	\Ex_{(x_i,z_i)_{i \in [r]} \sim \cD^R_e}\left[\prod_{i \in [r]}f_i(x_i,z_i)\right] \lesssim_r \Ex_{(x_i)_{i \in [r]} \sim \theta^R_e}
	\left[\prod_{i \in [r]}\Ex_{z_i \sim \{\bot,\top\}^R_\beta}f_i(x_i,z_i)\right] + \mu^r,
	\]
	where $\lesssim_r$ hides multiplicative factors which depend only on $r$.
\end{lemma}

The above lemma is a key component of our soundness analysis, and we sketch a proof of the lemma later in this section. It implies that at the cost of making $(z_i)_i$ variables slightly correlated (due to which we get slightly weaker completeness), we can average them out completely when the functions $f_1,\ldots,f_r$ all have small influences. We point out that the (weaker) observation that products low-degree functions on slightly correlated spaces behave almost similarly to the independent setting is not new in the hardness literature, and has been used in several works~\cite{MNT16,KS15,GL22}. However, the above half-decoupling version (which retains the correlation structure on the $x_i$ variables, while making the slightly correlated $z_i$ variables completely independent) is new to the literature to the best of our knowledge, and requires more ideas. 

The above immediately provides a way of stitching together the $(x_i)_i$ and $(z_i)_i$ variables to derive the test distribution over the lifted space $L_2(\Omega^R)$. Our final test is the test from Figure \ref{fig:test-0} where for every $e \in E_{\rm gap}$, the joint distribution over $(x_i,z_i)_{i \in e}$ variables is the distribution $\cD^R_e$ from Definition \ref{defn:dist}. We now briefly analyze this test, and then conclude this part by providing a proof sketch of Lemma \ref{lem:prod-inf}. 

{\bf Completeness}. Due to the introduction of the $z_i$-variables, the completeness analysis is relatively more involved in SSE based reductions. Here, the basic idea is to use the $z$-vector to identify the choice of the dictator function. Let $j^*:\{\bot,\top\}^R \to [R]$ be the following map: For $z \in \{\bot,\top\}^R$, if there exists a unique index $j' \in [R]$ such that $z(j') = \top$, assign $j^*(z) = j'$, otherwise assign $j$ arbitrarily. Finally, for any $z$, we define the map\footnote{Here we denote $f_z(x) = f(x,z)$ for every $x,z$.} $f_z(x) = x({j^*(z)})$. 

To analyze this assignment, let us consider the test distribution for fixed edge $\ce = [r]$. Our first step is to observe that setting $R = 1/r\beta$, with probability at least $e^{-1}/r$ over the draw of $z \sim \{\bot,\top\}^R_{\beta}$, we will have that there exists a unique index $j' \in [R]$ such that $z(j') = \top$. Furthermore, conditioned on this, we can argue that with probability at least $\rho(1 - (1 - \rho)r\beta)^{1/r\beta} \gtrsim \rho \cdot e^{-1}$, we have that $z_1,\ldots,z_r$ also have $j'$ as the unique index such that $z_i(j') = \top$. In other words, we will have $j^*(z_i) = j'$ for every $i \in [r]$. This in turn implies that $f_{z_1},\ldots,f_{z_r} = x(j')$ i.e., they will be identical dictator functions. Note that since $(x_i)_i$ and $(z_i)_i$ variables are independent under the test distribution, conditioning on the above events does not affect the distribution of $(x_i)_i$. Therefore, conditioning on the above, we can bound the probability of the test accepting as
\[
\Pr_{(x_i)_{i \in [r]}}\left[\forall_{i}~f_{z_i}(x_i) = 1 \right] = \Pr_{(x_i)_{i \in [r]}}\left[\forall_i~x_i(j') = 1\right] = p_\ce,
\] 
where $p_\ce = \Pr_{X_\ce \sim \theta_\ce} \left[X_\ce = 1_\ce\right]$ is the probability of the $\theta$ satisfying edge $\ce$. Therefore, for a fixed edge $\ce$, the test accepts the assignment with probability at least $(1/er) \rho \cdot e^{-1} \cdot p_\ce$. Averaging over the choice of $\ce$, we have that the overall the test accepts with probability $\Omega(\rho \cdot c/r)$. Furthermore, since $f_z$ is a dictator function for every choice of $z$, overall $f$ is a feasible assignment\footnote{We point out that the factor $\Omega(\rho)$ loss in the completeness of the reduction is the main reason behind the SSEH hardness losing an additional factor of $\Omega_r(\log(1/\mu)^2)$.}. 

{\bf Soundness}. For the soundness analysis, as is usual, let $f:\{0,1\}^R \times \{\bot,\top\}^R \to \{0,1\}$ be a feasible assignment i.e, it satisfies the weight constraint
\[
\Ex_{i \sim G_{\rm gap}}\Ex_{x \sim \{0,1\}^R_{\mu_i}}\Ex_{z \sim \{\bot,\top\}^R_\beta}\left[f(x,z)\right] = \mu.
\]
Furthermore, for the soundness analysis, we will assume that $f$ has no influential coordinates\footnote{In the actual argument, for any fixed edge $e \in E_{\rm gap}$ the small influences condition needs to be defined with respect to probability space $(\{0,1\}_{\mu_i} \otimes \{\bot,\top\}_{\beta})$ for every vertex $i \in e$. Using influence decoding arguments, we will then be able to show that this happens for most choices of $e$ in the soundness analysis of the reduction.}. Firstly, using the observation that the dictatorship test for $\dksh$ simply performs {\sf AND} checks, we can arithmetize the probability of the test accepting as:
\[
\Pr\left[\mbox{Test Accepts }\right] = \Ex_{e}\Ex_{(x'_i,z'_i)_{i \in e}}\left[\prod_{i \in e}f(x'_i,z'_i)\right].
\]
Next, by averaging over the action of $M_{z}$ operator and the $(1 - \eta)$-correlated noise operator, i.e., $\tilde{f}(x,z) = \Ex_{(\tilde{x}_i,z'_i) \sim \sT_{1 - \eta}(x,z)}\Ex_{x'_i \sim M_{z_i}(x_i)}f(x'_i,z'_i)$, we may further re-write the above RHS as:
\begin{equation}				\label{eqn:mix-inf}
\Ex_{e}\Ex_{(x'_i,z'_i)_{i \in e}}\left[\prod_{i \in e}f(x'_i,z'_i)\right]
= \Ex_{e}\Ex_{(x_i,z_i)_{i \in e}\sim \cD^R_e}\left[\prod_{i \in e}\tilde{f}(x_i,z_i)\right].
\end{equation}
We point out that while the above averaging step might seem superfluous in the context of the (local) analysis of the dictatorship test on a single function, it is a crucial step in the actual reduction where it leverages the mixing properties of the $M_z$ operator to ensure that the globally feasible long code table is also locally feasible for most local instantiations of the test\footnote{\label{foot:mix}In particular, in the actual reduction, the full assignment consists of assignments $\{f_A\}_{A}:\{0,1\}^R \times \{\bot,\top\}^R \to \{0,1\}$ where $f_A$ is the local assignment for a vertex $A$ in ($R$-powered) SSE instance. The folding step \eqref{eqn:mix-inf} along with the spectral properties of the $\{M_z\}$ operators would ensure that $\Ex_{(x,z)} \tilde{f}_A(x,z) \approx \Ex_A \Ex_{x,z} \tilde{f}_A(x,z)$ for most choices of $A$. We point the readers to Lemma \ref{lem:conc} for a more formal statement of this guarantee.}.
  
Now assuming $\tilde{f}$ has no large influences, using Lemma \ref{lem:prod-inf}, we can decouple the $z_i$ variables and average them out: 
\begin{align*}
\Ex_{e}\Ex_{(x_i,z_i)_{i \in e}\sim \cD^R_e}\left[\prod_{i \in e}\tilde{f}(x_i,z_i)\right]
&\lesssim_r \Ex_{e}\Ex_{(x_i)_{i \in e} \sim  \theta^R_e}\left[\prod_{i \in e} \Ex_{z_i \sim \{\bot,\top\}^R_\beta}\tilde{f}(x_i,z_i)\right] + o(1) \\
&= \Ex_{e}\Ex_{(x_i)_{i \in e} \sim  \theta^R_e}\left[\prod_{i \in e} \bar{f}(x_i)\right] + o(1), 
\end{align*}
where $\bar{f}(x) = \Ex_{z \sim \{\bot,\top\}^R_\beta}\left[\tilde{f}(x,z)\right]$. Finally, note that the expression in the above RHS precisely corresponds to accepting probability of the test $\cT_{\rm UGC}$ (from Figure \ref{fig:test-rag}) w.r.t. assignment $\bar{f}(\cdot) = \Ex_z \tilde{f}(\cdot,z)$. Therefore, again assuming that $\bar{f}$ has small influences\footnote{Note that now the influences of $\bar{f}$ have to be measured w.r.t. the probability space $\{0,1\}^R_{\mu_i}$.}, we can invoke \cite{RT12}'s soundness analysis to show that 
\[
\Ex_{e}\Ex_{(x_i)_{i \in e} \sim  \theta^R_e}\left[\prod_{i \in e} \bar{f}(x_i)\right] \lesssim {\sf Opt}_\mu(G),
\] 
which concludes the soundness analysis.

\begin{remark}
	The key difference between the distribution employed by the test in this work and those of \cite{RST12} is the following: here, the $(x_i)_i$ and $(z_i)_i$ variables are sampled independent from each other, whereas in \cite{RST12} the variables $(x_{i_1},z_{i_1}),\ldots,(x_{i_r},z_{i_r})$ are sampled as independent $\rho$-correlated copies of some common assignment $(x,z)$, and in particular, $(x_i)_i$ and $(z_i)_i$ are not independent. This is crucially used in their analysis as they can use {\em Noise Stability}~\cite{MOO10} based arguments in the joint space of $(x_i,z_i)$ variables to argue soundness in one shot. However, executing a similar correlated-sampling based composition for arbitrary bias dependent tests is known to be challenging and might require newer tools for composition~\cite{RT12}. 
	
	On the other hand, 	the independence of $(x_i)_i$ and $(z_i)_i$ variables, and the distribution on the $(z_i)_i$ variables under our test distribution allow us to de-couple and average out the effect of $z_i$ variables in the soundness analysis, thus effectively letting us run our soundness analysis in the space of $(x_i)_i$ variables.
\end{remark}

\subsection{Establishing Lemma \ref{lem:prod-inf}}

We sketch a proof of Lemma \ref{lem:prod-inf} here, since it is the key technical result that drives our soundness analysis. Recall that in the setting of the lemma, we are given $R$-variate functions $f_1,\ldots,f_r:\Omega^R \to [0,1]$, where $\Omega = \{0,1\} \times \{\bot,\top\}$, which satisfy the condition
\[
\max_{j \in [R]} {\sf Inf}_{(x_i,z_i)_j}\left[f_i\right] \leq \tau(\mu,r)
\]
for every $i \in [r]$. Our first step is to show that we can transfer the above small-influences condition in the $R$-variate space (w.r.t. variables $(x_i,z_i)_1,\ldots,(x_i,z_i)_R)$ to a small-influences condition in the $2R$-variate space (w.r.t. variables $(x_i(1),\ldots,x_i(R),z_i(1),\ldots,z_i(R))$, and show that\footnote{We point out that we can only assume that the function has small influences in the $R$-variate space, since directly assuming the small influences condition in the $2R$-variate space is insufficient for the influence decoding argument.} 
\[
\max\Big\{{\sf Inf}_{x_i(j)}\left[f_i\right], {\sf Inf}_{z_i(j)}\left[f_i\right]\Big\} \leq {\sf Inf}_{(x_i,z_i)_j}\left[f_i\right] \leq \tau.
\]  	
The above follows from the observation that since $f_i$'s are defined on the probability space $\Omega^R_i := (\{0,1\}_\mu\otimes \{\bot, \top \}_\beta)^R$, we can use the Fourier characters for $\{0,1\}_{\mu_i}$ and $\{\bot,\top\}_\beta$ to derive a basis for $\Omega^R_i$ (see Claim \ref{cl:inf-tr} for the formal proof of the above inequality). Then, using the Fourier expansion of $f_i$ in $2R$-dimensional space, we can write it as a multi-linear polynomial:
\begin{equation}			\label{eqn:f-form}
f_i(x_i,z_i) = \sum_{S,T \subseteq [R]}\wh{f_i}(S,T)\prod_{j \in S}\phi^{(\mu_i)}(x_i(j))\prod_{j' \in T} \phi^{(\beta)}(z_i(j')),
\end{equation}
where $\phi^{(\mu_i)}$ and $\phi^{(\beta)}$ are the non-trivial Fourier characters in the $\mu_i$ and $\beta$-biased Boolean spaces. Then, as is standard, we extend $f_i$ to a $2R$-variate multi-linear polynomial $H_i:\mathbbm{R}^{2R} \to \mathbbm{R}$ as 
\[
H_i(\cW^x[i],\cW^z[i]) = \sum_{S,T \subseteq [R]} \wh{h_i}(S,T) \prod_{j \in S} \cW^x_j[i] \prod_{j' \in T} \cW^z_{j'}[i],
\]
where under the distribution, $\cW^x_j[i] := \phi^{(\mu_i)}(x_i(j))$ and $\cW^z_j[i]:= \phi^{(\beta)}(z_i(j))$. Next, since the functions $f_1,\ldots,f_r$ (and consequently the polynomials $H_1,\ldots,H_r$) have small influences, using the Invariance principle (Theorem \ref{thm:clt}), we can move the analysis from the Boolean space to the Gaussian space and show that 
\begin{equation}				\label{eqn:inf-val}
\Ex_{(x_i,z_i)_{i \in [r]}}\left[\prod_{i \in [r]}f_i(x_i,z_i)\right] = \Ex_{\cW}\left[\prod_{i \in [r]}H_i(\cW^x[i],\cW^z[i])\right]
\approx \Ex_{\cG}\left[\prod_{i \in [r]} H_i(g_{x,i},g_{z,i})\right]
\end{equation}
such that $\cG := (g_{x,i},g_{z,i})_{i \in [r]}$ is a collection of $r \times R$ jointly distributed Gaussian random variables which matches the covariance structure of $\cW$. In particular, the matching covariance structure will imply the following properties:
\begin{itemize}
	\item For every $i \in [r]$, $g_{x,i}$ and $g_{z,i}$ are marginally distributed as $R$-dimensional standard Gaussian vectors.
	\item $(g_{x,i})_{i \in [r]}$ and $(g_{z,i})_{i \in [r]}$ are independent from each other -- this is where we use the fact that under the test distribution $\cD^R_{[r]}$, the $(x_i)_i$ and $(z_i)_i$ variables are independent of each other. 
	\item Furthermore, the covariance structure of the $(z_i)_{i \in [r]}$ variables implies that $g_{z,1},\ldots,g_{z,r}$ will be distributed as $r$ independent $\sqrt{\rho}$-correlated copies of a standard Gaussian vector $g \sim N(0_R,I_R)$.
\end{itemize} 

The above properties allow us to re-write the RHS of \eqref{eqn:inf-val} as:
\begin{equation}			\label{eqn:sound}
\Ex_{\cG}\left[\prod_{i \in [r]} H_i(g_{x,i},g_{z,i})\right]
= \Ex_{(g_{x,i})_{i \in [r]}}\Ex_{g_{z,1},\ldots,g_{z,r} \underset{\sqrt{\rho}}{\sim} g} \left[\prod_{i \in [r]}{H}_i(g_{x,i},g_{z,i})\right].
\end{equation}

{\bf Averaging out $(g_{z,i})^r_{i = 1}$}. The next step is the crucial part of the proof. Consider for any fixing of $(g_{x,i})_{i \in [r]}$; note that this does not affect the distribution of $g_{z,1},\ldots,g_{z,r}$. Then using the Gaussian rearrangement from the multi-dimensional {\em Borell's Isoperimetric Inequality} (Theorem \ref{thm:egt}), we can upper bound the inner expectation from \eqref{eqn:sound} as  
\[
\Ex_{g_{z,1},\ldots,g_{z,r} \underset{\sqrt{\rho}}{\sim} g} \left[\prod_{i \in [r]}{H}_i(g_{x,i},g_{z,i})\right]
\leq \Lambda_{\sqrt{\rho}}\Big(\Ex_{g_{z,1}}{H}_i(g_{x,1},g_{z,1}),\ldots,\Ex_{g_{z,r}}{H}_r(g_{x,r},g_{z,r})\Big).
\]
Here $\Lambda_{\sqrt{\rho}}(\delta_1,\ldots,\delta_r)$ is the $r$-ary Gaussian noise stability\footnote{Formally, $\Lambda_{\rho}(\delta_1,\ldots,\delta_r) = \Pr_{g \sim N(0,1)}\Pr_{g_1,\ldots,g_r \underset{\rho}{\sim}g}\left[\forall_i~g_i \leq \Phi^{-1}(\delta_i)\right]$ where $\Phi(\cdot)$ is the Gaussian CDF function.} for $\sqrt{\rho}$-correlated Gaussians w.r.t. halfspaces of Gaussian volumes $\delta_1,\ldots,\delta_r$. Furthermore, using our choice of $\rho$ and explicit bounds on the multivariate Gaussian CDF (Lemma \ref{lem:hspace}), we can bound\footnote{In the actual analysis, we will lose the additive factor of $\mu^r$ only when $\Ex_{g_{z,i}}H_i(g_{x,i},g_{z,i}) \leq \mu^r$ for some $i \in [r]$.}
\[
\Lambda_{\sqrt{\rho}}\Big(\Ex_{g_{z,1}}{H}_i(g_{x,1},g_{z,1}),\ldots,\Ex_{g_{z,r}}{H}_r(g_{x,r},g_{z,r})\Big)
\lesssim_r \prod_{i \in [r]}\Ex_{g_{z,i}}\left[{H}_i(g_{x,i},g_{z,i})\right] + \mu^r.
\]
Applying the above sequence of arguments for each fixing of $(g_{x,i})_{i \in [r]}$, we get that 
\[
\Ex_{(g_{x,i})_{i \in [r]}} \Ex_{g_{z,1},\ldots,g_{z,r} \underset{\rho}{\sim} g}\left[\prod_{i \in [r]} H_i(g_{x,i},g_{z,i})\right]
\lesssim_r  \Ex_{(g_{x,i})_{i \in [r]}}\left[\prod_{i \in [r]}\Ex_{g_{z,i}} H_i(g_{x,i},g_{z,i})\right] + \mu^r.
\]
Finally, again using the Invariance principle, we can transfer the analysis back to the setting of $(x_i,z_i)$ variables and get that 
\[
\Ex_{(g_{x,i})_{i \in [r]}}\left[\prod_{i \in [r]}\Ex_{g_{z,i}} H_i(g_{x,i},g_{z,i})\right] \approx \Ex_{(\cW^x[i])_{i \in [r]}}\left[\prod_{i \in [r]}\Ex_{\cW^z[i]} H_i(\cW^x[i],\cW^z[i])\right] = \Ex_{(x_i)_{i \in [r]}}\left[\prod_{i \in [r]}\Ex_{z_i}f_i(x_i,z_i)\right].
\] 	 
Putting the above inequalities together completes the proof.

\subsection{Handling the Variance Blow-up}

We conclude our discussion by briefly describing our approach to handling the $2^{O(R)}$-blow-up in relative weight guarantee in \cite{RT12}'s soundness analysis. Roughly, in their analysis, the $2^{O(R)}$-blow up results from trying to relate the covariance of the rounded values to covariance of the corresponding local variables. The following proposition states their bound from this step formally.
\begin{proposition}[Implicit in \cite{RT12}]
	Let $F:\mathbbm{R}^R \to [0,1]$ be an $R$-variate function. Let $(X_i,X_j)$ be a pair of jointly distributed $\{0,1\}$-valued random variables, and let $(g_i,g_j)$ be a pair of jointly distributed $R$-dimensional Gaussian random variables such that $(g_i(\ell),g_j(\ell))$ matches the covariance structure of $(X_i,X_j)$ for every $\ell \in [R]$. Then, 
	\[
	|\Ex F(g_i)F(g_j) - \Ex F(g_i) \Ex F(g_j)| \leq 2^{O(R)}|{\rm Cov}(X_i,X_j)|.
	\] 
\end{proposition}	 
Their analysis needs to use the above bound since the local-distributions used to construct their test can only guarantee a bound on the average covariance. This eventually results in an additional $2^{O(R)}$ multiplicative blow-up in the variance of the weight of the rounded solution, which in turn shows up as the multiplicative error term in the relative weight.

{\bf Our Fix}. Towards addressing this, our first observation is that since $F$ is $[0,1]$-valued, one can directly bound the covariance of the rounded values using the correlation of the Gaussian vector variables i.e., using elementary Hermite analysis one can show that 
\[
|\Ex F(g_i)F(g_j) - \Ex F(g_i) \Ex F(g_j)| \leq |{\rm Corr}(X_i,X_j)|,
\] 
where ${\rm Corr}(X_i,X_j)$ is the correlation \footnote{Formally, ${\rm Corr}(X_i,X_j) = \frac{\Ex\left[X_iX_j\right] - \mu_i\mu_j}{\sqrt{\mu_i(1 - \mu_i)}\sqrt{\mu_j(1 - \mu_j)}}$, where $\mu_i = \Ex[X_i], \mu_j = \Ex[X_j]$.} between $X_i$ and $X_j$ under the local distribution. Hence, if we started with a local distribution with small average-correlation, then we would be able to get a $R$-independent bound on the variance of the rounded weight. However, note that in general the correlation of a pair of random variable is incomparable to its covariance, and hence a bound on the average covariance is insufficient to exploit this observation. To get around this issue, we use the observation that random conditionings of {\em smoothened} Lasserre solutions directly yield local distributions with the stronger property of small average correlation. Combining this with the above observation yields the improved bound on the variance of the rounded solution. We refer the readers to Sections \ref{sec:pre-proc} and \ref{sec:round} for more details on this.

%

%% file: prelims-lift.tex
\section{Notation and Preliminaries}			\label{sec:prelim}

We introduce some frequently used notation in this work. Given a distribution $\cD$, we will use $x \sim \cD$ to denote the draw of a random variable $x$ from the distribution $\cD$. We use $x_1,\ldots,x_r \sim \cD$ to denote $r$-independent draws from the distribution $\cD$. If $x$ is vector-valued, then we use $x(j)$ to denote the $j^{th}$ entry of the variable $x$. Furthermore if $(x_1,\ldots,x_r)$ are jointly distributed as $\cD$, then we use $(x_i)_{i \in [r]} \sim \cD$ to denote a single draw of $(x_1,\ldots,x_r)$ from $\cD$, and it is not to be confused with the notation $x_1,\ldots,x_r \sim \cD$ defined above. For $R \geq 1$, for any positive semidefinite matrix $\Sigma \in \mathbbm{R}^{R \times R}$, we use $N(0_R,\Sigma)$ to denote the distribution over $R$-dimensional Gaussian distribution with mean $0_R$ and covariance matrix $\Sigma$; here $0_R$ denotes the all-zeros vector of dimension $R$. We use $I_R$ to denote the identity matrix of dimension $R \times R$. We will also use $\Pi_{[0,1]}$ to denote the clipping function which is defined as 
\begin{equation}			\label{eqn:pi-def}
\Pi_{[0,1]}(x) = 
\begin{cases}
	0 & \mbox{ if } x < 0, \\
	x & \mbox{ if } x \in [0,1], \\
	1 & \mbox{ if } x > 1. 
\end{cases}
\end{equation} 
For any $R \in \mathbbm{N}$, we use $\mathbbm{S}_R$ to denote the set of permutations on $R$ indices. For any $R$-dimensional vector $x$, and permutation $\pi \in \mathbbm{S}_R$, we use $\pi(x)$ to denote permuted vector $(x(\pi(1)),x(\pi(2)),\ldots,x(\pi(R)))$.

\subsection{Fourier Analysis}

A finite probability space $(\Omega,\mu)$ is characterized by a set $\Omega$ and a measure $\mu$ on the set. The vector space of all square integrable functions on $\Omega$ w.r.t. measure $\mu$ is denoted as $L_2(\Omega,\mu)$. For ease of notation, whenever the measure is clear from context, we will omit the measure and simply write $L_2(\Omega) = L_2(\Omega,\mu)$. Given a probability space $(\Omega,\mu)$ a Fourier basis for $L_2(\Omega,\mu)$ is an orthornomal basis $\phi_0 \equiv 1, \phi_1,\ldots,\phi_{|\Omega|-1}:\Omega \to \mathbbm{R}$. 

Using these Fourier bases, we can construct a basis for functions on the product probability spaces. Formally, given a product probability space $L_2(\prod_{i \in [R]} \Omega_i, \prod_{i \in [R]} \mu_i)$, let $\phi_{i,0} \equiv 1,\ldots,\phi_{i,|\ell_i| - 1}: \Omega_i \to \mathbbm{R}$ denote the Fourier basis for $L_2(\Omega_i,\mu_i)$, where $\ell_i := |\Omega_i|$. Given this basis, the elements of the Fourier basis for $L_2(\prod_i \Omega_i,\prod_i \mu_i)$ are $\{\phi_\alpha\}_{\alpha \in \prod_{i \in [r]}\mathbbm{Z}^R_{\leq \ell_i - 1}}$ where 
\[
\phi_\alpha(\omega) = \prod_{i \in [R]} \phi_{i,\alpha(i)}(\omega(i))		\qquad\qquad \forall~\omega \in \Omega^R, \alpha \in \prod_{i \in [R]} \mathbbm{Z}_{\leq \ell_i - 1}.
\]
Using this basis, any $f \in L_2(\Omega^R)$ can be expressed as a multi-linear polynomial in $\{\phi_{i}\}_{i \in |\Omega| - 1}$ as 
\[
f(\omega) = \sum_{\alpha \in \prod_{i \in [R]} \mathbbm{Z}_{\leq \ell_i - 1}}\wh{f}(\alpha) \prod_{i \in [R]} \phi_{\alpha(i)}(\omega(i))
\]
where $\wh{f}(\alpha)$ is referred to as the {\em Fourier coefficient} corresponding to multi-index $\alpha$.

{\bf Influences.} Given a function $f \in L_2(\prod^R_{i = 1} \Omega_i,\prod^R_{i = 1} \mu_i)$, the influence of a coordinate $j \in [R]$ on function $f$ -- denoted as ${\sf Inf}^{(\mu_j)}_j\left[f\right]$ -- is defined as 
\[
{\sf Inf}_j\big[f\big] = \Ex_{(\omega(j'))_{j' \neq j} \sim \prod_{j' \neq j} \mu_{j'}}\left[{\rm Var}_{\omega(j)}\left[f(\omega(1),\ldots,\omega(j))\right]\right].
\]
It is well-known that the influences can be expressed in terms of the Fourier coefficients of the function, as stated in the following fact.
\begin{fact}
	Given a function $f \in L_2(\prod_{i \in [r]} \Omega_i, \prod_{i \in [r]} \mu_i)$, for any fixed choice of Fourier basis, we have 
	\[
	{\sf Inf}_j\big[f\big] = \sum_{\alpha: \alpha(j) \neq 0} \wh{f}(\alpha)^2.
	\]
\end{fact}

{\bf Noise Operator} Given a product probability space $L_2(\Omega^R,\mu^R)$, and a $\omega \in \Omega^R$, for any $\rho \in [0,1]$, a $\rho$-correlated copy of $\omega$ -- denoted as $\omega' \underset{\rho}{\sim} \omega$ is sample as follows. For every $j \in [R]$, do the following independently:
\begin{itemize}
	\item W.p. $\rho$, set $\omega'(j) = \omega(j)$.
	\item W.p. $1 - \rho$, sample $\omega'(j) \sim \mu$ independently.
\end{itemize}
Then, the $\rho$-correlated noise operator in the space $(\Omega_R,\mu^R)$ is the stochastic operator $\sT^{(\Omega)}_{\rho}$ defined as
\[
\sT^{(\Omega)}_{\rho} f(\omega) = \Ex_{\omega' \underset{\rho}{\sim} \omega}\Big[f(\omega')\Big],		
\]
for every $f \in L_2(\Omega^R,\mu^R)$. The following are some well-known properties of noise operators.

\begin{fact}[Fourier Decay]				\label{fact:decay}
	Let $f \in L_2(\{0,1\}^R,\mu^R)$ be a function satisfying $\|f\|_\infty \leq 1$. Then for any $\eta \in (0,1)$ and $d \in \mathbbm{N}$ we have 
	\[
	{\rm Var}\left[f^{>d}\right] \leq (1 - \eta)^d
	\]
\end{fact}

{\bf Fourier basis for Biased Hypercube}. Our hardness reduction and soundness analysis will often involve the $p$-biased hypercubes. For any $p \in [0,1]$, the $p$-biased cube -- denoted by $\{0,1\}_p$ -- is the cube equipped with the ${\sf Bernoulli}(p)$ measure. We use $\phi^{(p)}$ to denote the {\em unique} non-trivial Fourier character corresponding to the $p$-biased cube, where 
\begin{equation}			\label{four:def}
\phi^{(p)}(x) \defeq \frac{x - p}{\sqrt{p(1 - p)}}.
\end{equation} 
The following is well-known, and follows directly using the above definitions.
\begin{fact}				\label{fact:p-fourier}
	For any $p \in [0,1]$, the Fourier expansion of any $f \in L_2(\{0,1\}^R_p)$ is unique and can be expressed as 
	\[
	f(x)  = \sum_{S \subseteq [R]}\wh{f}(S) \phi^{(p)}_S(x) = \sum_{S \subseteq [R]}\wh{f}(S) \prod_{i \in S} \phi^{(p)}_i(x(i)). 
	\]
\end{fact}
Using the above, we can derive the following corollary.
\begin{corollary}
	Let $f:\{0,1\}^R \to \mathbbm{R}$ be a function in $L_2(\{0,1\}^R_\mu)$ for some $\mu \in [0,1]$. Then the Fourier expansion of $f(x)$ is a multilinear polynomial $H$ in the variables $x(1),\ldots,x(R)$.
\end{corollary}
\begin{proof}
	 Let $P:\mathbbm{R}^R \to \mathbbm{R}$ denote the multilinear polynomial
	 \[
	 P(a) \defeq \sum_{S \subseteq [R]} \wh{f}(S) \prod_{i \in S} a(i).
	 \]
	 Now from fact \ref{fact:p-fourier}, it follows that there exists a linear transformation $A:\mathbbm{R}^R \to \mathbbm{R}^R$ which satisfies the identity
	 \[
	 A\left(x\right) = \left(\phi^{(p)}(x(j))\right)^R_{j = [R]}
	 \]
	 Then, $H = P \circ A$ is a multilinear polynomial in $(x(j))^r_{j = 1}$ which satisfies the identity:
	 \[
	 H(x) = P\left(\left(\phi^{(p)}(x(j))\right)^R_{j = 1}\right) = \sum_{S \subseteq [R]} \wh{f}(S)\prod_{j \in S} \phi^{(p)}(x(j)).
	 \]	
\end{proof}

\subsection{Ensembles, Polynomials, and Invariance Principle}				\label{sec:clt}

Let us recall some notation and terminology from \cite{IM12} which is required for stating its invariance principle. To begin with, an ensemble $\cX_j = (\cX_{j}[0] \equiv 1,\cX_{j}[1],\ldots,\cX_{j}[k])$ is just an ordered collection of (possibly dependent) random variables. An independent sequence of random variable ensembles $\cX:= (\cX_1,\cX_2,\ldots,\cX_n)$ is a sequence of individual ensembles such that across $j \in [n]$, $\cX_1,\ldots,\cX_n$ are jointly independent. A pair of independent ensemble sequences $\cX = (\cX_1,\ldots,\cX_n)$ and $\cX' = (\cX'_1,\ldots,\cX'_n)$ are said to have matching covariance structure if for every $j \in [n]$, we have $|\cX_j| = |\cX'_j|$ and $\Ex_{\cX}[\cX_j\cX^\top_j] = \Ex_{\cX'}[\cX'_j\cX'^\top_j]$. 

{\bf Multi-linear Polynomial}. A multi-index $\sigma:= (\sigma_1,\ldots,\sigma_n)$ is a sequence of non-negative integers. The degree of $\sigma$ is the number of non-zero entries. Given an ensemble sequence $\cX= (\cX_1,\ldots,\cX_n)$, a multilinear polynomial $Q$ on $\cX$ is a function of the form
\[
Q(\cX) = \sum_{\sigma \in \mathbbm{Z}^n_{\leq k}} \wh{Q}_\sigma \prod_{j \in [n]} \cX_{j}[\sigma(j)],
\]
where $\mathbbm{Z}_{\leq k} := \{0,1,\ldots,k\}$ i.e., every monomial contains exactly one variables from every ensemble in the sequence. The degree of a monomial corresponding to $\sigma$ is simply $|\sigma|$ i.e., the number of non-zero entries in $\sigma$. Finally, we use $Q^{\leq d}(\cX)$ to denote the truncation of $Q$ to monomials with degree at most $d$. Finally, we can define the influence of the $j^{th}$-ensemble on $P$ as 
\[
{\sf Inf}_{\cX_j}\left[P\right] = \Ex_{(\cX_{j'})_{j' \neq j}}\left[{\rm Var}_{\cX_j}\left[P(\cX)\Big|\left(\cX_{j'}\right)_{j'\neq j}\right]\right].
\] 
Now we are ready to state the Invariance principle used in our applications
\begin{theorem}[Theorem 3.6~\cite{IM12}]					\label{thm:clt}
	Let $\cX = (\cX_1,\ldots,\cX_n)$ be an independent sequence of ensembles, such that $\Pr\left[\cX_j = x\right] \geq \alpha > 0$ for all $j,x$. Fix $\gamma, \tau \in (0,1)$ and let $Q = (Q_1,\ldots,Q_r)$ be a $r$-dimensional multilinear polynomial on $\cX$ such that ${\rm Var}[Q_i(\cX)] \leq 1$, ${\rm Var}\left[Q^{>d}_i(\cX)\right] \leq (1 - \gamma)^{2d}$ and ${\sf Inf}_j\left[Q^{\leq d}_i\right] \leq \tau$, where $d = \frac{1}{18}\log\frac1\tau/\log\frac1\alpha$. Finally, let $\Psi:\mathbbm{R}^r \to \mathbbm{R}$ be Lipschitz continuous with constant $A$. Then, 
	\[
	\left|\Ex\left[\Psi(Q(\cX))\right] - \Ex\left[\Psi(Q(\cG))\right] \right| \leq C_r A \tau^{\frac{\gamma}{18}/\log\frac1\alpha},
	\]
	where $\cG$ is an independent sequence of Gaussian ensembles with the same covariance structure as $\cX$, and $C_r$ is a constant depending only on $r$.
\end{theorem}

\subsection{Noise Stability Bounds}					\label{sec:stab}

We shall need the following multi-dimensional version of Borell's isoperimetric inequality.

\begin{theorem}[Theorem 1.2~\cite{IM12}]					\label{thm:egt}
	The following holds for any $r \geq 2, d\in \mathbbm{N}$ and $\rho \in [0,1]$. Let $(X_1,\ldots,X_r)$ be a set of jointly distributed Gaussian vector random variables, such that $X_i \sim N(0_d,I_d)$ for every $i \in [r]$, and for any $i \neq j$ we have ${\rm Cov}(X_i,X_j)  = \rho I_d$. Then for any $A_1,\ldots,A_r \subseteq \mathbbm{R}^d$ we have 
	\[
	\Pr_{X_1,\ldots,X_r}\left[\forall^r_{i = 1} X_i \in A_i\right] \leq \Pr_{X_1,\ldots,X_r}\left[\forall^r_{i = 1} X_i(1) \leq \Phi^{-1}(\mu_i) \right],
	\]
	where for every $i \in [r]$, $\mu_i := \Pr_{X_i}\left[X_i \in A_i\right]$.
\end{theorem}

Next, we define the notion of sampling correlated  Gaussians.
\begin{definition}[Correlated Gaussians]			\label{defn:g-corr}
	Fix $R \in \mathbbm{N}$ and $\rho \in [-1,1]$. Then given a Gaussian vector $g \sim N(0_R,I_R)$, a $\rho$-correlated draw of a Gaussian $h$ from $g$, denoted as $h \underset{\rho}{\sim} g$, is generated using the following process: independently sample a Gaussian vector $\zeta \sim N(0_R,I_R)$, and let $h = \rho \cdot g + \sqrt{1 - \rho^2} \cdot \zeta$. It can be verifed that for every $j \in [R]$, 
	\[
	(g(j),h(j)) \sim N\left(\begin{bmatrix} 	0 \\ 0 	\end{bmatrix},\begin{bmatrix}
		1 & \rho \\ \rho  & 1
	\end{bmatrix}\right).
	\]
\end{definition}

{\bf $r$-ary Gaussian Stability}. We define the $r$-ary Gaussian stability for halfspaces with volumes $\delta_1,\ldots,\delta_r$ as
\begin{equation}				\label{eqn:lambda-def}
\Lambda_\rho(\delta_1,\ldots,\delta_r) \defeq \Pr_{\substack{g \sim N(0,1) \\ g_1,\ldots,g_r \underset{\rho}{\sim} g}}\left[\forall_{i \in [r]}~g_i \leq \Phi^{-1}(\delta_i)\right],
\end{equation}
where $\Phi(\cdot)$ is the Gaussian CDF function. We derive the following elementary corollary which extends the above theorem to the setting of $[0,1]$-valued functions.
\begin{corollary}					\label{corr:egt}
	The following holds for any $r \geq 2, d\in \mathbbm{N}$ and $\rho \in [0,1]$. Let $f_1,\ldots,f_r :\mathbbm{R}^d \to [0,1]$ be arbitrary functions. Then,
	\begin{equation}				\label{eqn:egt}
		\Ex_{g \sim N(0,I_d)}\Ex_{g_1,\ldots,g_r \underset{\rho}{\sim} g} \left[\prod_{j \in [r]} f_j(g_j)\right]
		\leq \Lambda_\rho(\mu_1,\ldots,\mu_r),
	\end{equation}
	where $\mu_i:= \Ex_{g \sim N(0,I_d)} \left[f_i(g)\right]$ for every $i \in [r]$ .
\end{corollary}
\begin{proof}
	Fix a choice of $\mu_1,\ldots,\mu_r$, and consider the set $\cF$ of functions $(f_1,\ldots,f_r):\mathbbm{R}^d \to [0,1]^r$ such that $\Ex_{g \sim N(0,I_d)} \left[f_j(g)\right] = \mu_j$ for every $j \in [r]$. Then observe that the set of functions is convex. Furthermore, fixing for any $j \in [r]$, and choices of functions $\{f_{j'}\}_{j' \neq j}$, the mapping 
	\[
	f_j \mapsto \Ex_{g_1,\ldots,g_r}\left[\prod_{j' \in [r]} f_{j'}(g_{j'})\right]
	\] 
	is linear in $f_j$, and hence, it is maximized in $\cF$ for some choice of $f_j :\mathbbm{R}^d \to \{0,1\}$. Therefore, it follows that it suffices to prove the inequality for $\{0,1\}$-valued choices of $f_1,\ldots,f_r$. To that end, fix such a choice of $f_1,\ldots,f_r:\mathbbm{R}^d \to \{0,1\}$, and let $\mu_j := \Ex_g\left[f_j(g)\right]$, and $A_j = {\rm supp}(S_j)$. Then,
	\begin{align*}
		\Ex_{g \sim N(0,1)^r}\Ex_{g_1,\ldots,g_r \underset{\rho}{\sim} g} \left[\prod_{j \in [r]} f_j(g_j)\right]
		&= \Pr_{\substack{ g \sim N(0,1)^r \\ g_1,\ldots,g_r \underset{\rho}{\sim} g}} \left[\mathop{\forall}^r_{j = 1}~g_j \in A_j\right] \\
		&\leq \Pr_{\substack{ g \sim N(0,1) \\ g_1,\ldots,g_r \underset{\rho}{\sim} g}} \left[\mathop{\forall}^r_{j = 1}~g_j \leq  \Phi^{-1}(\mu_j)\right],
	\end{align*}
	where the last inequality follows from applying Theorem \ref{thm:egt}.
\end{proof}

\subsection{Small-Set Expansion Hypothesis}

Given a $d$-regular graph $G = (V,E)$, and a subset $S \subseteq V$ of size at most $|V|/2$, the expansion of $S$ in $G$ -- denoted by $\phi_G(S)$ -- is defined as follows: 
\[
\phi_G(S) \defeq \Pr_{(i,j) \sim E| i \in S}\left[j \in S\right].
\]
Our reductions are from the \smallsetexpansion~problem which we define formally below:

\begin{definition}[\smallsetexpansion]					\label{defn:sseh}
	For any $\epsilon,\delta \in (0,1)$ and $M \in \mathbbm{N}$, an instance of $(\epsilon,\delta,M)$-\smallsetexpansion~problem is characterized by a regular graph $G = (V,E)$. The objective here is to distinguish between the following cases:
	\begin{itemize}
		\item {\bf YES Case}. There exists a set $S$ of volume $\delta$ such that $\phi_G(S) \leq \epsilon$.
		\item {\bf NO Case}. For every $S \subseteq V$ such that ${\sf Vol}(S) \in \left[\frac{\delta}{M},M\delta\right]$ we have $\phi_G(S) \geq 1 - \epsilon$.
	\end{itemize}
\end{definition}

Our reduction uses the hard instances of \smallsetexpansion~given by the following theorem as the starting point.

\begin{conjecture}[\cite{RS10},\cite{RST12}]						\label{conj:sseh}
	There exists a constant $\epsilon_0 \in (0,1)$, such that the following holds. For every $\epsilon \in (0,\epsilon_0)$, and $M \leq 1/\sqrt{\epsilon}$ there exists $\delta = \delta(\epsilon,M)$ such that $(\epsilon,\delta,M)$-\smallsetexpansion~is \NP-hard. 
\end{conjecture} 

\subsection{Lasserre Hierarchy, Pseudo-variables, and Vector Solution}

The Lasserre (aka the Sum-of-Squares) hierarchy is a sequence of strengthenings of a basic SDP relaxation. The $\ell$-round Lasserre relaxation introduces local distributions on subsets of size at most $\ell$ that are locally consistent, and whose $\ell^{th}$-order moment matrix is PSD.  
In particular, the $\ell$-round Lasserre lifting of a basic SDP in Figure \ref{fig:lass} optimizes over the space of ``collections'' of local distributions $\theta := \{\theta_S\}_{|S| \leq \ell}$, which for any subset $S \subseteq V_{\rm gap}$ of size at most $\ell$, $\theta_S$ defines a local distribution over partial assignments to vertices in $S$. It then enforces local consistency constraints that ensure that for any pair of sets $A,B$, the corresponding distributions $\theta_A$ and $\theta_B$ are consistent on the support of the intersection $A \cap B$. In particular, this allows us to define pseudo-variables $X_1,\ldots,X_n$ which are not necessarily jointly distributed but are locally consistent i.e., for every subset $S$ of size at most $\ell$, the corresponding collection of variables $X_S = (X_i)_{i \in S}$ is distributed as $\theta_S$. Finally, it ensures that $\ell^{th}$-order pseudo-covariance matrix is PSD. We describe the $\ell$-round Lasserre relaxation in the figure below:

\begin{figure}[ht!]
	\begin{mdframed}
		\begin{eqnarray}
			\textnormal{Maximize} & \Ex_{e \sim E} \Pr_{X_e \sim \theta_e}\Big[\psi(X_e) = 1\Big]	&  \\
			\textnormal{Subject to} & \Ex_{i \sim V_{\rm gap}} \Pr_{X_i \sim \theta_\{i\}}\left[X_i = 1\right] = \mu & \\
									& \Pr_{\theta_A}\left[X_{A \cap B} = \omega\right] = \Pr_{\theta_B}\left[X_{A \cap B} = \omega \right]
									& \forall~|A \cup B| \leq \ell, \omega \in \{0,1\}^{A \cap B} \\
									& M_\ell(\theta) \succeq 0,
		\end{eqnarray}
	\end{mdframed}
	\caption{Lasserre Relaxation ${\sf Lass}_{\mu,\ell}(G_{\rm gap})$}
\end{figure}
In the above, $M_\ell(\theta)$ is the ${[n]\choose{\ell/2}} \times {[n]\choose {\ell/2}}$-size matrix whose rows and columns are indexed by subsets of size at most $\ell/2$, and for any such pair of row-column subset index $A,B$, the corresponding entry is defined as $M_\ell(\theta)[A,B] = \Ex_{\theta}\left[\prod_{i \in A \cup B} X_i\right]$ -- note that the RHS is {\em well-defined} due to the local consistency constraints. 

{\bf Covariance, Correlation}. We will refer a ${\sf Lass}_{\mu,\ell}(G_{\rm gap})$ feasible solution $\theta = \{\theta_S\}_{|S| \leq \ell}$ as a degree-$\ell$ pseudo-distribution. The $\ell$-wise local consistency allows us to consistently define various quantities involving subsets of at most $\ell$-variables. In particular, for any pair of pseudo-variables $X_i,X_j$, we can define the covariance 
\[
{\sf Cov}_{\theta}(X_i,X_j) = \Ex_{\theta}\left[X_iX_j\right] - \Ex_{\theta}\left[X_i\right]\Ex_{\theta}\left[X_j\right].
\]
Analogously, we can also define the standard deviation of variable $X_i$ as ${\sf stdev}_\theta(X_i) = (\Ex_{\theta}X^2_i - (\Ex_\theta X_i)^2)^{1/2}$, and the correlation between a pair of variables $X_i,X_j$ is then defined as 
\[
{\sf Corr}_\theta(X_i,X_j) = \frac{{\sf Cov}_{\theta}(X_i,X_j)}{{\sf stdev}_\theta(X_i)\cdot{\sf stdev}_\theta(X_j)}.
\]

{\bf Degree-$2$ Solution}. Given a degree-$\ell$ pseudo-distribution, one can identify vectors $\{u_\emptyset\} \cup \{u_i\}_{i \in V_{\rm gap}}$ from the Cholesky decomposition of the second order covariance matrix $M_2(\theta)$. The following proposition lists some easy to verify properties of these vectors.

\begin{proposition}				\label{prop:vector}
	Let $\{u_\emptyset\} \cup \{u_i\}_{i \in V_{\rm gap}}$ be the vector solution as defined above. For every $i \in V_{\rm gap}$ let $u_i = \mu_i \uphi + w_i$, where $w_i \perp u_i$. Then the following properties hold:
	\begin{itemize}
		\item For any $i \in V_{\rm gap}$, $\mu_i = \Pr_{X_i \sim \theta}\left[X_i = 1\right]$.
		\item For any $i \in V_{\rm gap}$, $\|w_i\| = {\sf stdev}_{\theta}(X_i)$.
		\item For any $i,j \in V_{\rm gap}$, $\langle w_i,w_j \rangle = {\sf Cov}_{\theta}(X_i,X_j)$
		\item For any $i,j \in V_{\rm gap}$. $\langle u_i,u_j \rangle = \Pr_{\theta}\left[X_i = 1,X_j = 1\right]$.
	\end{itemize} 
\end{proposition}

%% file: pre-proc.tex
\section{Small Average-Correlation via Pre-processing}				\label{sec:pre-proc}

Let $(G_{\rm gap},\theta)$ be a $(\ell,\mu,c,s,\gamma)$-gap instance as in Definition \ref{defn:gap-inst} with $G_{\rm gap} = (V_{\rm gap},E_{\rm gap},\tilde{w},w)$ where $\tilde{w}:V_{\rm gap} \to \mathbbm{R}_{\geq 0}$ and $w:E_{\rm gap} \to \mathbbm{R}_{\geq 0}$ define the vertex and edge weights, and $\theta$ is a feasible solution to ${\sf Lass}_{\mu,\ell}\left(G_{\rm gap}\right)$ i.e., the $\mu$-constrained $\ell$-round Lasserre relaxation described in Figure \ref{fig:lass-1}.

\begin{figure}[ht!]
	\begin{mdframed}
		\begin{eqnarray*}
			\textnormal{Maximize} & \Ex_{e \sim w} \Pr_{X_e \sim \theta_e}\left[\psi(X_e) = 1\right] & \\
			\textnormal{Subject to} & \Ex_{i \sim G_{\rm gap}}\Pr_{X_i \sim \theta_i}\left[X_i = 1\right] = \mu &\\
			& \langle u_{A,\alpha},u_{B,\beta} \rangle = \Pr_{\theta_{A \cup B}}\left[X_A = \alpha, X_B = \beta\right] & 
			\forall~A,B \subseteq [\ell] \textnormal{such that} \\
			& & |A|,|B| \leq \ell/2, \\
			& & \alpha \in \{0,1\}^A, \beta \in \{0,1\}^B
		\end{eqnarray*}
	\end{mdframed}
	\caption{${\sf Lass}_{\mu,\ell}(G_{\rm gap})$}
	\label{fig:lass-1}
\end{figure} 

We introduce a definition that specifies some useful properties of a feasible pseudo-distribution.

\begin{definition}[Smooth and Independent Pseudo-distributions]
	Let $\theta$ be a feasible solution ${\sf Lass}_{\mu,\ell}(G_{\rm gap})$. We will say that $\theta$ is $\gamma$-{\bf smooth} if for every edge $e$, and every $\alpha \in \{0,1\}^e$, we have $\Pr_{X_e \sim \theta_e}\left[X_e = \alpha\right] \geq \gamma$. 	Furthermore, we will say that $\theta$ is $\gamma$-{\bf independent} if 
	\[
	\Ex_{i,j \sim G_{\rm gap}}\left[|{\sf Corr}_\theta(X_i,X_j)|\right] \leq \gamma,
	\]
	where ${\rm Corr}_\theta(X_i,X_j)$ is the correlation of the variables corresponding to vertices $i,j \in V_{\rm gap}$ under the local distribution $\theta$.
\end{definition}	

The following shows that given a $\ell$-round feasible Lasserre solution, we can construct a  $\ell - O(1/\gamma^C)$-round Lasserre solution that is $(\gamma\mu)^r$ smooth and $\gamma^2$ independent with almost matching completeness.

\begin{lemma}					\label{lem:smooth}
	Let $\theta:= \{\theta_A\}_{|A| \leq \ell}$ be a feasible solution to ${\sf Lass}_{\mu,\ell}(G_{\rm gap})$. Then there exists $t = t(\gamma)$ such that we can construct $\theta' = \{\theta'_A\}_{|A| \leq \ell - t}$ which is a feasible solution to ${\sf Lass}_{\mu,\ell - t}(G_{\rm gap})$ satisfying the following properties:
	\begin{itemize}
		\item {\bf Smoothness}. There exists a subset of $S$ of size at most $O(1/\gamma^4)$ such that for every subset $A \subseteq V_{\rm gap} \setminus S$ of size at most $\ell - t$, and every assignment $\alpha_A \in \{0,1\}^A$ corresponding to the variables in $A$ we have 
		\[
		\Pr_{X_A \sim \theta'_A}\left[X_A = \alpha_A\right] \geq (\gamma\mu)^{|A|}.
		\] 
		\item {\bf Small Average Correlation}.
		\[
		\Ex_{i \sim G_{\rm gap}}\left[|{\rm Corr}_{\theta'}(X_i,X_j)|\right] \leq \gamma^2.
		\]
		\item {\bf Almost matching completeness}.
		\[
		\Ex_{e \sim E_{\rm gap}}\Pr_{X_e \sim \theta'_e}\left[\psi(X_e) = 1\right] \geq c - 2\gamma r,
		\]
		where $c$ denotes the SDP objective value w.r.t. $\theta$.	 
	\end{itemize} 
\end{lemma}

We shall use the following result from \cite{RT12} to derive the above.

\begin{theorem}[\cite{RT12}]				\label{thm:cond}
	Let $\theta$ be a ${\sf Lass}_{\mu,\ell}(G_{\rm gap})$-feasible solution. Fix $\gamma > 0$. Then there exists a subset $S \subseteq V_{\rm gap}$ of size at most $O(1/\gamma^8)$, and a partial assignment $\alpha_S \in \{0,1\}^S$ to the variables in $S$ such that $\tilde{\theta} = \theta|X_S \gets \alpha$ satisfies
	\begin{equation}				\label{eqn:thm-1}
		\Ex_{i,j \sim V_{\rm gap}}\left[|{\rm Cov}_{\tilde{\theta}}(X_i,X_j)|\right] \leq \gamma^4
	\end{equation}
	and 
	\begin{equation}				\label{eqn:thm-2}
		\Ex_{e \sim E_{\rm gap}}\Pr_{X_e \sim \tilde{\theta}_e}\left[\psi(X_e) = 1\right] \geq c - \gamma
	\end{equation}
\end{theorem}	

Now we prove the above lemma.

\begin{proof}
	We construct the new pseudo-distribution by first applying a natural {\em smoothening} operator, and then apply Theorem \ref{thm:cond} to find a conditioning under which the resultant pseudo-distribution will satisfy the properties claimed in the lemma.
	
	{\bf Step 1}. Given $\theta$, construct a ${\sf Lass}_{\mu,\ell}(G_{\rm gap})$ feasible solution $\hat{\theta}$ as follows. For every $A \subseteq V_{\rm gap}$ of size at most $\ell$, we define the local distribution $\hat{\theta}$ as follows.
	\begin{itemize}
		\item Sample $X_A \sim \theta_A$.
		\item For every $i \in A$, do the following independently: w.p $1 - \gamma$, set $\hat{X}_i = X_i$, and w.p. $\gamma$, sample $\hat{X}_i \sim \{0,1\}_\mu$.
	\end{itemize}
	Note that the resulting distribution is $\ell$-Lasserre feasible (Claim \ref{cl:feas}). Furthermore, 
	\begin{equation}				\label{eqn:proof-1}
		\Ex_{i \sim G_{\rm gap}}\Pr_{\hat{X}_i \sim \hat{\theta}}\left[\hat{X}_i = 1\right]
		= (1 - \gamma) \Ex_{i \sim G_{\rm gap}}\Pr_{{X}_i \sim \theta}\left[{X}_i = 1\right] + \gamma \mu = \mu.
	\end{equation}
	Finally, we have that
	\begin{equation}				\label{eqn:proof-2}
		\Ex_{e \sim E_{\rm gap}}\Pr_{\hat{X}_e \sim \hat{\theta_e}}\left[\psi(\hat{X}_e) = 1\right] \geq (1 - \gamma)^r \Ex_{e \sim E_{\rm gap}}\Pr_{{X}_e \sim {\theta_e}}\left[\psi({X}_e) = 1\right] \geq c - \gamma r.
	\end{equation}
	
	{\bf Step 2}. Apply Theorem \ref{thm:cond} on $\hat{\theta}$ to find a ${\sf Lass}_{\mu,\ell - t}(G_{\rm gap})$-feasible solution $\tilde{\theta}$ that satisfies \eqref{eqn:thm-1} and \eqref{eqn:thm-2}.
	
	{\bf Analysis}. Let $X_S \gets \alpha_S$ be the conditioning identified by Theorem \ref{thm:cond} in step $2$. For the smoothness property, we observe that for any subset $A \subseteq V_{\rm gap} \setminus S$, using the definition on the conditioning we have
	\begin{equation}						\label{eqn:prob}
		\Pr_{\tilde{X}_A \sim \tilde{\theta}_A}\Big[\tilde{X}_A = \alpha_A \Big]  
		= \frac{\Pr_{\hat{X}_{A \cup S} \sim \hat{\theta}_{A \cup S}}\Big[\hat{X}_{A \cup S} = \alpha_{A \cup S} \Big]}{\Pr_{\hat{X}_S \sim \hat{\theta}_S}\Big[\hat{X}_S = \alpha_S\Big]}.
	\end{equation}
	Further, note that for sampling $\hat{X}_{A \cup S} \sim \hat{\theta}_{A \cup S}$, fixing $X_{A \cup S} \sim \theta_{A \cup S}$, we have that $\hat{X}_A$ and $\hat{X}_S$ are independent. Using this observation, we can further simplify and bound the above numerator as 
	\begin{align*}
		\Pr_{\hat{X}_{A \cup S} \sim \hat{\theta}_{A \cup S}}\Big[\hat{X}_{A \cup S} = \alpha_{A \cup S} \Big]
		&= \Ex_{\hat{X}_{A \cup S} \sim \hat{\theta}_{A \cup S}}\Pr_{\hat{X}_S}\left[\hat{X}_S = \alpha_S \Big|X_S\right]\Pr_{\hat{X}_A}\left[\hat{X}_A = \alpha_A \Big|X_A\right]		\\
		&\geq (\gamma \mu)^{|A|}\Ex_{\hat{X}_{A \cup S} \sim \theta_{A \cup S}}\Pr_{\hat{X}_S}\left[\tilde{X}_S = \alpha_S \Big|X_S\right] \\
		& = (\gamma \mu)^{|A|} \Pr_{\hat{X}_S \sim \hat{S}}\left[\hat{X}_S = \alpha_S\right],
	\end{align*}
	where in the first step, we use that in the construction of $\hat{\theta}$, fixing $X_S$, $\hat{X}_S$ is sampled by re-randomizing each coordinate of $S$ independently. This combined with \eqref{eqn:prob} gives us that 
	\[
	\Pr_{\tilde{X}_A \sim \tilde{\theta}_A}\Big[\tilde{X}_A = \alpha_A \Big] \geq (\gamma \mu)^{|A|}
	\]
	for every $A \subseteq V_{\rm gap}\setminus S$ of size at most $(\ell - t)$ -- this establishes the smoothness property. The almost matching completeness property follows from \eqref{eqn:proof-2} and \eqref{eqn:thm-2}. 

	To finish the proof, we establish the small-average-correlation property. Towards that, denote $n = |V_{\rm gap}|$ and let $i\sim G_{\rm gap} \setminus S$ denote a draw of a random vertex conditioned on $i \notin S$. Then note that since $|V_{\rm gap}|^{-1/C} \leq \gamma$ for a large constant $C$, we have 
	\[
	\left|\Ex_{i,j \sim G_{\rm gap}\setminus S} |{\rm Cov}_{\tilde{\theta}}(X_i,X_j)| - \Ex_{i,j \sim G_{\rm gap}} |{\rm Cov}_{\tilde{\theta}}(X_i,X_j)| \right| \leq o_n(1)
	\]
	since $S$ of size at most $1/\gamma^8$. Note that for any $i \in V_{\rm gap} \setminus S$, we have $\Pr_{\tilde{X}_i}\left[\tilde{X}_i = 1\right] \in [\gamma \mu, 1 - \gamma \mu]$ and hence ${\sf stdev}(\tilde{X}_i) \geq \sqrt{\gamma\mu/2}$. Hence, 
	\begin{align*}
		\Ex_{i,j \sim G_{\rm gap}}\left[|{\sf Corr}_{\tilde{\theta}}(\tilde{X}_i,\tilde{X}_j)|\right]
		&\leq \Ex_{i,j \sim G_{\rm gap}\setminus S}\left[|{\sf Corr}_{\tilde{\theta}}(\tilde{X}_i,\tilde{X}_j)|\right] + o_n(1) \\
		&= \Ex_{i,j \sim G_{\rm gap}\setminus S}\left[\frac{|{\sf Cov}_{\tilde{\theta}}(\tilde{X}_i,\tilde{X}_j)|}{{\sf stdev}_{\tilde{\theta}}(X_i)\cdot{\sf stdev}_{\tilde{\theta}}(X_j)}\right] + o_n(1) \\
		&{\leq} \Ex_{i,j \sim G_{\rm gap}\setminus S}\left[\frac{|{\sf Cov}_{\tilde{\theta}}(\tilde{X}_i,\tilde{X}_j)|}{\sqrt{\gamma \mu/2}\cdot\sqrt{\gamma \mu/2}}\right] + o_n(1) \\
		&= \frac{2}{\gamma\mu}\Ex_{i,j \sim G_{\rm gap}\setminus S}\left[|{\sf Cov}_{\tilde{\theta}}(X_i,X_j)|\right] + o_n(1) \\
		&= \frac{2}{\gamma\mu}\Ex_{i,j \sim G_{\rm gap}\setminus S}\left[|{\sf Cov}_{\tilde{\theta}}(X_i,X_j)|\right] + o_n(1) \\
		&\leq \frac{1}{\gamma\mu}\Ex_{i,j \sim G_{\rm gap}}\left[|{\sf Cov}_{\tilde{\theta}}(X_i,X_j)|\right] + o_n(\gamma^{-2}\mu^{-2}) \\
		&\leq \gamma^2.
	\end{align*}
\end{proof}

%% file: redn-lift.tex
\section{The Reduction}				\label{sec:redn}

Let $(G_{\rm gap},\theta')$ be a $(c,s,\mu,\ell,\gamma)$-gap instance as in Definition \ref{defn:gap-inst}. Then using Lemma \ref{lem:smooth} on $(G_{\rm gap},\theta')$, we obtain a $\gamma^2$-independent $(c - O(r\gamma),s,\mu,\ell')$-gap instance $(G_{\rm gap},\theta)$ such that $\theta$ satisfies conditions $1$-$3$ from Lemma \ref{lem:smooth}. 

{\bf Additional Noise Operators}. We introduce some additional noise operators that will be used in the construction of the reduction. 

\begin{definition}[Noise Random Walk Operator]
Given a regular graph $G = (V,E)$, for any $\eta \in (0,1)$, the $\eta$-noisy random walk operator on $G$ -- denoted by $G_\eta$ -- is a stochastic operator on $L_2(V)$ which is defined as follows. For any $A \in V$, we sample $B \sim G_\eta(A)$ as follows: 
\begin{itemize}
	\item W.p. $(1 - \eta)$, sample $B$ as a uniformly random neighbor of $A$.
	\item W.p. $\eta$, sample $B$ uniformly from $V$.
\end{itemize}
\end{definition}

\begin{definition}[$R$-dimensional noise operators with Leakage]
Given $z \in \{\bot,\top\}^R$, a regular graph $G = (V,E)$, and $\mu \in [0,1]$, the operator $M^{(\mu)}_z$ is the following stochastic operator on $L_2(V^R \otimes \{0,1\}^R_\mu)$. Given $(A,x) \in V^R \times \{0,1\}^R$, we sample $(A',x') \sim M^{(\mu)}_z(A,x)$ by doing the following independently for every $i \in [R]$:
\begin{itemize}
	\item If $z(i) = \top$, then set $(A'(i),x'(i)) = (A(i),x(i))$.
	\item If $z(i) = \bot$, then sample $A'(i) \sim V$ and $x'(i) \sim \{0,1\}_\mu$ independently. 
\end{itemize}
\end{definition}

\subsection{Test Distribution}

Given the gap instance $(G_{\rm gap},\theta)$ from above, we describe our reduction as a dictatorship test in the following figure (Figure \ref{fig:test-1}).
\begin{figure}[ht!]
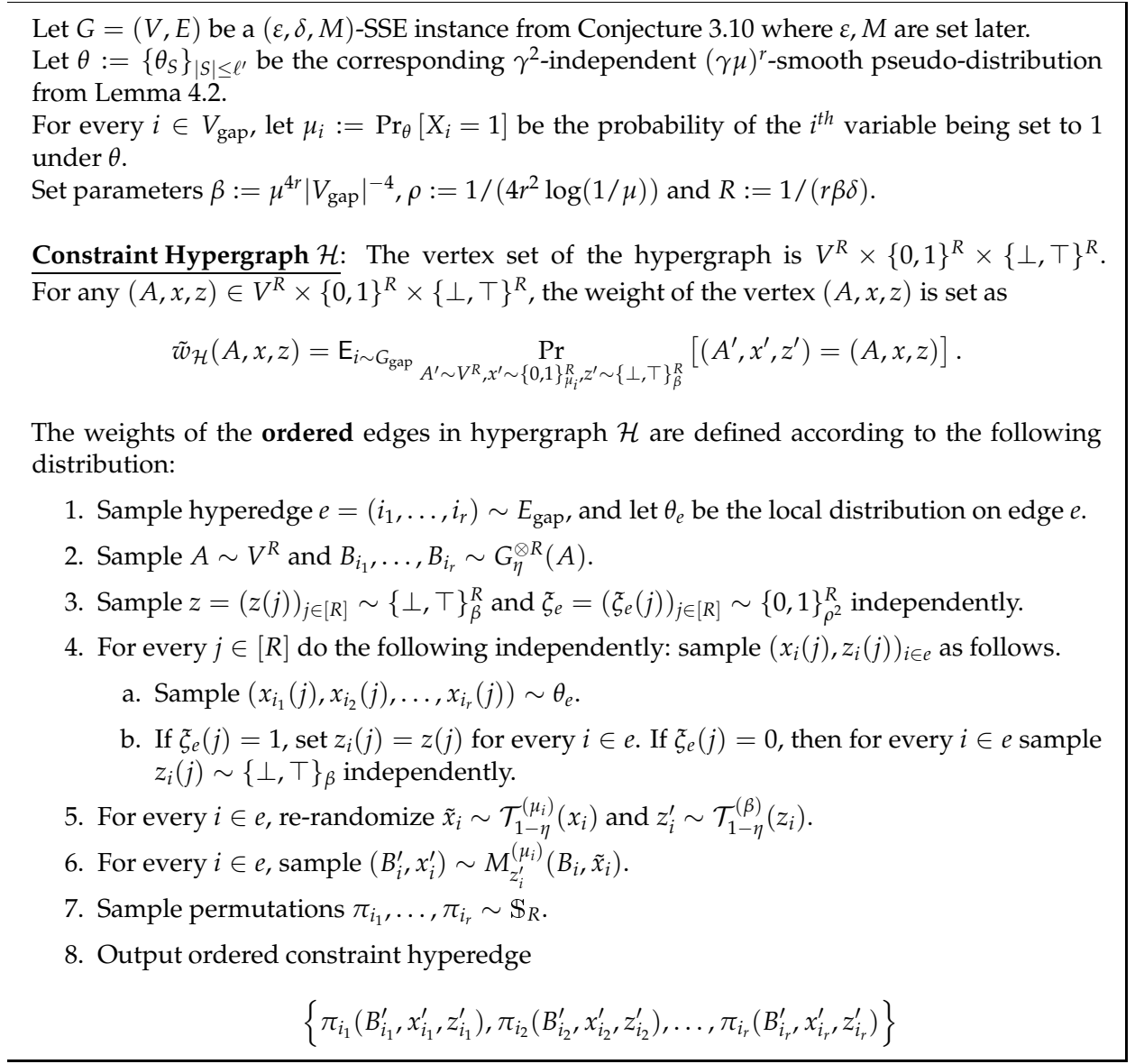

	\begin{mdframed}
		Let $G = (V,E)$ be a $(\epsilon,\delta,M)$-SSE instance from Conjecture \ref{conj:sseh} where $\epsilon,M$ are set later. \\
		Let $\theta := \{\theta_S\}_{|S| \leq \ell'}$ be the corresponding $\gamma^2$-independent $(\gamma\mu)^r$-smooth pseudo-distribution from Lemma \ref{lem:smooth}. \\ 
		For every $i \in V_{\rm gap}$, let $\mu_i := \Pr_\theta\left[X_i = 1\right]$ be the probability of the $i^{th}$ variable being set to $1$ under $\theta$.\\
		Set parameters $\beta := \mu^{4r}|V_{\rm gap}|^{-4}$, $\rho := 1/(4r^2\log(1/\mu))$ and $R := 1/(r\beta\delta)$. \\ \\
		\underline{\bf Constraint Hypergraph $\cH$}:		
		The vertex set of the hypergraph is $V^R \times \{0,1\}^R \times \{\bot,\top\}^R$. For any $(A,x,z) \in V^R \times \{0,1\}^R \times \{\bot,\top\}^R$, the weight of the vertex $(A,x,z)$ is set as 
		\[
		\tilde{w}_{\cH}(A,x,z) = \Ex_{i \sim G_{\rm gap}}\Pr_{A' \sim V^R, x' \sim \{0,1\}^R_{\mu_i},z' \sim \{\bot,\top\}^R_\beta}\left[(A',x',z') = (A,x,z)\right]. 
		\] 
		The weights of the {\bf ordered} edges in hypergraph $\cH$ are defined according to the following distribution: \\[-8pt]
		\begin{enumerate}
			\item Sample hyperedge $e = (i_1,\ldots,i_r) \sim E_{\rm gap}$, and let $\theta_e$ be the local distribution on edge $e$. \\[-8pt]
			\item Sample $A \sim V^R$ and $B_{i_1},\ldots,B_{i_r} \sim G^{\otimes R}_\eta(A)$. \\[-8pt]
			\item Sample $z = (z(j))_{j \in [R]} \sim \{\bot,\top\}^R_\beta$ and $\xi_e = (\xi_e(j))_{j \in [R]} \sim \{0,1\}^R_{\rho^2}$ independently. \\[-8pt]
			\item For every $j \in [R]$ do the following independently: sample $(x_i(j),z_i(j))_{i \in e}$ as follows. \\[-8pt]
			\begin{itemize}
				\item[a.] Sample $(x_{i_1}(j),x_{i_2}(j),\ldots,x_{i_r}(j)) \sim \theta_e$. \\[-8pt]
				\item[b.] If $\xi_e(j) = 1$, set $z_i(j) = z(j)$ for every $i \in e$. If $\xi_e(j) = 0$, then for every $i \in e$ sample $z_i(j) \sim \{\bot,\top\}_\beta$ independently. \\[-8pt]
			\end{itemize}
			\item For every $i \in e$, re-randomize $\tilde{x}_i \sim \sT^{(\mu_i)}_{1 - \eta}(x_i)$ and $z'_i \sim \sT^{(\beta)}_{1 - \eta}(z_i)$. \\[-8pt]
			\item For every $i \in e$, sample $(B'_i,x'_i) \sim M^{(\mu_i)}_{z'_i}(B_i,\tilde{x}_i)$. \\[-8pt]
			\item Sample permutations $\pi_{i_1},\ldots,\pi_{i_r} \sim \mathbbm{S}_R$. \\[-8pt]
			\item Output ordered constraint hyperedge 
			\[
			\Big\{\pi_{i_1}(B'_{i_1},x'_{i_1},z'_{i_1}), \pi_{i_2}(B'_{i_2},x'_{i_2},z'_{i_2}),\ldots,\pi_{i_r}(B'_{i_r},x'_{i_r},z'_{i_r})\Big\}
			\]  
		\end{enumerate}
	\end{mdframed}
	\caption{SSEH to $\mu$-Constrained CSP Reduction}
	\label{fig:test-1}
\end{figure}

We instantiate the various parameters used in the above reduction and its analysis:
\begin{itemize}
	\item $\beta = \mu^{4r}/|V_{\rm gap}|^4$.
	\item $R = 1/(r\beta\delta)$.
	\item $\nu = s/10^r$.
	\item $\gamma = 2^{-10 R}\nu^2$.
	\item $\tau = (s\mu/r)^{100r^2\log(1/\gamma)/\eta\beta}$
	\item $\eta = \min\{\beta^2/r, \nu\}$.
	\item $\epsilon = \frac{\beta^2\nu^4\eta^4\tau^6}{2^{24}r^4}$.
	\item $M = 1/\sqrt{\epsilon} = 	\frac{2^{12}r^2}{\beta\nu^2\eta^2\tau^3}$.
	\item $\kappa = \beta/\log(1/\gamma)$.
\end{itemize}

In the following sections, we analyze the completeness and soundness guarantees of the above reduction.

%% file: completeness-lift.tex
\section{Completeness}						\label{sec:comp}

Suppose $G = (V,E)$ is a YES instance. Then from Definition \ref{defn:sseh}, there exists a subset $S \subseteq V$ satisfying ${\sf Vol}(S) = \delta$ and $\phi_G(S) \leq \epsilon$. We begin by defining a map $i^* : V^R \times \{\bot,\top\}^R \to [R]$ using this choice of $S$, which will be used to identify the dictator labeling. For every $(A,z) \in V^R \times \{\bot,\top\}^R$, define the set 
\[
\Pi(A,z) := \{i \in [R]~|~(A(i),z(i)) \in S \times \{\top\}\}.
\]
Then we use the following process to construct the map $i^*$: 
\begin{itemize}
	\item[I] Firstly, consider the choices of $(A,z) \in V^R \times \{\bot,\top\}^R_\beta$ satisfying $|\Pi(A,z)| = 1$. For any such $(A,z)$, we let $i^*(A,z) = i$ where $i$ is the {\em unique} element contained in $\Pi(A,z)$. Note that by construction, we have $i^*(\pi(A,z)) = \pi(i^*(A,z))$, for every permutation $\pi$, for every choice of $(A,z)$ satisfying $|\Pi(A,z)| = 1$.
	\item[II]  For the remaining choices of $(A,z)$ (i.e., $|\Pi(A,z)| \neq 1$), we assign $i^*$ using the following process:
	\begin{itemize}
		\item While there exists $(A,z)$ such $i^*(A,z)$ is unassigned, do the next step.
		\item Pick any $(A,z)$ such that $i^*(A,z)$ is unassigned, and assign $i^*(A,z) = 1$. Furthermore, assign $i^*(\pi(A,z)) = \pi(1)$ for every non-identity permutation $\pi$. 
	\end{itemize}
	In other words, we fix $i^*(A,z)$ for a given unassigned $(A,z)$ pair, and then use that to determine the indices for the remaining pairs in the orbit of $(A,z)$. Note that this construction ensures that $i^*(\pi(A,z)) = \pi(i^*(A,z))$ holds for every permutation $\pi$, for all the pairs assigned in II. 
\end{itemize}

Then for every $(A,z) \in V^R \times \{\bot,\top\}^R$, we assign the dictator function $f_{A,z} := \chi_{i^*(A,z)}$ i.e., 
\[
f_{A,z}(x) = x(i^*(A,z)) \qquad\qquad\forall~x \in \{0,1\}^R
\]
Finally, the full assignment for the instance $f:V^R \times \{0,1\}^R \times \{\bot,\top\}^R \to \{0,1\}$ is defined as $f(A,x,z) := f_{A,z}(x) = x(i^*(A,z))$, for every $(A,x,z)$.

\begin{observation}			\label{obs:perm}
By construction, the above assignment satisfies the property $f(\pi(A,x,z)) = f(A,x,z)$ for every permutation $\pi$ and every choice of $(A,x,z)$.
\end{observation}

{\bf Analysis of assignment}. Now we analyze the assignment $f$. To begin with, observe that the relative weight of the set indicated by $f$ is
\begin{align*}
\Ex_{(A,x,z) \sim \cH}\left[f(A,x,z)\right] 
&= \Ex_{i \sim G_{\rm gap}}\Ex_{(A,z) \sim V^R \times \{\bot,\top\}^R_\beta}\Ex_{x \sim \{0,1\}^R_{\mu_i}}\left[f(A,x,z)\right] \\
&= \Ex_{i \sim G_{\rm gap}}\Ex_{(A,z) \sim V^R \times \{\bot,\top\}^R_\beta}\Ex_{x \sim \{0,1\}^R_{\mu_i}}\left[x(i^*(A,z))\right] \\  
&= \Ex_{i \sim G_{\rm gap}}\Ex_{(A,z) \sim V^R \times \{\bot,\top\}^R_\beta}\left[\mu_i\right] \\ 
&= \mu, 
\end{align*} 
i.e., $f$ is a feasible assignment for the instance $\cH$. Next, we shall bound the weight of constraints satisfied by $f$. Towards that, we shall need the following key lemma. 

\begin{lemma}			\label{lem:comp-main}
For any fixed edge $e = (j_1,\ldots,j_r) \in E_{\rm gap}$, let $p_e := \Pr_{X_e \sim \theta_e}\left[\psi(X_e) = 1\right]$ denote the probability of local distribution $\theta_e$ satisfying $e$. Then, under the test distribution for the fixing of $e$ we have
	\begin{equation}				\label{eqn:comp-main}
	\Pr_{(B'_j,x'_j,z'_j)_{j \in e}}\left[ \left(f(B'_{j},x'_j,z'_j) \right)_{j \in e} \in \psi^{-1}(1) \right] \gtrsim \rho^2\cdot p_e - O(r\eta),
	\end{equation}
	where for a realization of $(B'_j,x'_j,z'_j)_{j \in e}$, $\left(f(B'_j,x'_j,z'_j)\right)_{j \in e}$ denotes the $r$-ary ordered tuple 
	\[
	(f(B'_{j_1},x'_{j_1},z'_{j_1}),f(B'_{j_2},x'_{j_2},z'_{j_2}),\ldots,f(B'_{j_r},x'_{j_r},z'_{j_r})).
	\]
\end{lemma}
We defer the proof of Lemma \ref{lem:comp-main} for now and continue with the completeness analysis. We can bound the total weight of constraints satisfied by the assignment $f$ as  
\begin{align*}
	&\Ex_{e \sim E_{\rm gap}}\Pr_{(B'_j,x'_j,z'_j)_{j \in e}}\Pr_{\{\pi_j\}_{j \in e}}\left[\left( f\Big(\pi_j(B'_j,x'_j,z'_j)\Big)\right)_{j \in e} \in \psi^{-1}(1)\right] \\
	&= \Ex_{e \sim E_{\rm gap}}\Pr_{(B'_j,x'_j,z'_j)_{j \in e}}\left[ \left(f(B'_j,x'_j,z'_j)\right)_{j \in e} \in \psi^{-1}(1)\right] 
	\tag{Using Observation \ref{obs:perm}}\\
	&\gtrsim \Ex_{e \sim E_{\rm gap}}\left[\frac{\rho^2}{r}\cdot p_e - O(r \eta)\right]		\tag{Using Lemma \ref{lem:comp-main}} \\ 
	&= \Ex_{e \sim E_{\rm gap}}\big[\frac{\rho^2}{r}\cdot p_e \big] - O(r \eta) \\
	&= \frac{\rho^2}{r}\Ex_{e \sim E_{\rm gap}}\Pr_{X_e \sim \theta_e}\left[\psi(X_e) = 1\right] - O(r \eta)		\tag{Using definition of $p_e$} \\
	&\geq \frac{\rho^2}{r}\cdot c - O(r \eta),
\end{align*} 
where the last step uses the completeness value of the gap instance $(G_{\rm gap},\theta)$.

\input{comp-lem}

%% file: comp-lem.tex
\subsection{Proof of Lemma \ref{lem:comp-main}}

The proof of the above lemma follows along the lines of the completeness analysis of D$k$SH (Section $9.2$) from \cite{GL22arXiv}. Without loss of generality, assume $e = [r]$. We introduce some additional notations and definitions that are used in the proof. Define the set $V_{\rm good} \subseteq V^R \times \{\bot,\top\}^R$ as 
\[
V_{\rm good} := \Big\{ (\hat{A},\hat{z}) \in V^R \times \{\bot,\top\}^R~\Big|~|\Pi(\hat{A},\hat{z})| = 1 \Big\}.
\]
In other words, $V_{\rm good}$ is the set of $(A,z)$ pairs for which $i^*(A,z)$ was assigned in Step I in the definition of $i^*$. The following useful observation follows directly from the definition.
\begin{observation}					\label{obs:v-good}
	Suppose $(\hat{A},\hat{z}) \in V_{\rm good}$, and let $\hat{i} = i^*(\hat{A},\hat{z})$. Then $\hat{A}(\hat{i}) \in S$ and $\hat{z}(\hat{i}) = \top$.
\end{observation}
\begin{proof}
	Since $(\hat{A},\hat{z}) \in V_{\rm good}$, we have that $\Pi(\hat{A},\hat{z}) = \{\hat{i}\}$ for some $\hat{i} \in [R]$, and hence using the definition of $i^*$ we must have $\hat{i} = i^*(\hat{A},\hat{z})$. The claim now follows due to the definition of $\Pi(\hat{A},\hat{z})$.
\end{proof}
We also define some events that will be useful in the analysis:
\begin{gather*}
	\cE_{\rm good} := \Big\{(A,z) \in V_{\rm good} \Big\}, \\
	\cE_1 := \Big\{\xi_e(i^*(A,z)) = 1\} \wedge \cE_{\rm good},	 \\
	\cE_2 := \Big\{ \forall j \in [r], \Pi(B'_j,z'_j) = \Pi(A,z) \Big\}.
\end{gather*}
We shall use the above events to condition on and then bound the probability of the event $\left(f(B'_j,x'_j,z'_j) \right)_{j \in [r]} \in \psi^{-1}(1)$. To that end, we observe that
\begin{align}
	&\Pr_{\{(B'_j,x'_j,z'_j)\}_{j \in [r]}}
	\left[\left(f(B'_j,x'_j,z'_j)\right)_{j \in [r]} \in \psi^{-1}(1)\right] 		\non\\
	&\geq \Pr\left[\cE_1 \right]\Pr\left[\cE_2 \Big| \cE_{1}\right] 
	\Pr\left[\left(f(B'_j,x'_j,z'_j) \right)_{j \in [r]} \in \psi^{-1}(1) \Big|~\cE_1 \wedge \cE_{2}\right]				\label{eqn:comp-1} 
\end{align}
We lower bound the probability of the various events in the above expression one-by-one. To begin with, observe that 
\[				
\Pr\Big[\cE_{\rm good}\Big] = \sum_{i \in [R]}\Pr\left[\Pi(A,z) = \{i\}\right] = \sum_{i \in [R]} \beta \delta\left(1 - \beta \delta\right)^{R-1} \geq \frac{e^{-1}}{r},
\]
where the last inequality uses $R = 1/(\beta\delta r)$ from our choice of parameters. Since $(A,z)$ and $\xi_e$ are independent, then using the above bound and the fact that $\xi_e \sim\{0,1\}^R_{\rho^2}$, we get that
\begin{equation}					\label{eqn:comp-2}
	\Pr\left[\cE_1\right] = \Pr_{(A,z)}\left[\cE_{\rm good}\right] \Pr_{\xi_e}\left[\xi_e(i^*(A,z)) = 1 \Big| (A,z)\right] \geq \frac{e^{-1}}{r}\cdot \rho^2,
\end{equation}
For the second term of \eqref{eqn:comp-1}, we have the following lemma which is the main technical component of the proof.
\begin{lemma}				\label{lem:ix-eq}
	Under the test distribution we have,
	\[
	\Pr_{(A,x,z),\{(B'_j,x'_j,z'_j)\}^r_{j = 1}}\left[\forall j \in [r],~\Pi(B'_j,z'_j) = \Pi(A,z) \Big|~\cE_{1} \right] \geq
	e^{-5}.
	\]
\end{lemma}
\begin{proof}
	Let us denote $\cA := (A,z)$, and use $\cB$ to denote the collection of random variables $\{(B_1,z'_1),\ldots,(B_r,z'_r)\}$. To begin with, observe that for every $j \in [r]$, we have $\Pi(B'_j,z'_j) = \Pi(B_j,z'_j)$ (for e.g., see Claim 9.5~\cite{GL22arXiv}), and hence in the rest of the proof, we will simply bound the probability of the event that $\Pi(B_j,z'_j) = \Pi(A,z)$ holds for every $j \in [r]$. To that end, for a fixing of $(A,z) \in V^R \times \{\bot,\top\}^R$, let us define the sets $Q_0,Q_1,Q_2 \subseteq [R]$ as 
	\begin{align*}
		1.& ~ Q_0:= \Big\{i \in [R] \Big|~(A(i),z(i)) \in S \times \{\bot\}\} \\
		2.& ~ Q_1:= \Big\{i \in [R] \Big|~(A(i),z(i)) \in S^c \times \{\top\}\} \\
		3.& ~ Q_2:= \Big\{i \in [R] \Big|~(A(i),z(i)) \in S^c \times \{\bot\}\}.
	\end{align*}
	By definition it follows that (without conditioning on $\cE_1$) we have
	\[
	\Ex_{(A,z)}\left[|Q_0|\right] \leq R\cdot\delta = \frac{1}{r \beta} 
	\ \ \ \ \ \ \textnormal{and} \ \ \ \ \ \ 
	\Ex_{(A,z)}\left[|Q_1|\right] \leq R \cdot \beta = \frac{1}{r\delta}.
	\]
	Let $\cE'$ be the event that $\cQ$ satisfies $|Q_0| \leq 4/(r\beta)$ and $|Q_1| \leq 4/(r\delta)$. Then, using Chernoff bound we have that 
	\[
		\Pr_{(A,z)}\left[\neg\cE'\right] \leq \Pr_{(A,z)}\left[|Q_0| \geq \frac{4}{r\beta}\right] + \Pr_{(A,z)}\left[|Q_1| \geq \frac{4}{r\delta}\right] \leq e^{-2/r\beta} + e^{-2/r\delta} \leq 2e^{-2/r\beta},
	\]
	where the last inequality follows due to $\delta \leq \epsilon \ll \beta$. Using the above bound and \eqref{eqn:comp-2}, we then have 
	\[
	2e^{-2/(r\beta)} \geq \Pr_{(A,z)}\left[\neg \cE'\right] \geq \Pr_{(A,z)}\left[\cE_{1} \right]\Pr_{(A,z)}\left[\neg \cE' \Big|~\cE_1 \right] \geq \frac{e^{-1}\rho^2}{r}\Pr_{(A,z)}\left[\neg \cE' \Big|~\cE_1 \right],
	\]
	which on rearranging gives us that 
	\begin{equation}				\label{eqn:pr-val}
		\Pr_{(A,z)}\left[\neg \cE' \Big|~\cE_{1}\right] \leq e^{-1/(r\beta)}.
	\end{equation}
	
	Now we condition on the events $\cE'$ and $\cE_1 = \{\xi_e(i^*(A,z)) = 1\} \wedge \cE_{\rm good}$, and fix a realization of $\{i^*(A,z),Q_0,Q_1,Q_2\}$ satisfying $\cE'$ and $\cE_1$. Note that conditioned on these events, we have that $\Pi(A,z) = \{i^*(A,z)\}$ and furthermore, the sets $Q_0,Q_1,Q_2$, and the index $i^*(A,z)$ together form a partition of $[R]$. We denote this partition as $\cQ:= (\{i^*\},Q_0,Q_1,Q_2)$. For every fixing of $\cQ$ satisfying $\cE_1$ and $\cE'$, we bound the probability of the event that $\Pi(B_j,z'_j) = \Pi(A,z)$ holds for every $j \in [r]$. This entails lower bounding the probability of the following events:
	\begin{itemize}
		\item Denoting $i_A := i^*(A,z)$, for every $j \in [r]$, we have $(B_j(i_A),z'_j(i_A)) \in S \times \{\top\}$.
		\item For every $\ell \in \{0,1,2\}$, for every $i \in Q_\ell$ and $j \in [r]$ , we have $(B_j(i),z'_j(i)) \notin S \times \{\top\}$.
	\end{itemize}
	We handle these events one-by-one in the following claims.

	\begin{claim}				\label{cl:c1}
		For a fixed $\cQ$ satisfying the events $\cE_1$ and $\cE'$ we have
		\[
		\Pr_{\cA,\cB|\cQ}\left[\forall j \in [r], (B_j(i_A),z'_j(i_A)) \in S \times \{\top\}\right] \geq 1 - 2r(\epsilon + \eta).
		\]
	\end{claim}
	\begin{proof}
	Note that under the conditioning we have $(A,z) \in V_{\rm good}$, and hence using Observation \ref{obs:v-good} we have that $z(i_A) = \top$ (Observation \ref{obs:v-good}). Furthermore, conditioning on $\cE_1$ implies that $\xi_e(i_A) = 1$, which in turn implies that $z_j(i_A) = z(i_A)$ for every $j \in [r]$ (from Step 4b. of Figure \ref{fig:test-1}). Putting these together, we get that
	\[
	\Pr_{(z_j(i_A))^r_{j = 1}|z(i_A) = \top}\left[\forall j \in [r], z_j(i_A) = \top\right]  = 1.
	\]
	Furthermore, for every $j \in [r]$, observe that $z'_j(i_A)$ is an independent $(1 - \eta)$-correlated copy of $z_j(i_A)$ and hence
	\begin{equation}				\label{eqn:bound}
	\Pr_{(z'_j(i_A))^r_{j = 1}|z(i_A)} = \top\left[\forall j \in [r],~z'_j(i_A) = \top\right] \geq \Pr_{(z_j(i_A))^r_{j = 1}|z(i_A)} = \top\left[\forall j \in [r],~ z_j(i_A) = \top\right] - r\eta = 1 - r\eta.
	\end{equation}
	On the other hand, fixing $\cQ$, $A(i_A)$ is a uniformly random vertex in $S$ and $B_1(i_A),\ldots,B_r(i_A)$ are all sampled using $1$-step random walks from $A(i_A)$ in the graph $G_\eta$. Using these observations we can bound:
	\begin{align}
		&\Pr_{\cA,\cB|\cQ}\left[\forall j \in [r],~(B_j(i_A),z'_j(i_A)) \in S \times \{\top\}\right] 		\non\\
		&=\Pr_{\{B_j(i_A)\}_j \sim G_{\eta}(A(i_A))}\left[\forall j \in [r],~ B_j(i_A) \in S \Big| A(i_A) \in S\right] \Pr_{z'_j(i_A)|z(i_A) = \top }\left[\forall j \in [r],~ z'_j(i_A) = \top\Big| z(i_A) = \top\right] 		\non\\
		&\overset{1}{\geq}\left(1 - r(\epsilon + \eta)\right)\left(1 - r\eta\right) 		\non\\
		&\geq 1 - 2r(\epsilon + \eta),			\non
	\end{align}
	where in step $1$, the bound on the first term uses the fact that the expansion of the set $S$ in the graph $G_\eta$ is at most $\epsilon + \eta$, and the second term uses the bound from \eqref{eqn:bound}.
	\end{proof}

	\begin{claim}					\label{cl:c2}
		For a fixed $\cQ$ satisfying the events $\cE_1$ and $\cE'$ we have
		\[
		\Pr_{\cA,\cB|\cQ}\left[\forall i \in Q_0,\forall j \in [r],~(B_j(i),z'_j(i)) \notin S \times \{\top\}\right] \geq e^{-4}.
		\]
	\end{claim}
	\begin{proof}	
	Observe that for any fixed $i \in Q_0$ we have $z(i) = \bot$, and for $j \in [r]$, fixing $z(i)$, marginally $z'_j(i)$ is a $\rho':=\rho^2(1 - \eta)$ correlated copy of $z(i)$. Hence,
	\begin{align}
		&\Pr_{\cA,\cB|\cQ}\left[\forall i \in Q_0,\forall j \in [r],~(B_j(i),z'_j(i)) \notin S \times \{\top\}\right] 		\non\\
		&\geq \Pr_{\cA,\cB|\cQ}\left[\forall i \in Q_0,\forall j \in [r],~ z'_j(i) = \bot \right] 		\non\\
		&= \prod_{i \in Q_0} \Pr_{\cA,\cB|\cQ}\left[\forall j \in [r],~ z'_j(i) = \bot \right] 		\non\\
		&\geq \prod_{i \in Q_0} \left(1 - \sum_{j \in [r]} \Pr_{z'_j(i) \underset{\rho'}{\sim}z(i)} \left[z'_j(i) = \top \Big| i \in Q_0 \right] \right) 		\non\\
		&\overset{1}{=} \prod_{i \in Q_0} \left(1 - \sum_{j \in [r]} \Pr_{z'_j(i) \underset{\rho'}{\sim}z(i)} \left[z'_j(i) = \top \Big| z(i) = \bot \right] \right) \non \\
		&\overset{2}{=} \prod_{i \in Q_0}\left(1 - r(1 - \rho')\beta\right)			\non \\
		&\overset{3}{\geq} (1 - r(1 - \rho')\beta)^{4/r\beta}		\non \\
		&\geq e^{-4} 
		\non
	\end{align}
	where step $1$ follows using the definition of $Q_0$, step $2$ follows from the distribution of $z'_j(i)$ conditioned on $z(i) = \bot$, and step $3$ follows using that the bound $|Q_0| \leq 4/r\beta$ holds conditioned on $\cE'$.
	\end{proof}

	\begin{claim}				\label{cl:c3}
		For a fixed $\cQ$ satisfying the events $\cE_1,\cE'$ we have
		\[
		\Pr_{\cA,\cB|\cQ}\left[\forall i \in Q_1,\forall j \in [r],~(B_j(i),z'_j(i)) \notin S \times \{\top\}\right] \geq 1 - 8(\epsilon + \eta).
		\]
	\end{claim}
	\begin{proof}
	Here we observe that fixing $\cQ$, for every $i \in Q_1$, $A(i)$ is a uniformly random vertex in $S^c$. Hence we can proceed as follows:
	\begin{align}
		&\Pr_{\cA,\cB|\cQ}\left[\forall i \in Q_1,\forall j \in [r],~(B_j(i),z'_j(i)) \notin S \times \{\top\}\right] 			\non\\
		&\geq \Pr_{\cA,\cB|\cQ}\left[\forall i \in Q_1,\forall j \in [r],~B_j(i) \in S^c\right] 			\non\\
		&\geq 1 - \sum_{i \in Q_1} \sum_{j \in [r]}\Pr_{B_j(i) {\sim} G_\eta(A(i))} \left[B_j(i) \in S \Big| i \in Q_1\right] 		\non\\ 
		&\overset{1}{=} 1 - \sum_{i \in Q_1} \sum_{j \in [r]}\Pr_{B_j(i) {\sim} G_\eta(A(i))}\left[B_j(i) \in S \Big| A(i) \in S^c\right] 	\non\\ 
		&\overset{2}{\geq} 1 - |Q_1| r \cdot 2(\eta + \epsilon)\delta 		\non\\
		&\overset{3}{\geq} 1 - 8(\eta + \epsilon),				\non
	\end{align}
	where step $1$ again follows using the definition of $Q_1$. In step $2$, we use that 
	\[
	\Pr_{B_j(i) \sim G_{\eta}(A_{j}(i))}\left[B_j(i) \in S \Big| A_j(i) \in S^c\right]
	\leq \frac{(\epsilon + \eta)\delta}{(1 - \delta)} \leq 2(\epsilon + \eta)\delta,
	\]
	where the first inequality follows using the observation that the expansion of $S$ in $G_\eta$ is at most $\epsilon + \eta$. In step $3$, we use the bound on $|Q_1| \leq 4/r\delta$ from the conditioning on the event $\cE'$.
	\end{proof}

	\begin{claim}				\label{cl:c4}
		For a fixed $\cQ$ satisfying the events $\cE_1,\cE'$ we have
		\[
		\Pr_{\cA,\cB|\cQ}\left[\forall i \in Q_2,\forall j \in [r],~(B_j(i),z'_j(i)) \notin S \times \{\top\}\right] \geq 1 - 2(\epsilon + \eta).
		\]
	\end{claim}
	\begin{proof}
	This case follows by combining the arguments for $Q_0$ and $Q_1$ indices:
	\begin{align}
		&\Pr_{\cA,\cB|\cQ}\left[\forall i \in Q_2,\forall j \in [r],~(B_j(i),z'_j(i)) \notin S \times \{\top\}\right] 		\non\\
		&\geq 1 - \sum_{i \in Q_2} \sum_{j \in [r]}\Pr_{B_j(i) {\sim} G_\eta(A(i))}\Pr_{z'_j(i) \underset{\rho'}{\sim}z(i)} \left[(B_j(i),z'_j(i)) \in S \times \{\top\}  \Big| i \in Q_2\right] 			\non\\
		&\overset{1}{=} 1 - \sum_{i \in Q_2} \sum_{j \in [r]}\Pr_{B_j(i) {\sim} G_\eta(A(i))} \left[B_j(i) \in S \Big| A(i) \in S^c\right] \cdot \Pr_{z'_j(i) \underset{\rho'}{\sim} z(i)} \left[z'_j(i) = \top \Big| z(i) = \bot\right] 		\non\\
		&\overset{2}{\geq} 1 - |Q_2| r \cdot 2(\epsilon + \eta)\delta \cdot (1  - \rho') \beta 		\non\\
		&\overset{3}{\geq} 1 - 2(\eta + \epsilon),					\non
	\end{align}
	where in step $1$, we use the definition of $Q_2$, as well as the fact that for every $i$, $B_j(i) \sim G_{\eta}(A(i))$ and $z'_j(i) \sim z(i)$ are drawn independent of each other. In step $2$, we bound the two probability terms using the ``step-$2$'' arguments from the analysis of $Q_0$ and $Q_1$, and step $3$ uses the elementary bound $|Q_2| \leq R = 1/r\beta\delta$.
	\end{proof}

	{\bf Putting Things together}. Now we put together the bounds for the probabilities of the various events. For brevity, denote $\cE'' := \cE' \wedge  \cE_{1}$. Then,
	\begin{align*}
	&\Pr_{\cA,\cB}\left[\forall j \in [r],~\Pi(B_j,z'_j)  = \Pi(A,z) \Big|~\cE_1 \right] \\
	&\geq \Pr_{\cA,\cB}\left[\forall j \in [r],~~\Pi(B_j,z'_j)  = \Pi(A,z) \Big|~\cE''\right] - \Pr\left[\neg \cE' \Big|~\cE_{1} \right] \\
	&\overset{1}{\geq}  \Ex_{\cQ|\cE''}\Pr_{\cA,\cB|\cQ}\left[\forall j \in [r],~~\Pi(B_j,z'_j)  = \Pi(A,z)\Big|~\cE'',\cQ\right] - e^{-1/r\beta} \\
	&\overset{2}{=} \Ex_{\cQ|\cE''}\Bigg[\Pr_{\cA,\cB|\cQ}\left[\forall j \in [r], (B_j(i_A),z'_j(i_A)) \in S \times \{\top\}\Big|~\cE'',\cQ\right] \cdot \\
	&\qquad\qquad\qquad\prod^2_{\ell = 0} \Pr_{\cA,\cB|\cQ}\left[\forall i \in Q_\ell, j \in [r], (B_j(i),z'_j(i)) \notin S \times \{\top\}\right] \Bigg] - e^{-1/r\beta} \\
	&\overset{3}{\geq} (1 - 2r(\epsilon + \eta)) \cdot e^{-4} \cdot(1 - 8(\epsilon + \eta))\cdot(1 - 2(\epsilon + \eta)) - e^{-1/r\beta} \\
	&\overset{4}{\geq} e^{-5},
	\end{align*}
	where step $1$ follows using the bound from \eqref{eqn:pr-val}. In step $2$, we use the following observation: for a fixing of $\cQ$, the distribution of rows across the sets $\{i_A\}, Q_0,Q_1,Q_2$ are all independent. Step $3$ uses the bounds from Claims \ref{cl:c1}, \ref{cl:c2}, \ref{cl:c3}, and \ref{cl:c4}. Step $4$ follows from our choices of parameters for $\beta,\eta$, and $\epsilon$. This concludes the proof of Lemma \ref{lem:ix-eq}.
\end{proof}

Continuing with the proof of Lemma \ref{lem:comp-main}, we now bound the third probability term from \eqref{eqn:comp-1}. For brevity, we introduce another event $\cE_3 := \cE_1 \wedge \cE_{2}$. Then, randomizing over the choice of $\cV := (A,z),(B'_1,x'_1,z'_1),\ldots,(B'_r,x'_r,z'_r)$, we have
\begin{align}
	&\Pr_{\cV}\bigg[~\psi\Big(f(B'_1,x'_1,z'_1),\ldots,f(B'_2,x'_2,z'_2)\Big) = 1 \Big|~\cE_1 \wedge \cE_{2}\bigg]		\non \\ 		
	&\overset{1}{=} \Pr_{\cV}\bigg[~\psi\Big(x'_1(i^*(B'_1,z'_1)),\ldots,x'_r(i^*(B'_r,z'_r))\Big) = 1 \Big|~\cE_3\bigg]		\non\\
	&\overset{2}{=} \Pr_{\cV}\bigg[~\psi\Big(x'_1(i^*(A,z)),\ldots,x'_r(i^*(A,z))\Big) = 1 \Big|~\cE_3\bigg] 			\non\\
	&\overset{3}{=} \Pr_{\cV}\bigg[~\psi\Big(\tilde{x}_1(i^*(A,z)),\ldots,\tilde{x}_r(i^*(A,z))\Big) = 1 \Big|~\cE_3\bigg]	\non\\
	&\overset{4}{\geq} \Pr_{\cV}\bigg[~\psi\Big({x}_1(i^*(A,z)),\ldots,{x}_r(i^*(A,z))\Big) = 1 \Big|~\cE_3\bigg] - r \eta	\non\\
	&\overset{5}{=}  p_e - r \eta.			\label{eqn:comp-3}
\end{align}
	 
We justify that each of the above steps hold conditioning on $\cE_3 = \cE_1 \wedge \cE_2$. For step $1$, we use the definition of the assignment $f$. Step $2$ uses the following argument: since we condition on $\cE_2$, we have $\Pi(B'_j,z'_j) = \Pi(A,z)$ for every $j \in [r]$. This implies the following for every $j$:
\begin{itemize}
	\item[(i)] Since $(A,z) \in V_{\rm good}$ (due to conditioning on $\cE_1$), we have $\Pi(B'_j,z'_j) = \Pi(A,z) = \{i^*(A,z)\}$, and hence $(B'_j,z'_j) \in V_{\rm good}$. 
	\item[(ii)] Again, since (i) implies that for every $j \in [r]$, we have $|\pi(B'_j,z'_j)| = 1$, we have $i^*(B'_j,z'_j) = i^*(A,z)$.
\end{itemize}
For step $3$, we build upon the above observations and further note that for every $j \in [r]$ we can argue the following:
\begin{itemize}
	\item[(iii)] Since $(B'_j,z'_j) \in V_{\rm good}$ and $i^*(B'_j,z'_j) = i^*(A,z)$, using Observation \ref{obs:v-good} we have $z'_j(i^*(A,z)) = \top$.
	\item[(iv)] Since under the test distribution we have $(B'_j,x'_j) \sim M^{(\mu_j)}_{z'_j}(B_j,\tilde{x}_j)$, this along with (iii) implies that $x'_j(i^*(A,z)) = \tilde{x}_j(i^*(A,z))$ with probability $1$.
\end{itemize}
Step $4$ follows by observing that for every $j \in [r]$, $\tilde{x}_j(i)$ is a $(1 - \eta)$-correlated copy of $x_j(i)$, and hence the event $(\tilde{x}_j(i^*(A,z))) = x_j(i^*(A,z))$ holds for every $j \in [r]$ with probability least $1 - r\eta$. Finally, Step $5$ uses that since  $(x_j(i^*(A,z))_{j \in [r]}$ is sampled from $\theta_{[r]}$, it is an accepting string with probability $p_e$. 

{\bf Finishing the Proof}. We plug in the bounds from \eqref{eqn:comp-2}, Lemma \ref{lem:ix-eq}, and \eqref{eqn:comp-3} into \eqref{eqn:comp-1}, and get that 
\begin{align*}
&\Pr_{\{(B'_j,x'_j,z'_j)\}_{j \in [r]}} \left[\psi\Big(f(B'_1,x'_1,z'_1),\ldots,f(B'_2,x'_2,z'_2) \Big) = 1\right] \\
& \geq \Pr\left[\cE_1 \right]\Pr\left[\cE_2 \Big| \cE_{1}\right] 
\Pr\left[\left(f(B'_j,x'_j,z'_j) \right)_{j \in [r]} \in \psi^{-1}(1) \Big|~\cE_1 \wedge \cE_{2}\right]		\tag{Using \eqref{eqn:comp-1}} \\
&\geq \frac{\rho^2}{e\cdot r}\cdot e^{-5} \cdot (p_e - r\eta) 		\tag{Using \eqref{eqn:comp-2}, Lemma \ref{lem:ix-eq}, \eqref{eqn:comp-3}}	\\
&\geq \frac{\rho^2 p_e}{r}\cdot e^{-6} - O(r\eta),
\end{align*}
which concludes the proof of Lemma \ref{lem:comp-main}.

%% file: soundness-lift.tex
\section{Soundness}					\label{sec:sound}

We introduce some additional notation used in the soundness analysis. Throughout the proof, we will be dealing with functions defined over various probability spaces. We will always use $\Omega$ to denote the underlying $4$-point set $\{0,1\} \times \{\bot,\top\}$, and for any $i \in V_{\rm gap}$, we will use  $\Omega_i$ to denote the probability space $\{0,1\}_{\mu_i} \otimes \{\bot,\top\}_\beta$. Furthermore, our expressions shall feature noise operators acting on various probability spaces. To avoid ambiguity, we introduce the notation for each noise operator used here. The following holds for any $\kappa \in (0,1)$:

\begin{itemize}
	\item $\sT^{(\beta)}_\kappa$ is noise operator on $L_2(\{\bot,\top\}^R_\beta)$ acting on the $z_i$ variables.
	\item For any $i \in V_{\rm Gap}$, we use $\sT^{(\mu_i)}_\kappa$ to denote the noise operators on $L_2(\{0,1\}^R_{\mu_i})$ which acts on the vector variable $x_i$. 
	\item For any $i \in V_{\rm gap}$, we use $\sT^{(\Omega_i)}_\kappa$ to denote the noise operator $\sT^{(\mu_i)}_\kappa \circ \sT^\Pb_\kappa$ on  $L_2(\{0,1\}^R_{\mu_i} \otimes \{\bot,\top\}^R_\beta)$ which acts on the vector variables $x_i$ and $z_i$ independently.  
\end{itemize}

We extend the above notation and use ${\sf Inf}^{(\Omega_i)}_j\left[ \cdot\right]$ and ${\sf Inf}^{(\mu_i)}_j\left[\cdot\right]$ to denote that the influence of a function is being defined with respect to $\Omega^R_i$ and $\mu^R_i$ respectively.

\begin{remark}[Definition of $\sT^{(\Omega_i)}_{1 - \eta}$]
We point out that in the above convention, $\sT^{(\Omega_i)}_{1 - \eta}$ is {\em not} the $(1-\eta)$-correlated Bonami-Beckner operator in the space $L_2(\Omega_i)$. The $(1 - \eta)$-correlation Bonami-Beckner operator on $L_2(\Omega_i)$ re-randomizes each coordinate$(x,z)(i)$ with probability $(1 - \eta)$, whereas in the above definition of $\sT^{(\Omega_i)}_{1 - \eta}$, for every $i \in [R]$, both $x(i)$ and $z(i)$ are each chosen to be re-randomized with probability $(1 - \eta)$ independently. However, it still exhibits Fourier decay properties similar to the Bonami-Beckner operator on $\Omega^R_i$ (see Fact \ref{fact:decay-2}). 
\end{remark}

Finally, we will use $\cD^R_e$ to denote the distribution on the variables $(x_i,z_i)_{i \in e}$ under the test distribution conditioned on the fixing of $e$.

\subsection{First Steps}

Suppose $G$ is a NO instance, and let $f:V^R \times \{0,1\}^R \times \{\bot,\top\}^R \to \{0,1\}$ be a feasible labeling of $\cH$ i.e., it satisfies the global constraint:
\begin{equation}				\label{eqn:bias-constr}
	\Ex_{i \sim G_{\rm gap}} \Ex_{A \sim V^R} \Ex_{(x,z) \sim \{0,1\}^R_{\mu_i} \otimes \{\bot,\top\}^R_\beta}
	\left[f(A,x,z)\right] = \mu.
\end{equation} 
Now, for every $A \in V^R$, and for every $i \in V_{\rm gap}$, we shall find it useful to define the functions $g_{A,i}:\left(\{0,1\} \times \{\bot,\top\}\right)^R \to [0,1]$as 
\begin{equation}				\label{eqn:g-def}
g_{A,i}(x,z) = \Ex_{B \sim G^{\otimes R}_\eta(A)} \Ex_{(B',x') \sim M^{(\mu_i)}_{z}(B,x)} \Ex_{\pi \sim \mathbbm{S}_R} \left[f(\pi(B',x',z))\right],
\end{equation}
i.e., $g_{A,i}$ is the restriction of $f$ to $A$ averaged over the noisy walk on $G^{\otimes R}_\eta$, the $M^{(\mu_i)}_{z'_i}$ operators, and the random permutation $\pi_i$. The following observation on the averaged functions will be useful throughout the various steps of the soundness analysis.
\begin{observation}					\label{obs:prob-space}
	For every $i \in V_{\rm gap}$, the function $g_{A,i}$ is an $R$-variate function in $L_2(\Omega^R_i)$, where $\Omega_i = \{0,1\}_{\mu_i} \otimes \{\bot,\top\}_\beta$. Consequently, the influences of the function $g_{A,i}$ are defined with respect to the Fourier basis corresponding to $\Omega^R_i$.
\end{observation}
Let us introduce some additional notation that will be useful throughout the soundness analysis. Let 
\[
\psi(x) = \sum_{a \in \psi^{-1}(1)} \prod_{j \in S_+(a)} x_j \prod_{j \in S_-(a)} (1 - x_j)
\] 
be the multilinear representation of $\psi$, where for any $a \in \{0,1\}^r$, $S_+(a)$ and $S_-(a)$ are subsets of $[r]$ consisting of the indices with $a(i) = 1$ and $a(i) = 0$ respectively. Next, for any ordered edge $e = (i_1,\ldots,i_r) \in E_{\rm gap}, a \in \psi^{-1}(1)$, and for every $i \in e$, define 
\begin{equation}		\label{eqn:f-defs}
f^{(a,e,i)} = 
\begin{cases}
	f & \mbox{ if } i = i_j, j \in S_+(a) \\
	1 - f & \mbox{ if } i = i_j, j \in S_-(a).
\end{cases}
\end{equation}
Analogously, for every $e = (i_1,\ldots,i_r) \in E_{\rm gap}$, $i \in e$ and $a \in \psi^{-1}(1)$ we define
\[
g^{(a,e,i)}_{A,i} = 
\begin{cases}
	g_{A,i} & \mbox{ if } i = i_j, j \in S_+(a) \\
	1 - g_{A,i} & \mbox{ if } i = i_j, j \in S_-(a).
\end{cases}
\]
As in Observation \ref{obs:prob-space}, note that for any $i \in V_{\rm gap}$, for any choices of  $A,a,e$, the function $g^{(a,e,i)}_{A,i}$ is in $L_2(\Omega^R_i)$. 

{\bf Arithmetization}. We now proceed to arithmetize the fraction of constraints satisfied by the assignment $f$. To that end, let us fix a choice of $e = (i_1,\ldots,i_r)  \in E_{\rm gap}$, and express the fraction of satisfied constraints corresponding to the distribution conditioned on the choice of edge $e$:
\begin{align*}				\label{eqn:sound-expr}
	&\Pr\Big[\mbox{ Test Accepts for fixed } e\Big] \\
	&= \Ex_{A \sim V^R}\Ex_{(\tilde{x}_i,z'_i)_{i \in e}} \Ex_{(B_i)_{i \in e} \sim G^{\otimes R}_{\eta}(A)} 
	\Ex_{\{(B'_i,x'_i) \sim M^{(\mu_i)}_{z'_i}(B_i,\tilde{x}_i)\}_{i \in e}}\Ex_{\{\pi_i\}_{i \in e} \sim \mathbbm{S}_R}
	\left[\psi\Big(f(\pi_{i_1}(B'_{i_1},x'_{i_1},z'_{i_1})),\ldots,f(\pi_{i_r}(B'_{i_r},x'_{i_r},z'_{i_r}))\Big)\right]  \\
	& = \Ex_{A \sim V^R}\Ex_{(\tilde{x}_i,z'_i)_{i \in e}} \Ex_{(B_i)_{i \in e} \sim G^{\otimes R}_{\eta}(A)} 
	\Ex_{\{(B'_i,x'_i) \sim M^{(\mu_i)}_{z'_i}(B_i,\tilde{x}_i)\}_{i \in e}}\Ex_{\{\pi_i\}_{i \in e} \sim \mathbbm{S}_R} \\
	&\qquad\qquad\qquad
	\left[\sum_{a \in \psi^{-1}(1)}\prod_{i = i_j: a(j) = 1} f(\pi_{i}(B'_{i},x'_{i},z'_{i})) 
	 \prod_{i = i_j :a(j) = 0} (1 - f(\pi_i(B'_i,x'_i,z'_i))) \right] \\
	& = \sum_{a \in \psi^{-1}(1)} \Ex_{A \sim V^R}\Ex_{(\tilde{x}_i,z'_i)_{i \in e}} \Ex_{(B_i)_{i \in e} \sim G^{\otimes R}_{\eta}(A)} 
	\Ex_{\{(B'_i,x'_i) \sim M^{(\mu_i)}_{z'_i}(B_i,\tilde{x}_i)\}_{i \in e}}\Ex_{\{\pi_i\}_{i \in e} \sim \mathbbm{S}_R}
	\left[\prod_{i \in e} f^{(a,e,i)}(\pi_{i}(B'_{i},x'_{i},z'_{i})) \right] \\
	&\overset{1}{=} \sum_{a \in \psi^{-1}(1)}\Ex_{A \sim V^R}\Ex_{({x}_i,z_i)_{i \in e}} \left[\prod_{i \in e}\Ex_{\tilde{x}_i \underset{1 - \eta}{\sim}x_i} \Ex_{z'_i \underset{1 - \eta}{\sim} z_i} \Ex_{B_i \sim G^{\otimes R}_{\eta}(A)} \Ex_{(B'_i,x'_i) \sim M^{(\mu_i)}_{z'_i}(B_i,\tilde{x}_i)}\Ex_{\pi_i \sim \mathbbm{S}_R}  f^{(a,e,i)}(\pi_i(B'_i,x'_i,z'_i)) \right] \\
	&\overset{2}{=} \sum_{a \in \psi^{-1}(1)} \Ex_{A \sim V^R}\Ex_{({x}_i,z_i)_{i \in e}} \left[ \prod_{i \in e}\Ex_{\tilde{x}_i \underset{1 - \eta}{\sim}x_i} \Ex_{z'_i \underset{1 - \eta}{\sim} z_i} g^{(a,e,i)}_{A,i}(\tilde{x}_i,z'_i)\right] \\
	&\overset{3}{=} \sum_{a \in \psi^{-1}(1)} \Ex_{A \sim V^R}\Ex_{({x}_i,z_i)_{i \in e}} \left[ \prod_{i \in e}  \sT^{(\Omega_i)}_{1 - \eta} g^{(a,e,i)}_{A,i}(x_i,z_i)
	 \right],
\end{align*}
where step $1$ uses the independence of the various averaging operators, step $2$ follows using the definition of $g_{A,i}$ from \eqref{eqn:g-def} and step $3$ follows from the definition of $\sT^{(\Omega_i)}_{1 - \eta}$. Using the above, the overall weight of constraints satisfied by the assignment indicated by $f(\cdot)$ can be expressed as 
\begin{equation}				\label{eqn:sound-1}
	\Pr\Big[\mbox{ Test Accepts}\Big] = 
	\sum_{a \in \psi^{-1}(1)}\Ex_{e \sim E_{\rm gap}}\Ex_{A \sim V^R}\Ex_{({x}_i,z_i)_{i \in e} \sim \cD^{R}_e} \left[ \prod_{i \in e}  \sT^{(\Omega_i)}_{1 - \eta} g^{(a,e,i)}_{A,i}(x_i,z_i) \right].
\end{equation}	 

Now we proceed with the soundness analysis in steps.

\input{z-avg}

\input{rag-step}

\input{note}

\input{decoding}

%% file: z-avg.tex
\subsection{Averaging out the $z$-variables}				\label{sec:z-avg}

Now in this step, we will decouple the $z_i$ variables and average them out. Formally, we have the following key lemma which states that if for a fixed edge, the corresponding functions have small influences, then for that term, we can replace the correlated sampling of the $z_i$-variables with independent sampling, while only losing small additive factors in the soundness value.
\begin{lemma}					\label{lem:prod}
	Fix an edge $e = (i_1,\ldots,i_r) \in E_{\rm gap}$, and let $(h_i)_{i \in e}:\Omega^R \to [0,1]^e$ be a collection of functions where for every $i \in e$ we have $h_i \in L_2(\Omega^R_i)$, and 
	\[
	\max_{i \in e}\max_{j \in [R]} {\sf Inf}^{(\Omega_i)}_{j}\left[{h_i}\right] \leq \tau. 
	\]
	Assume that $\theta_e$ is $(\gamma\mu)^r$-smooth as in the setting of the reduction i.e., for any $\alpha \in \{0,1\}^e$ we have $\Pr_{X_e \sim \theta_e}\left[X_e = \alpha\right] \geq (\gamma\mu)^r$. Furthermore, suppose ${\rm Var}\big[h^{>d}_i\big] \leq (1 - \eta)^d$ for every $d \geq \frac{1}{18}\log\frac1\tau/\log\frac{1}{\eta^r\gamma^r}$. Then,
	\[
	\Ex_{(x_i,z_i)_{i \in e} \sim \cD^{R}_e}\left[\prod_{i \in e} h_i(x_i,z_i)\right]
	\leq 2^r \cdot \Ex_{(x_i)_{i \in e} \sim \theta^R_e} \left[\prod_{i \in [r]} \bar{h_i}(x_i)\right] + \mu^r  + C'_r\cdot\sqrt{r}\tau^{C\eta\kappa/r^2}, 
	\]
	where for every $i \in e$, $\bar{h_i}:\{0,1\}^R \to [0,1]$ is the function in $L_2(\{0,1\}^R_{\mu_i})$ defined as $\bar{h_i}(x) \defeq \Ex_{z \sim \{\bot,\top\}^R_\beta} \left[h_i(x,z)\right]$, and $\kappa = \beta/\log(1/\gamma)$. 
\end{lemma}

We defer the proof of the above lemma to Section \ref{sec:prod} for now. Using the above lemma, we will show that for most choices of $A \in V^R$, we can average out the $z_i$-variables in the corresponding inner expectation term with negligible losses. To that end, we define the set $V_{\rm nice} \subset V^R$ as  
\begin{equation}					\label{eqn:vnice-def}
	V_{\rm nice} := \left\{A \in V^R \Big| \Pr_{e \sim E_{\rm gap}}\left[\max_{i \in e} \max_{j \in [R]} {\sf Inf}^{(\Omega_i)}_j\left[\sT^{(\Omega_i)}_{1 - \eta} g_{A,i}\right] > \tau\right] \leq \nu\right\}
\end{equation} 
In other words, $V_{\rm nice}$ is the set of vertices in $A \in V^R$ for which almost all choices of edges $e$ have the property that all functions in $\{g_{A,i}\}_{i \in e}$ have small influences w.r.t. their respective probability spaces. The following lemma bounds the size of $|V_{\rm nice}|$.
\begin{lemma}			\label{lem:dec-set0}
	Since $G$ is a NO instance, we have
	\[
	\Pr_{A \sim V^R} \Big[A \in V_{\rm nice}\Big] \geq 1 - \nu.
	\]
\end{lemma}
The proof of the above lemma proceeds using the {\em influence-decoding} argument for \smallsetexpansion~\cite{KKMO07,RST12}. We defer its proof to Section \ref{sec:dec-set0} for now and continue with the soundness analysis. As a next step, we use Lemma \ref{lem:prod} to derive the following bound on the inner expectation corresponding to the vertices in $V_{\rm nice}$.

\begin{corollary}					\label{corr:z-avg}
	For every $A \in V_{\rm nice}$ we have that 
	\begin{align*}
	&\sum_{a \in \psi^{-1}(1)}\Ex_{e \sim E_{\rm gap}} \Ex_{(x_i,z_i)_{i \in e} \sim \cD^R_e}\left[\prod_{i \in e} \sT^{(\Omega_i)}_{1 - \eta} g^{(a,e,i)}_{A,i}(x_i,z_i)\right] \\
	&\leq 2^r\sum_{a \in \psi^{-1}(1)}\Ex_{e \sim E_{\rm gap}} \Ex_{(x_i)_{i \in e} \sim \theta^R_e}\left[\prod_{i \in e} \sT^{(\mu_i)}_{1 - \eta} \bar{g}^{(a,e,i)}_{A,i}(x_i)\right] +  2^r\mu^r + 2^{2r}\nu + 2^{2r}C'_r\sqrt{r}\tau^{C\eta\kappa/r^2},
	\end{align*}
	where $\bar{g}^{(a,e,i)}_{A,i}:\{0,1\}^R \to [0,1]$ is the function in $L_2(\{0,1\}^R_{\mu_i})$ defined as $\bar{g}^{(a,e,i)}_{A,i}(x) = \Ex_{z \sim \{\bot,\top\}^R_\beta} \left[g^{(a,e,i)}_{A,i}(x,z)\right]$.
\end{corollary}
\begin{proof}
	Fix a choice of $A \in V_{\rm nice}$, and let $E_{\rm gap}(A) \subseteq E_{\rm gap}$ denote the subset of edges whose vertices identify functions with small influence i.e., 
	\[
	E_{\rm gap}(A) := \left\{ e \in E_{\rm gap} \Big| \max_{i \in e} \max_{j \in [R]} {\sf Inf}^{(\Omega_i)}_j\left[{\sT^{(\Omega_i)}_{1 - \eta}g_{A,i}}\right] \leq \tau\right\}.
	\]
	Then using the definition of $V_{\rm nice}$ we have that
	\begin{equation}				\label{eqn:enice-bd}
		\Pr_{e \sim E_{\rm gap}}\left[e \notin E_{\rm gap}(A) \right] \leq \nu.
	\end{equation}
	We also have the following useful observation.
	\begin{observation}				\label{obs:inf-eq}
		Suppose $e \in E_{\rm gap}(A)$. Then for every $a \in \psi^{-1}(1)$ and $i \in e$ we have 
		\[
		\max_{j \in [R]}{\sf Inf}^{(\Omega_i)}_j\left[\sT^{(\Omega_i)}_{1 - \eta}g^{(a,e,i)}_{A,i}\right] \leq \tau.
		\]
	\end{observation}
	\begin{proof}
		Since $e \in E_{\rm gap}(A)$, by definition of $E_{\rm gap}(A)$ we have that 
		\[
			\max_{j \in [R]} {\sf Inf}^{(\Omega_i)}_j\left[\sT^{(\Omega_i)}_{1 - \eta}g_{A,i}\right] \leq \tau.
		\]
		Furthermore, note that for any $i \in e$, we must have $\sT^{(\Omega_i)}_{1 - \eta}g^{(a,e,i)}_{A,i} = \sT^{(\Omega_i)}_{1 - \eta}g_{A,i}$ or $\sT^{(\Omega_i)}_{1 - \eta}g^{(a,e,i)}_{A,i} = 1 - \sT^{(\Omega_i)}_{1 - \eta} g_{A,i}$, and hence it follows that
		\[
		{\sf Inf}^{(\Omega_i)}_j\left[\sT^{(\Omega_i)}_{1 - \eta}g^{(a,e,i)}_{A,i}\right]
		= {\sf Inf}^{(\Omega_i)}_j\left[\sT^{(\Omega_i)}_{1 - \eta}g_{A,i}\right]
		\]
		for every $j \in [R]$, since the influence of a function is not affected by translation and negation. The claim now follows by combining the two observations from above.
	\end{proof}
	
	Now fix an accepting string $a \in \psi^{-1}(1)$  and fix an edge $e \in E_{\rm gap}(A)$. Then note that the functions $\left\{\sT^{(\Omega_i)}_{1 - \eta}g^{(a,e,i)}_{A,i}\right\}_{i \in e}$ satisfy the premise of Lemma \ref{lem:prod} w.r.t. the probability spaces $(\Omega_i)_{i \in e}$ i.e., 
	\begin{align*}
		(1)&~\max_{i \in e}\max_{j \in [R]}{\sf Inf}^{(\Omega_i)}_{j}\left[\sT^{(\Omega_i)}_{1 - \eta} g^{(a,e,i)}_{A,i}\right] \leq \tau, 
		\tag{Observation \ref{obs:inf-eq}}\\
		(2)&~\max_{i \in e}{\rm Var}\left[\left(\sT^{(\Omega_i)}_{1 - \eta} g^{(a,e,i)}_{A,i}\right)^{>d}\right] \leq (1 - \eta)^d \qquad\qquad\forall~d \geq \frac{1}{18}\log\frac1\tau/\log\frac{1}{\eta^r\gamma^r}.			\tag{Fact \ref{fact:decay-2}}
	\end{align*}
	Therefore, instantiating Lemma \ref{lem:prod} with $h_i := \sT^{(\Omega_i)}_{1 - \eta}g^{(a,e,i)}_{A,i}$ for every $i \in e$ we get that:
	\begin{equation}				\label{eqn:step}
		\Ex_{(x_i,z_i)_{i \in e} \sim \cD^R_e} \left[\prod_{i \in e} \sT^{(\Omega_i)}_{1 - \eta} g^{(a,e,i)}_{A,i}(x_i,z_i)\right]
		\leq 2^r \Ex_{(x_i)_{i \in e} \sim \theta^R_e}\left[\prod_{i \in e} \sT^{(\mu_i)}_{1 - \eta} \bar{g}^{(a,e,i)}_{A,i}(x_i)\right] + \mu^r + C'_r\sqrt{r}\tau^{C\eta\kappa/r^2}.
	\end{equation}
	Note that above bound holds for any $e \in E_{\rm gap}(A)$ and $a \in \psi^{-1}(1)$. Towards finishing the proof, we now observe that 
	\begin{align}
		&\Ex_{e \sim E_{\rm gap}} \Ex_{(x_i,z_i)_{i \in e} \sim \cD^R_e}\left[\prod_{i \in e} \sT^{(\Omega_i)}_{1 - \eta} g^{(a,e,i)}_{A,i}(x_i,z_i)\right] 
			\non\\
		&\overset{1}{\leq} \Ex_{e \sim E_{\rm gap}| e \in E_{\rm gap}(A)} \Ex_{(x_i,z_i)_{i \in e} \sim \cD^R_e}\left[\prod_{i \in e} \sT^{(\Omega_i)}_{1 - \eta} g^{(a,e,i)}_{A,i}(x_i,z_i)\right] + \nu			\non\\
		&\overset{2}{\leq}  2^r\cdot\Ex_{e \sim E_{\rm gap}| e \in E_{\rm gap}(A)} \Ex_{(x_i)_{i \in e} \sim \theta^R_e}\left[\prod_{i \in e} \sT^{(\mu_i)}_{1 - \eta} \bar{g}^{(a,e,i)}_{A,i}(x_i)\right] + C'_r\sqrt{r}\tau^{C\eta\kappa/r^2} + \mu^r + \nu		\non\\
		&\overset{3}{\leq} 2^r\cdot\Ex_{e \sim E_{\rm gap}} \Ex_{(x_i)_{i \in e} \sim \theta^R_e}\left[\prod_{i \in e} \sT^{(\mu_i)}_{1 - \eta} \bar{g}^{(a,e,i)}_{A,i}(x_i)\right] + \mu^r + C'_r\sqrt{r}\tau^{C\eta\kappa/r^2} + 3\nu,			\label{eqn:exp-bd}
	\end{align}
	where in step $1$, we use the fact that $E_{\rm gap}(A)$ has weight at least $(1 - \nu)$ in $E_{\rm gap}$ (Eq. \eqref{eqn:enice-bd}). In step $2$, we use the bound from \eqref{eqn:step} for every fixed choice of $e \in E_{\rm gap}$ and in step $3$, we again use the bound on the weight of $E_{\rm gap}(A)$ from \eqref{eqn:enice-bd}. Finally, summing over all accepting strings $a \in \psi^{-1}(1)$ and applying \eqref{eqn:exp-bd} we get that
	\begin{align*}
		&\sum_{a \in \psi^{-1}(1)}  \Ex_{e \sim E_{\rm gap}} \Ex_{(x_i,z_i)_{i \in e} \sim \cD^R_e} \left[\prod_{i \in e} \sT^{(\Omega_i)}_{1 - \eta} g^{(a,e,i)}_{A,i}(x_i,z_i)\right] \\
		&\leq 2^r \sum_{a \in \psi^{-1}(1)}\Ex_{e \sim E_{\rm gap}} \Ex_{(x_i)_{i \in e} \sim \theta^R_e}\left[\prod_{i \in e} \sT^{(\mu_i)}_{1 - \eta} \bar{g}^{(a,e,i)}_{A,i}(x_i)\right] + 2^r\mu^r + C''_r\sqrt{r}\tau^{C\eta\kappa/r^2} + 3\cdot2^{r}\nu.
	\end{align*}
\end{proof}			
Using the above lemma, we now continue with bounding the expectation term from the RHS of \eqref{eqn:sound-1}:
\begin{align}
	&\sum_{a \in \psi^{-1}(1)}\Ex_{A \sim V^R} \Ex_{e \sim E_{\rm gap}} \Ex_{(x_i,z_i)_{i \in e} \sim \cD^{R}_e}  \left[\prod_{i \in e} \sT^{(\Omega_i)}_{1 - \eta}g^{(a,e,i)}_{A,i}(x_i,z_i)\right] 	 \non\\
	&=\Ex_{A \sim V^R} \left[\sum_{a \in \psi^{-1}(1)} \Ex_{e \sim E_{\rm gap}} \Ex_{(x_i,z_i)_{i \in e} \sim \cD^{R}_e} \prod_{i \in e} \sT^{(\Omega_i)}_{1 - \eta}g^{(a,e,i)}_{A,i}(x_i,z_i)\right] 	 \non\\
	&\overset{1}{\leq} \Ex_{A \sim V_{\rm nice}} \left[\sum_{a \in \psi^{-1}(1)} \Ex_{e \sim E_{\rm gap}} \Ex_{(x_i,z_i)_{i \in e} \sim \cD^{R}_e}\prod_{i \in e} \sT^{(\Omega_i)}_{1 - \eta}g^{(a,e,i)}_{A,i}(x_i,z_i)\right] + 2^r\nu 		\non\\
	&\overset{2}{\leq} 2^r\Ex_{A \sim V_{\rm nice}} \left[\sum_{a \in \psi^{-1}(1)} \Ex_{e \sim E_{\rm gap}} \Ex_{(x_i)_{i \in e} \sim \theta^R_e}\prod_{i \in e} \sT^{(\mu_i)}_{1 - \eta} \bar{g}^{(a,e,i)}_{A,i}(x_i)\right] + 2^r\mu^r + 2^rC'_r\sqrt{r}\tau^{C\eta\kappa/r^2} + 4\cdot2^r\nu 		\non\\
	&\overset{3}{\leq} 2^r\Ex_{A \sim V^R}\left[\sum_{a \in \psi^{-1}(1)} \Ex_{e \sim E_{\rm gap}} \Ex_{(x_i)_{i \in e} \sim \theta^R_e}\prod_{i \in e} \sT^{(\mu_i)}_{1 - \eta} \bar{g}^{(a,e,i)}_{A,i}(x_i)\right] + 2^r\mu^r + 2^rC'_r\sqrt{r}\tau^{C\eta\kappa/r^2} +  6\cdot2^r\nu  \non \\
	&= 2^r\sum_{a \in \psi^{-1}(1)} \Ex_{A \sim V^R} \Ex_{e \sim E_{\rm gap}} \Ex_{(x_i)_{i \in e} \sim \theta^R_e}\left[\prod_{i \in e} \sT^{(\mu_i)}_{1 - \eta} \bar{g}^{(a,e,i)}_{A,i}(x_i)\right] + 2^r\mu^r + 2^rC'_r\sqrt{r}\tau^{C\eta\kappa/r^2} +  6\cdot2^r\nu  \non \\
	 		\label{eqn:invar-2}
\end{align}
where step $1$ uses Lemma \ref{lem:dec-set0} and the fact that the summation term is bounded by $2^r$ with probability $1$. Step $2$ applies Corollary \ref{corr:z-avg} for every fixed choice of $A \in V_{\rm nice}$, and step $3$ again uses the bound on $|V_{\rm nice}|$ from Lemma \ref{lem:dec-set0}.

%% file: rag-step.tex
\subsection{Raghavendra's Rounding Step} 				\label{sec:rag-step}
In this section, our goal would be to upper bound the RHS of \eqref{eqn:invar-2}, i.e.,
\begin{equation}				\label{eqn:sound-val1}	
\sum_{a \in \psi^{-1}(1)}\Ex_{A \sim V^R} \Ex_{e \sim E_{\rm gap}} \Ex_{(x_i)_{i \in e} \sim \theta^R_e} \left[\prod_{i \in e} \sT^{(\mu_i)}_{1 - \eta}\bar{g}^{(a,e,i)}_{A,i}(x_i) \right] + 2^r\mu^r
\end{equation} 
We first begin with the easy observation that we can bound $\mu^r \leq {\sf Opt}_\mu(G_{\rm gap})$, since $\mu \leq 1/2$, and therefore, the expected weight of constraints satisfied by setting every variable to $1$ independently with probability $\mu$ is at least $\mu^r$. On the other hand, upper bounding the expectation term is relatively more challenging. However, as we describe below, we can handle it along the lines of the soundness analysis of \cite{Rag08} and \cite{RT12}.

Consider the dictatorship test derived from the gap instance $G_{\rm gap}$ as in \cite{Rag08,RT12} (Figure \ref{fig:test-alt}).
\begin{figure}[ht!]
	\begin{mdframed}
		{\bf Input}: Assignments $h_i:\{0,1\}^R \to [0,1]$ for $i \in V_{\rm gap}$ satisfying 
		\[
		\Ex_{i \sim G_{\rm gap}} \Ex_{x_i \sim \{0,1\}^R_{\mu_i}} \left[h_i(x_i)\right] = \mu.
		\]
		{\bf Test Distribution}:
		\begin{enumerate}
			\item Sample edge $e = (i_1,\ldots,i_r) \sim E_{\rm gap}$.
			\item For $j = 1,\ldots,R$ do the following independently: sample $(x_1(j),\ldots,x_r(j)) \sim \theta_e$.
			\item For every $i \in [r]$, sample $x'_i \underset{1 - \eta}{\sim} x_i$.
			\item Return payoff\footnote{Here the payoff refers to the value obtained by evaluating the multi-linear polynomial corresponding to $\psi$ on $h_{i_1}(x_{i_1}),\ldots,h_{i_r}(x_{i_r})$.} $\psi(h_{i_i}(x_{i_1}),\ldots,h_{i_r}(x_{i_r}))$.
		\end{enumerate}
		\caption{Gaps-to-Dictatorship Test}
		\label{fig:test-alt}
	\end{mdframed}
\end{figure}		
We claim that for a fixed choice of $A \in V^R$, the expectation term from \eqref{eqn:sound-val1} is precisely the arithmetization of the acceptance probability of the test with respect to the assignments $\{h_i = \bar{g}_{A,i}\}_i$. And we know from the analysis of \cite{Rag08},\cite{RT12}, whenever the assignment $\{h_i\}_i$ satisfies a suitable ``small-influences'' condition, the expectation term can be upper bounded by the ($\mu$-constrained) optimal value of the instance $G_{\rm gap}$. We state the precise formulation of this guarantee in the lemma below.

\begin{lemma}[Theorem 6.2~\cite{RT12} Restated]				\label{lem:round}
	Suppose $\{\theta_e\}_{e \in E_{\rm gap}}$ is a $\gamma^2$-independent $(\gamma\mu)^r$-smooth feasible solution for the SDP in Figure \ref{fig:lass}. Let $\{h_i\}_{i \in V_{\rm gap}}$ be a collection of functions where for every $i \in V_{\rm gap}$, $h_i:\{0,1\}^R \to [0,1]$ is a function in $L_2(\{0,1\}^R_{\mu_i})$. Furthermore, suppose the collection of functions satisfy the conditions:
	\begin{equation}				\label{eqn:cond-1}
		\Pr_{i \sim G_{\rm gap}}\left[\max_{j \in [R]} {\sf Inf}^{(\mu_i)}_j\left[\sT^{(\mu_i)}_{1 - \eta} h_i\right] > \tau \right] \leq \nu, 
	\end{equation}
	 \begin{equation}				\label{eqn:cond-1a}
	 	\Pr_{e \sim E_{\rm gap}}\left[\max_{i \in e}\max_{j \in [R]} {\sf Inf}^{(\mu_i)}_j\left[\sT^{(\mu_i)}_{1 - \eta} h_i\right] > \tau \right] \leq \nu, 
	 \end{equation}
 	and
	\begin{equation}			\label{eqn:cond-2}
		\Ex_{i \sim G_{\rm gap}} \Ex_{x \sim \{0,1\}^R_{\mu_i}} \left[h_i(x)\right] = \mu.
	\end{equation}
	Then there exists an efficient randomized algorithm which outputs an assignment $\sigma:V_{\rm gap} \to \{0,1\}$ such that the expected value under the assignment $\sigma$ is at least
	\[
	\sum_{a \in \psi^{-1}(1)}\Ex_{e \sim E_{\rm gap}} \Ex_{(x_i)_{i \in e} \sim \theta^R_e} \left[\prod_{i \in e} \sT^{(\mu_i)}_{1 - \eta}h^{(a,e,i)}_i(x_i)\right] - C'_r\tau^{C\eta \kappa/r^2},
	\]
	and the relative weight of the assignment is $\mu(1 \pm  2\sqrt{\gamma})$ with probability at least $1 - 2\sqrt{\gamma}$. Here $h^{(a,e,i)}_i$ is defined from $h_i$ identically as in \eqref{eqn:f-defs}.
\end{lemma}	

\cite{RT12}~actually proves the above for the {\sc Max-Cut} predicates, but with a weaker guarantee on the deviation of the weight of the rounded assignment. While extending their analysis to work for all Boolean predicates is straightforward, improving the bound on the deviation requires new ideas. We elaborate on this issue and detail how to improve the bound when we prove the lemma in Section \ref{sec:round}.

{\bf Bounding the Expectation term via Lemma \ref{lem:round}}. Now as preparation towards applying Lemma \ref{lem:round}, we verify that the conditions \eqref{eqn:cond-1}, \eqref{eqn:cond-1a}, and \eqref{eqn:cond-2} hold for the collection of functions $\{\bar{g}_{A,i}\}_{i \in V_{\rm gap}}$, for most choices of $A \in V^R$.	To that end, let us define the subsets $V'_{{\rm nice},1}, V'_{{\rm nice},2} \subset V^R$ as 
\[
V'_{{\rm nice},1} := \left\{ A \in V^R \Big| \Pr_{i \sim {G_{\rm gap}}} \left[\max_{j \in [R]} {\sf Inf}^{(\mu_i)}_j\left[\sT^{(\mu_i)}_{1 - \eta}\bar{g}_{A,i}\right] > \tau \right] \leq \nu \right\},
\]
and
\[
V'_{{\rm nice},2} := \left\{ A \in V^R \Big| \Pr_{e \sim {E_{\rm gap}}} \left[\max_{i \in e}\max_{j \in [R]} {\sf Inf}^{(\mu_i)}_j\left[\sT^{(\mu_i)}_{1 - \eta}\bar{g}_{A,i}\right] > \tau \right] \leq \nu \right\}.
\]
Let $V'_{\rm nice} := V'_{{\rm nice},1} \cap V'_{{\rm nice},2}$. As in Section \ref{sec:z-avg}, we have the following lemma analogous to Lemma \ref{lem:dec-set0}, but now in terms of the averaged functions $\{\bar{g}_{A,i}\}_{A \in V^R}$ on the $(x_i)_{i \in V_{\rm gap}}$-variables.
\begin{lemma} 			\label{lem:dec-set1}
	Since $G$ is a NO instance, we have 
	\[
	\Pr_{A \sim V^R} \left[ A \in V'_{\rm nice} \right] \geq 1 - \nu.
	\]
\end{lemma}
The proof of the above lemma again goes through the influence-decoding argument, we defer its proof to Section \ref{sec:dec-set1} for now. Next, we have the following lemma which says that for most choices of $A$, the collection of functions $\{g_{A,i}\}_i$ satisfy \eqref{eqn:cond-2} (up to negligible error).
\begin{lemma}[Long Code Mixing Lemma]			\label{lem:conc}
	Suppose $f$ satisfies \eqref{eqn:bias-constr}, and let $\{\bar{g}_{A,i}\}_{A,i}$ be the averaged functions constructed from $f$ (using  \eqref{eqn:g-def} and Corollary \ref{corr:z-avg}). For every $A \in V^R$, define $\mu_A := \Ex_{i \sim G_{\rm gap}} \Ex_{x_i \sim \{0,1\}^R_{\mu_i}} \left[\bar{g}_{A,i}(x_i)\right]$. Then for every $\alpha > 0$ we have 
	\[
	\Pr_{A \sim V^R}\left[\Big| \mu_A - \mu\Big| \geq \alpha \sqrt{\mu}\right] \leq \frac{|V_{\rm gap}|\beta}{\alpha^2}.
	\]
\end{lemma}
The above lemma is a straightforward application of Lemma 6.7 from \cite{RST12}; we prove it in Section \ref{sec:conc}. Let $V''_{\rm nice} \subset V^R$ denote the subset of vertices $A$ for which $|\mu_A - \mu| \leq \beta^{1/4} \sqrt{\mu}$. Instantiating the above lemma with $\alpha = \beta^{1/4}$ and using our choice of $\nu$ we get that $|V''_{\rm nice}| \geq (1 - \nu)|V^R|$. Overall, the above arguments imply that $|V'_{\rm nice} \cap V''_{\rm nice}| \geq (1 - 2\nu)|V^R|$.

Now fix an $A \in V'_{\rm nice} \cap V''_{\rm nice}$; for every such $A$, we know $\{\bar{g}_{A,i}\}_i$ satisfies \eqref{eqn:cond-1}, \eqref{eqn:cond-1a} (using the definition of $V'_{\rm nice}$) and \eqref{eqn:cond-2} (using the definition of $V''_{\rm nice}$), and hence we can use Lemma \ref{lem:round} to bound 
\begin{equation}				\label{eqn:sound-val2}
	\Ex_{e \sim E_{\rm gap}} \Ex_{(x_i)_{i \in e} \sim \theta^R_e} \left[\prod_{i \in e} \sT^\Pmu_{1 - \eta}\bar{g}_{A,i}(x_i)\right] \leq {\sf Opt}_{\mu(1\pm 2\sqrt{\gamma})}(G_{\rm gap}) + C'_r\sqrt{r}\tau^{C\eta\kappa/r^2}.
\end{equation}
Since the fraction of $A$'s which can be bounded using the above argument is at least $(1 - 2\nu)$, plugging in the bound from \eqref{eqn:sound-val2} into \eqref{eqn:sound-val1} gives us: 
\begin{equation}				\label{eqn:invar-4}
	\sum_{a \in \psi^{-1}(1)}\Ex_{A \sim V^R} \Ex_{e \sim E_{\rm gap}} \Ex_{(x_i)_{i \in e} \sim \theta^R_e} \left[\prod_{i \in e} \sT^{(\mu_i)}_{1 - \eta} \bar{g}_{A,i}(x_i) \right] + \mu^r		
	\leq 2\cdot {\sf Opt}_{\mu(1\pm 2\sqrt{\gamma})}(G_{\rm gap}) + O(\nu) + C'_r\sqrt{r}\tau^{C\eta\kappa/r^2}. 
\end{equation}

\subsection{Finishing the Soundness Analysis}				\label{sec:sound-pf}

Denoting $C_r := 2^r$, we stitch together the bounds from the various steps:
\begin{align*}
	&\Pr\Big[\mbox{ Test Accepts }\Big] \\
	& = \sum_{a \in \psi^{-1}(1)}\Ex_{e \sim E_{\rm gap}} \Ex_{A \sim V^R} \Ex_{(x_i,z_i)_{i \in e} \sim \cD^{R}_e} \left[\prod_{i \in e} \sT^{(\Omega_i)}_{1 - \eta}g^{(a,e,i)}_{A,i}(x_i,z_i)\right] 	\tag{Using \eqref{eqn:sound-1}} \\
	&\leq C_r\sum_{a \in \psi^{-1}(1)}\Ex_{e \sim E_{\rm gap}} \Ex_{A \sim V^R} \Ex_{(x_i)_{i \in e}} \left[\prod_{i \in e} \sT^{(\mu_i)}_{1 - \eta} \bar{g}^{(a,e,i)}_{A,i}(x_i)\right] + 2^r\mu^r + C_rC'_r\sqrt{r}\tau^{C\eta\kappa/r^2} + O(2^{2r}\nu) 			\tag{Using \eqref{eqn:invar-2}}\\
	&\leq 2C_r\cdot {\sf Opt}_{\mu(1\pm 2\sqrt{\gamma})}(G_{\rm gap}) + C_rC'_r\sqrt{r}\tau^{C\eta\kappa/r^2} + O(\nu \cdot C_r). 			\tag{Using \eqref{eqn:invar-4}} \\
	&\leq 2C_r {\sf Opt}_{\mu(1 \pm 2\sqrt{\gamma})}(G_{\rm gap}) + 2s \\
	&\leq 4C_r s,
\end{align*}
where the penultimate step follows using our choice of the parameters $\tau$ and $\nu$, and the last step follows using the robust soundness of the instance $G_{\rm gap}$. This concludes the soundness analysis of the reduction.

\subsection{Proof of Theorem \ref{thm:main}}				\label{sec:thm-main}

In this section, we put together the analyses from the previous sections to prove Theorem \ref{thm:main}

\begin{proof}[Proof of Theorem \ref{thm:main}]
Let $(G_{\rm gap},\theta')$ be a $(\ell,\mu,c,s,\gamma)$-gap instance for the {\sc Max-CSP}$(\psi)$ as in the setting of the Theorem \ref{thm:main}. Then as described in Section \ref{sec:redn}, we construct a $(\ell - 1/\gamma^C,\mu,c - O(\gamma,r),s)$-gap instance $(G_{\rm gap},\theta)$. Let $\epsilon,M$ be parameters defined as in below Figure \ref{fig:test-1}, and let $G = (V,E)$ be a $(\epsilon,\delta,M)$-SSE instance as in Conjecture \ref{conj:sseh}. 
Finally, we run the reduction from Figure \ref{fig:test-1} on $(G_{\rm gap},\theta)$ and $G$, and let $\cH= (V_\cH,E_{\cH},\tilde{w}_\cH, w_{\cH})$ be the resulting $\mu$-constrained {\sc Max-CSP}$(\psi)$ instance.

{\bf Completeness}. Suppose $G$ is a YES instance. Then the arguments from Section \ref{sec:comp} imply that there exists a feasible assignment which satisfies at least $\Omega(\rho^2 (c - O(\gamma r) - O(\eta r)))/r = \Omega(c\cdot r^{-3}\log^{-2}(1/\mu))$-fraction of constraints.

{\bf Soundness}. Suppose $G$ is a NO instance. Then, the arguments from Section \ref{sec:sound-pf} imply that any feasible assignment can satisfy at most $O(2^r \cdot s)$-fraction of constraints.

Combining the above with the \NP-hardness of SSE from Conjecture \ref{conj:sseh} completes the proof of the theorem.
\end{proof}

\subsection{Proof of Lemma \ref{lem:conc}}					\label{sec:conc}

Denote $\ell := V_{\rm gap}$. Also for every $i \in V_{\rm gap}$, define $\mu_{A,i}$ as 
\[
\mu_{A,i} := \Ex_{x \sim \{0,1\}^R_{\mu_i}}\bar{g}_{A,i}(x) = \Ex_{x \sim \{0,1\}^R_{\mu_i}} \Ex_{z \sim \{\bot,\top\}^R_\beta} g_{A,i}(x,z)
\]
Finally, let $\tilde{\mu}_i := \Ex_{A \sim V^R}[\mu_{A,i}]$. The key tool here is the following concentration bound from the soundness analysis in \cite{RST12}.

\begin{lemma}[Restatement of Lemma 6.7~\cite{RST12}]
	Fix any $i \in V_{\rm gap}$. Then for any $\alpha > 0$, we have 
	\[
	\Pr_{A \sim V^R} \Big[|\mu_{A,i} - \tilde{\mu}_i| \geq \alpha\sqrt{\tilde{\mu}_i}\Big] \leq \frac{\beta}{\alpha^2}.
	\]
\end{lemma}
Then, applying the above lemma for every $i \in V_{\rm gap}$ followed by a union bound we get that
\[
\Pr_{A \sim V^R}\left[\exists i \in V_{\rm gap} : 
|\mu_{A,i} - \tilde{\mu}_i| \geq \alpha\sqrt{\tilde{\mu}_i}\right] \leq \ell \beta/\alpha^2.
\]
Finally, to conclude the proof, we observe that for any $A \in V^R$ for which we have $|\mu_{A,i} - \tilde{\mu}_i| \leq \alpha\sqrt{\mu_i}$ for every $i \in V_{\rm gap}$, we can show that
\[
\mu_A = \Ex_{i \sim V_{\rm gap}} \mu_{A,i} \leq \Ex_{i \sim V_{\rm gap}}\left[\tilde{\mu}_i + \alpha\sqrt{\tilde{\mu}_i}\right] \leq \mu + \alpha\sqrt{\mu}, 
\]
where in the last step, using Jensen's inequality we can bound that $\Ex_i\sqrt{\mu_i} \leq \sqrt{\Ex_{i} \tilde{\mu}_i}  = \sqrt{\mu}$. Using similar arguments we can also show that $\mu_A \geq \mu - \alpha\sqrt{\mu}$. This establishes the first inequality. The second inequality then follows directly from our choice of parameters.

%% file: note.tex
\section{Proof of Lemma \ref{lem:prod}}					\label{sec:prod}

The proof of the lemma uses Theorem \ref{thm:clt} and Corollary \ref{corr:egt}, along with $r$-ary noise stability estimates for $\Lambda_\rho(\cdot)$. Recall that in \eqref{eqn:lambda-def}, we defined the $r$-ary Gaussian stability as
\[
\Lambda_\rho(\delta_1,\ldots,\delta_r) = \Pr_{g \sim N(0,1)} \Pr_{g_1,\ldots,g_r \underset{\rho}{\sim} g}\left[\forall j \in [r], g_j \leq \Phi^{-1}(\delta_j)\right],
\] 
where $\Phi:\mathbbm{R} \to [0,1]$ is the Gaussian CDF function. The proof of the lemma will require the following explicit bound on $\Lambda_\rho$.

\begin{lemma}[Folklore]					\label{lem:hspace}
	There exists $\delta_0 \in (0,1/2)$ such that the following holds. Let $\delta_1,\ldots,\delta_r \in [0,\delta_0]$, and let $0 
	\leq \rho \leq 1/\big({4r^2\log(1/\delta^*)}\big)$ where $\delta^* = \min_{i \in [r]} \delta_i$.  Then,
	\[
	\Lambda_{\rho}\big(\delta_1,\ldots,\delta_r\big) \leq 2^r \prod_{i \in [r]} \delta_i.
	\] 	
\end{lemma}

Several similar bounds are known in the literature for related notions of Gaussian stability (for e.g., see \cite{MNT16},\cite{KS15}). For completeness, we include a proof of the above in Section \ref{sec:hspace}.

{\bf Moving from $R$-variate to $2R$-variate space}. Before we prove the lemma, we first show that we can transfer the analysis from the $R$-variate spaces $\Omega^R_i$ to the $2R$-variate space $\{0,1\}^R_{\mu_i} \otimes \{\bot,\top\}^R_\beta$ -- this will enable us to deal with the covariance structure of the $x$ and $z$ variables separately in the subsequent steps. Recall that in the setting of the lemma, we have a collection of functions $h_1,\ldots,h_r:\Omega^R \to [0,1]$ where $h_i \in L_2(\Omega^R_i)$ for every $i \in [r]$. Now for every $i \in [r]$, let  $h'_i:\{0,1\}^R \times \{\bot,\top\}^R \to [0,1]$ be the corresponding $2R$-variate function in the space $L_2(\{0,1\}^R_{\mu_i} \otimes \{\bot,\top\}^R_\beta)$ i.e.,
\[
h'_i\Big(x_i(1),\ldots,x_i(R),z_i(1),\ldots,z_i(R)\Big) \defeq h_i\Big((x_i,z_i)(1),\ldots,(x_i,z_i)(R)\Big).
\]
Note that although $h_i(\cdot)$ and $h'_i(\cdot)$ are identical functions on $\Omega^R$, we will use $h'_i$ to explicitly denote that the function is being defined w.r.t the $2R$-dimensional probability space $\{0,1\}^R_{\mu_i} \otimes \{\bot,\top\}^R_\beta$. Then by definition, the influences of the function $h'_i$ will be defined with respect to the $2R$-variables $x_i(1),\ldots,x_i(R)$ and $z_i(1),\ldots,z_i(R)$. In particular, we shall use ${\sf Inf}_{x_i(j)}[h'_i]$ to denote the influence of the coordinate corresponding to the variable $x_i(j)$, and similarly we use ${\sf Inf}_{z_i(j)}[h'_i]$ to denote the influence of the coordinate corresponding to variable $z_i(j)$, for any $j \in [R]$. The following observation is a direct consequence of the definition of the probability spaces corresponding to the functions $\{h'_i\}_{i \in [r]}$.

\begin{observation}					\label{obs:h-exp}
	For every $i \in [r]$, the Fourier expansion of $h'_i$ with respect to the $2R$-dimensional probability space $\{0,1\}^R_{\mu_i} \otimes \{\bot,\top\}^R_\beta$ is given as 
	\[
	h'_i(x_i,z_i) = \sum_{S,T \subseteq [R]} \wh{h'_i}(S,T)\prod_{j \in S}  \phi^{(\mu_i)}(x_i(j)) \prod_{j' \in T}\phi^{(\beta)}(z_i(j')),
	\]
	where recall that $\phi^{(\mu_i)}$ and $\phi^{(\beta)}$ are the non-trivial Fourier characters for $\{0,1\}_{\mu_i}$ and $\{0,1\}_\beta$. Consequently, for every $j \in [R]$ we have that 
	\[
	{\sf Inf}_{x_i(j)}\left[h'_i\right] = \sum_{\substack{S,T \subseteq [R] \\ j \in S}} \wh{h'_i}(S,T)^2 
	\ \ \ \ \ \ \ \textnormal{and} \ \ \ \ \ \ \
	{\sf Inf}_{z_i(j)}\left[h'_i\right] = \sum_{\substack{S,T \subseteq [R] \\ j \in T}} \wh{h'_i}(S,T)^2.
	\]
\end{observation}

Next, we have the following claim which shows that if $h_i$ has small influences in the $R$-dimensional probability space $\Omega^R_i$, then $h'_i$ has small influences in the $2R$-dimensional space $\{0,1\}^R_{\mu_i} \otimes \{\bot,\top\}^R_\beta$.

\begin{claim}					\label{cl:inf-tr}
	The following holds for any fixed $i \in [r]$. For every $j \in [R]$, we have 
	\[
	\max\left\{{\sf Inf}_{x_i(j)}[h'_i], {\sf Inf}_{z_i(j)}[h'_i]\right\} \leq {\sf Inf}^{(\Omega_i)}_j[h_i],
	\]
	where recall that ${\sf Inf}^{(\Omega_i)}_j[h_i]$ is the influence of $(x_i,z_i)(j)$ in $h_i$ measured in the $R$-dimensional space $\Omega^R_i$ i.e.,
	\[
	{\sf Inf}^{(\Omega_i)}_j[h_i] = \Ex_{(x_i(j'),z_i(j'))_{j' \neq j}}\left[{\sf Var}_{(x_i(j),z_i(j))}\left[ f(x_i,z_i)~\Big|~(x_i(j'),z_i(j'))_{j' \neq j} \right] \right].
	\]
\end{claim}
\begin{proof}
	Fix an $i \in [r]$, and let $h_i \in L_2(\Omega^R_i)$ and $h'_i \in L_2(\{0,1\}^R_{\mu_i} \otimes \{\bot,\top\}^R_\beta)$ be as above. Now, since $\Omega_i := \{0,1\}_{\mu_i} \otimes \{\bot,\top\}_\beta$, we can use the Fourier characters of $\{0,1\}_{\mu_i}$ and $\{\bot,\top\}_\beta$ to construct a Fourier basis for $\Omega_i$. Let us denote the Fourier characters corresponding to the probability space $\Omega_i$ as $\{\bar{\phi}_\alpha:\Omega_i \to \mathbbm{R}\}_{\alpha \in \{0,1\}^2}$ and define them as:
	\[
	\bphi_{0,0} \equiv 1, \qquad \bphi_{1,0}(x,z) = \phi^{(\mu_i)}(x), \qquad \bphi_{0,1}(x,z) = \phi^{(\beta)}(z), \qquad \bphi_{1,1}(x,z) = \phi^{(\mu_i)}(x)\phi^{(\beta)}(z).
	\]
	As we will soon see, it is convenient for us to index the Fourier characters of $\Omega_i$ using elements of $\{0,1\}^2$ instead of $\mathbbm{Z}_{\leq 3}$. Note that since $h_i \in L_2(\Omega^R_i)$, we can use the above basis to write the Fourier expansion of $h_i$ as:
	\begin{equation}				\label{eqn:expr-1}
	h_i(x,z) = \sum_{\alpha \in (\{0,1\}^2)^R} \wh{h_i}(\alpha)\prod_{j \in [R]} \bphi_{\alpha(j)}\Big((x,z)(j)\Big). 
	\end{equation}
	Now consider the following one-to-one correspondence between $(\{0,1\}^2)^R$ and $2^{[R]} \times 2^{[R]}$. Given any $\alpha \in (\{0,1\}^2)^R$, we can uniquely identify sets $S_{\alpha}$ and $T_{\alpha}$ where
	\[
	S_{\alpha} = \left\{j \in [R]~\Big|~\alpha(j) \in \{(1,0),(1,1)\}\right\}
	\ \ \ \ \ \  \textnormal{and} \ \ \ \ \ \ 
	T_{\alpha} = \left\{j \in [R]~\Big|~\alpha(j) \in \{(0,1),(1,1)\}\right\}.
	\]
	Similarly, given $S,T \subseteq [R]$, there exists a unique $\alpha \in (\{0,1\}^2)^R$ such that $(S_{\alpha}, T_{\alpha}) = (S,T)$. Using this one-to-one correspondence along with the definition of $\bphi_\alpha$-characters, we re-index the summation and the summands in \eqref{eqn:expr-1}:
	\begin{align}
		&\sum_{\alpha \in (\{0,1\}^2)^R} \wh{h_i}(\alpha)\prod_{j \in [R]} \bphi_{\alpha(j)}\Big((x,z)(j)\Big) 		\non\\
		&= \sum_{\alpha \in (\{0,1\}^2)^R} \wh{h_i}(\alpha)\prod_{j: \alpha(j) \in \{(1,0),(1,1)\}} \phi^{(\mu_i)}(x(j))\prod_{j': \alpha(j') \in \{(0,1),(1,1)\}} \phi^{(\beta)}(z(j')) 		\non\\
		&= \sum_{\alpha \in (\{0,1\}^2)^R} \wh{h_i}(\alpha) \prod_{j \in S_\alpha} \phi^{(\mu_i)}(x(j)) \prod_{j' \in T_\alpha}\phi^{(\beta)}(z(j')). 	
				\label{eqn:expr}
	\end{align}
	But then \eqref{eqn:expr} expresses $h_i = h'_i$ as a multilinear polynomial in the Fourier basis corresponding to the $2R$-variate space $\{0,1\}^R_{\mu_i} \otimes \{\bot,\top\}^R_\beta$ and hence it follows that \eqref{eqn:expr} must be the Fourier expansion of $h'_i$ in the $2R$-dimensional space (as described in Observation \ref{obs:h-exp}). In particular, using the uniqueness of Fourier expansion, we must have that $\wh{h_i}(\alpha) = \wh{h'_i}(S_{\alpha},T_{\alpha})$ for every $\alpha \in \{0,1\}^2$.
	
	Using the above equivalence of the Fourier coefficients from the expansions of $h_i$ and $h'_i$, we can infer that
	\[
	{\sf Inf}^{(\Omega_i)}_j\big[h_i\big] = \sum_{\alpha : \alpha(j) \neq (0,0)} \wh{h_i}(\alpha)^2 = \sum_{j \in S \cup T} \wh{h'_i}(S,T)^2.
	\]
	Hence, using the above along with Observation \ref{obs:h-exp} we can conclude that
	\[
	{\sf Inf}_{x_i(j)}\big[h'_i\big] = \sum_{S \ni j}\sum_{T \subseteq [R]}\wh{h'_i}(S,T)^2 \leq \sum_{j \in S \cup T} \wh{h'_i}(S,T)^2 = {\sf Inf}^{(\Omega_i)}_j\big[h_i\big],
	\]
	and similarly,
	\[
	{\sf Inf}_{z_i(j)}\big[h'_i\big] = \sum_{S \subseteq [R]}\sum_{T \ni j}\wh{h'_i}(S,T)^2 \leq \sum_{j \in S \cup T} \wh{h'_i}(S,T)^2 = {\sf Inf}^{(\Omega_i)}_j\big[h_i\big],
	\]
	which completes the proof.
\end{proof}

\subsection{Proof of Lemma \ref{lem:prod}}

Before we begin, we point out to the reader that this section heavily uses the notation and terminology introduced in Section \ref{sec:clt}. For ease of notation, by re-indexing we may assume that $e = [r]$. In the setting of the lemma, we are given functions $h_1,\ldots,h_r:\Omega^R \to [0,1]$, where $h_i \in L_2(\Omega^R_i)$ is a function of $x_i $ and $z_i$ for every $i \in [r]$. We now introduce various quantities which will be used frequently in the proof of the lemma.

\begin{itemize}
	\item As defined above, given $h_1,\ldots,h_r$, let $h'_1,\ldots,h'_r$ be the corresponding $2R$-variate functions where $h'_i \in L_2(\{0,1\}^R_{\mu_i} \otimes \{\bot,\top\}^R_\beta)$ for every $i \in [r]$.  \\[-8pt]
	\item Let $\cW := (\cW^x_1,\ldots,\cW^x_R,\cW^z_1,\ldots,\cW^z_R)$ be an independent sequence of ensembles (see Section \ref{sec:clt} for the definition), where for every $j \in [R]$, the ensembles $\cW^x_j = (\cW^x_j[i])^r_{i = 0}$ and $\cW^z_j = (\cW^z_j[i])^r_{i = 0}$ are defined as functions of the variables $(x_i(j))_{i \in [r]}$ and $(z_i(j))_{i \in [r]}$ as follows:
	\[
	\cW^x_j[i] := 
	\begin{cases}
		1 & \mbox{ if } i = 0, 		\\
		\phi^{(\mu_i)}(x_i(j)) & \mbox{ if } i \in [r]
	\end{cases}
	\qquad\qquad  
	\cW^z_j[i] := 
	\begin{cases}
		1 & \mbox{ if } i = 0, 		\\
		\phi^{(\beta)}(z_i(j)) & \mbox{ if } i \in [r]
	\end{cases},
	\]
	where recall that $\phi^{(\mu_i)}$ and $\phi^{(\beta)}$ are the non-trivial Fourier characters for the $1$-dimensional spaces $\{0,1\}_{\mu_i}$ and $\{0,1\}_{\beta}$. Note that the ensemble sequences $(\cW^x_j)^R_{j = 1}$ and $(\cW^z_j)^R_{j = 1}$ are independent of each other.\\[-8pt]
	\item For any $i \in [r]$, since $h'_i \in L_2(\{0,1\}^R_{\mu_i} \otimes \{\bot,\top\}^R_\beta)$, using its Fourier expansion, we can express $h'_i$ as a polynomial $Q_i$ in the variables from the ensembles in $\cW$ (Observation \ref{obs:h-exp}): 
	\begin{align}				\label{eqn:Q-def1}
		Q_i({\cW}) 
		&:= \sum_{S,T \subseteq [R]} \wh{h'_i}(S,T) \prod_{j \in S} \cW^x_j[i]\prod_{j' \in T} \cW^z_{j'}[i] \\
		&= \sum_{S,T \subseteq [R]} \wh{h'_i}(S,T) \prod_{j \in S} \phi^{(\mu_i)}(x_i(j))\prod_{j' \in T} \phi^{(\beta)}(z_i(j')) 		\non\\
		&= h'_i(x_i,z_i). 		\non
	\end{align}
	Note that in the polynomial $Q_i$, each monomial contains {\em at-most} one variable from $\cW^x_j$ and $\cW^z_j$ for every $j \in [R]$ -- but we can also interpret it as having {\em exactly} one variable from $\cW^x_j$ and $\cW^z_j$ for every $j$, by including the `$1$' variable from the ensembles missing in that monomial. Hence $Q_i$ is a {\bf multi-linear} polynomial in the ensemble sequence $\cW$. \\[-8pt]
	\item Finally, let $Q := (Q_1,\ldots,Q_r): \cW \to \mathbbm{R}^r$ denote the $r$-dimensional vector-polynomial that is the concatenation of $Q_1,\ldots,Q_r$. Let $\Pi_{[0,1]}$ be the clipping function defined in \eqref{eqn:pi-def} and define $\Psi:\mathbbm{R}^r \to [0,1]$ as 
	\[
	\Psi(\ell) := \prod_{i \in [r]} \Pi_{[0,1]}(\ell_i),\qquad\forall~\ell\in \mathbbm{R}^r.
	\] 
	It is well-known that $\Psi$ is $\sqrt{r}$-lipschitz continuous.
\end{itemize}

Now using the above notation, we can rewrite the term we wish to upper bound in the lemma:
\begin{align}				\label{eqn:expr-0}
\Ex_{(x_i,z_i)_{i \in [r]} \sim \cD^R_e}\left[\prod_{i \in [r]} h_i(x_i,z_i)\right]
&= \Ex_{(x_i,z_i)_{i \in [r]} \sim \cD^R_e}\left[\prod_{i \in [r]} h'_i(x_i,z_i)\right] 		\non\\
&= \Ex_{\cW} \left[\prod_{i \in [r]} Q_i(\cW)\right] 	\non\\		
&= \Ex_{\cW} \left[\Psi(Q(\cW))\right].				
\end{align}

The rest of the proof consists of several steps which we outline below:
\begin{itemize}
	\item[1.] First, we construct a Gaussian ensemble sequence $\cG$ that matches the covariance structure of $\cW$. In particular, the ensemble sequence will have the property that the variables in $\cG$ correspond to coordinates of $R$-dimensional Gaussian vector variables $\{g_{x,i},g_{z,i}\}_{i \in [r]}$ where $\{g_{x,i}\}_{i \in [r]}$ and $\{g_{z,i}\}_{i \in [r]}$ are independent w.r.t. each other. Furthermore,  the vector variables $\{g_{z,i}\}_{i \in [r]}$ will be distributed as $r$ independent $\rho$-correlated copies of a Gaussian vector $g \sim N(0_R,I_R)$.
	\item[2.] Next, we will use the invariance principle to show that $\Ex\left[\prod_{i \in [r]}Q_i(\cW)\right] \approx \Ex\left[\prod_{i \in [r]} Q_i(\cG)\right]$, which will move the analysis to the Gaussian space (Lemma \ref{lem:prod-1}).
	\item[3.] We will then use the Exchangeable Gaussians Theorem (Theorem \ref{thm:egt}) along with Lemma \ref{lem:hspace} to show that we can switch from correlated structure in $\{g_{z,i}\}_i$ to independent structure in $\{g_{z,i}\}$ at the cost of a multiplicative factor depending only on $r$ (Lemma \ref{lem:prod-2}), which in turn will allow us to average out the $\cG_z$-variables.
	\item[4.] Finally, we will again use the invariance principle to move the analysis back into the Boolean space (Lemma \ref{lem:prod-3}).
\end{itemize}

\subsection{Step 1: Gaussian Ensemble Sequence Construction.} 

We begin by making a couple of observations about the first and second moments of the variables in $\cW$.

\begin{claim}			\label{cl:w-moment}
	For every $j \in [R]$ and $i \in [r]$, we have $\Ex\left[\cW^x_j[i]\right] = \Ex\left[\cW^z_j[i]\right] = 0$. Moreover, for every $j \in [R]$ and $i \in \{0,1,\ldots,r\}$ we have that $\Ex\left[\cW^x_j[i]^2\right] = \Ex\left[\cW^z_j[i]^2\right] = 1$.
\end{claim}
\begin{proof}
	Fix a $j \in [R]$. Then observe that for every $i \in [r]$ we have
	\[
	\Ex\left[\cW^x_j[i]\right] = \Ex_{x_i \sim \{0,1\}^R_{\mu_i}}\left[\phi^{(\mu_i)}(x_i(j))\right] = 0,
	\ \ \ \ \ \  \textnormal{and} \ \ \ \ \ \ \
	\Ex\left[\cW^x_j[i]^2\right] = \Ex_{x_i \sim \{0,1\}^R_{\mu_i}}\left[\phi^{(\mu_i)}(x_i(j))^2\right] = 1,
	\] 
	since $\phi^{(\mu_i)}(\cdot)$ is orthonormal in the probability space $\{0,1\}_{\mu_i}$. Furthermore, $\Ex\left[\cW^x_j[0]^2 \right] = 1$ since $\cW^x_j[0]$ by definition is the constant `1' variable. This establishes the claim for $\cW^x_j$ ensembles. The claim for the $\cW^z_j$ ensembles follows similarly.
\end{proof}

Furthermore, it will be useful to explicitly characterize the covariance structure of the $\cW^z_j$ ensembles. Recall that under the test distribution (from Figure \ref{fig:test-1}), for every $j \in [R]$ the variables $(z_i(j))_{i \in [r]}$ are distributed as follows:
\begin{itemize}
	\item W.p. $\rho^2$, for every $i \in [r]$, $z_i(j) = z(j)$, where $z(j)$ is sampled as $\{\bot,\top\}_\beta$.
	\item W.p. $1 - \rho^2$, for every $i \in [r]$, $z_i(j)$ is sampled from $\{\bot,\top\}_\beta$ independently.
\end{itemize}
Using the above, we can infer that $\cW^z_j$ has the following covariance structure.
\begin{claim}						\label{cl:cov-z}
For every $j \in [R]$ we have
\begin{equation}					\label{eqn:cov}
\Ex\left[\cW^z_j[i]\cW^z_j[i']\right]
= 
\begin{cases}
	1 & \mbox{ if } i = i', \\
	0 & \mbox{ if } i \neq i'\mbox{ and } i = 0 \mbox{ or } i' = 0, \\
	\rho^2 & \mbox{ if } i \neq i' \mbox{ and } i,i'\neq 0.
\end{cases}
\end{equation}
\end{claim}
\begin{proof}
	The first two cases follow using Claim \ref{cl:w-moment}. For the last item, observe that if $i \neq i'$ and $i,i' \neq 0$ then,
	\begin{align*}
	\Ex\left[\cW^z_j[i]\cW^z_j[i']\right] 
	&= \Ex_{(z_i)_{i \in [r]}}\left[\phi^{(\beta)}(z_i(j))\phi^{(\beta)}(z_{i'}(j))\right] \\
	&= \rho^2\Ex_z\left[\phi^{(\beta)}(z(j))^2\right] + (1 - \rho^2)\Ex_{z_i(j),z_{i'}(j) \sim \{\bot,\top\}_{\beta}}\left[\phi^{(\beta)}(z_i(j))\phi^{(\beta)}(z_{i'}(j))\right]  \\
	&= \rho^2,
	\end{align*} 
	where the first step follows due to the distribution of $(z_i(j))_{i \in [r]}$, and in the second step, the first expectation evaluates to $1$ and the second expectation evaluates to $0$ since $\phi^{(\beta)}(\cdot)$ is  orthonormal in the probability space $\{\bot,\top\}_\beta$.
\end{proof}

{\bf Gaussian Ensemble Construction}. Considering the above observations, we shall now construct an independent sequence of Gaussian ensembles $\cG :=  (\cG^x_1,\ldots,\cG^x_R,\cG^z_1,\ldots,\cG^z_R)$ which will match the covariance structure of ${\cW}$. We begin by defining a collection of $2r$ jointly distributed $R$-dimensional Gaussian vector random variables $\{g_{x,i}\}_{i \in [r]} \cup \{g_{z,i}\}_{i \in [r]}$. For every $j \in [R]$, let $\Sigma^x_j$ denote the covariance matrix
\[
\Sigma^x_j := \Ex_{\cW}\left[\cW^x_j[1:r]\left(\cW^x_j[1:r]\right)^\top\right],
\]
where $\cW^x_j[1:r] = (\cW^x_j[i])^r_{i = 1}$. We now define the joint distribution over the Gaussian vectors as follows:
\begin{itemize}
	\item For every $j \in [R]$, independently sample $r$-dimensional Gaussian vector $(g_{x,1}(j),\ldots,g_{x,r}(j)) \sim N(0_r,\Sigma^x_j)$.
	\item Sample independent $R$-dimensional Gaussian vectors $g, \zeta_1,\ldots,\zeta_r \sim N(0_R,I_R)$. For every $i \in [r]$, set $g_{z,i} := \rho\cdot g + \sqrt{1 - \rho^2} \cdot \zeta_i$.
\end{itemize}
Note that since every non-constant variable in $\cW$ has mean $0$ and variance $1$, for every $j \in [R]$ and $i \in [r]$, the variables $g_{x,i}(j)$ and $g_{z,i}(j)$ are marginally distributed as $N(0,1)$ in the above construction.
Using these vector variables, for every $j \in [R]$, we define the $j^{th}$-Gaussian ensembles $\cG^x_j := (\cG^x_j[i])^{r}_{i = 0}$ and $\cG^z_j:= (\cG^z_j[i])^r_{i = 0}$ as
\begin{equation}				\label{eqn:cg-def}
\cG^x_j[i]:= 
\begin{cases}
	1 & \mbox{ if } i = 0, \\
	g_{x,i}(j) &\mbox{ if } i \in [r]
\end{cases}
\qquad\qquad 
\cG^z_j[i] :=  
\begin{cases}
	1 & \mbox{ if } i = 0, \\
	g_{z,i}(j) &\mbox{ if } i \in [r],
\end{cases}
\end{equation}
and finally let $\cG = (\cG^x_1,\ldots,\cG^x_R,\cG^z_1,\ldots,\cG^z_R)$. The following claim shows that the Gaussian ensemble sequence $\cG$ has covariance structure matching  with $\cW$.
\begin{claim}				\label{cl:cov}
	The ensemble sequences $\cW$ and $\cG$ have matching covariance structure.
\end{claim}
\begin{proof}
Fix a  $j \in [R]$. Since every ensemble in $\cW$ and $\cG$ consists of $(r + 1)$-variables, we have $|\cW^x_j| = |\cG^x_j|$ and |$\cW^z_j| = |\cG^z_j|$ . Next, by construction we have that 
\[
\Ex_{\cG}\left[\cG^x_j[1:r]\left(\cG^x_j[1:r]\right)^\top\right] 
= \Ex_{\cG}\left[(g_{x,i}(j))_{i \in [r]} \left((g_{x,i}(j))_{i \in [r]}\right)^\top\right]= \Sigma^x_j = \Ex_{\cW}\left[\cW^x_j[1:r]\left(\cW^x_j[1:r]\right)^\top\right].
\]
Since the $0^{th}$ variable in both $\cW^x_j$ and $\cG^x_j$ is the constant `1' variable, and for every $i \in [r]$, $\cG^x_j[i]$ is marginally distributed as $N(0,1)$, the above implies that 
\[
\Ex_{\cG}\left[\cG^x_j\left(\cG^x_j\right)^\top\right] =  \Ex_{\cW}\left[\cW^x_j\left(\cW^x_j\right)^\top\right].
\]
On the other hand, for the variables in $\cG^z_j$ we observe that
\begin{equation}				\label{eqn:cov-2}
	\Ex_{\cG}\left[\cG^{z}_j[i]\cG^z_j[i']\right] = 
	\begin{cases}
		1 & \mbox{ if } i = i' \\
		0 & \mbox{ if } i \neq i', i = 0 \mbox{ or } i' = 0 \\
		\rho^2 & \mbox{ if } i \neq i' \mbox{ and } i,i\neq 0. 
	\end{cases}
\end{equation}
Here the first two items follow using the fact that $\cG^z_j[i]$ is marginally distributed as $N(0,1)$ for every $i \in [r]$, and $\cG^z_j[0]$ is the `1' variable. For the third item, we use the definition of $\cG^z_j$ and observe that 
\[
\Ex\left[\cG^z_j[i]\cG^z_j[i']\right] 
= \Ex\Big[g_{z,i}(j)g_{z,i'}(j)\Big]
= \Ex\left[\left(\rho g(j) + \sqrt{1 - \rho^2}\zeta_i(j)\right)\left(\rho g(j) + \sqrt{1 - \rho^2}\zeta_{i'}(j)\right)\right] = \rho^2.
\]
Hence, using \eqref{eqn:cov-2} and Claim \ref{cl:cov-z}, we can conclude that $\Ex_{\cG}\left[\cG^z_j\left(\cG^z_j\right)^\top \right] = \Ex_{\cW}\left[\cW^z_j\left(\cW^z_j\right)^\top \right]$. Since the above arguments holds for any $j \in [R]$, we have that $\cG$ matches the covariance structure of $\cW$.
\end{proof}
We conclude with a couple of related observations on the Gaussian vectors $\{g_{x,i},g_{z,i}\}_{i \in [r]}$.
\begin{observation}			\label{obs:g-vec}
	The Gaussian vectors $\{g_{x,i},g_{z,i}\}_{i \in [r]}$ as described above satisfy the following properties:
	\begin{itemize}
		\item For every $i \in [r]$, the Gaussian vector variables $g_{x,i}$ and $g_{z,i}$ are marginally distributed as $N(0_R,I_R)$.
		\item Furthermore, $g_{z,1},\ldots,g_{z,r}$ are distributed as independent $\rho$-correlated copies of $g$.
	\end{itemize}
\end{observation}
\begin{observation}
	For every $i \in [r]$, the polynomial $Q_i$ extends naturally the Gaussian ensemble $\cG$ as 
	\[
	Q_i(\cG) = \sum_{S,T \subseteq [R]} \wh{h'_i}(S,T) \prod_{i \in S} \cG^x_j[i] \prod_{j' \in T} \cG^z_{j'}[i].
	\]
\end{observation}

\subsection{Step 2: Invariance Application Step} 					\label{sec:step-2}

In this section, we use the Invariance Principle to move the analysis to the Gaussian space, as stated in the following lemma.
\begin{lemma}		\label{lem:prod-1}
	Let $\cW,\cG$, and $Q$ be as above. Then,
	\[
	\Ex_{\cW} \left[\Psi(Q(\cW))\right] 				
	\leq \Ex_{\cG} \left[\Psi(Q(\cG))\right] + C_r\sqrt{r}\tau^{O(\eta \kappa/r^2)},
	\]
	where $\kappa = \beta/\log(1/\gamma)$.
\end{lemma}
\begin{proof}

Our first step here is to verify that the ensemble sequences $\cW$,$\cG$ and the vector-valued polynomial $Q$ satisfy the conditions required to invoke Theorem \ref{thm:clt}. 

\begin{itemize}
	\item[$C_1$] From Claim \ref{cl:cov}, we have that the ensembles ${\cW}$ and $\cG$ have matching covariance structure.
	\item[$C_2$] Due to the $(\gamma\mu)^r$-smoothness of the local distributions (Lemma \ref{lem:smooth}), we know that for any $\alpha \in \{0,1\}^e$ we have $\Pr_{X_e \sim \theta_e}\left[X_e = \alpha\right] \geq (\gamma\mu)^r$. Since the random variables in $\cW^x_j$ are completely determined by the test distribution variables $x_1(j),\ldots,x_r(j)$, the above property also implies that $\Pr[\cW^x_j = \omega] \geq (\gamma \mu)^r$ for every $\omega$ in the event space corresponding to $\cW^x_j$, for every $j \in [R]$. Similarly, every event in the space corresponding to $\cW^z_j$ happens with probability at least $\eta^r \beta^r$.  
	\item[$C_3$] From the setting of the lemma, for every $i \in [r]$ and $j \in [R]$, we have that 
	\begin{align*}
	\Inf{\cW^x_j}{Q_i} 
	&= \Ex_{\cW^x_{j'\neq j},\cW^z}\left[{\rm Var}_{\cW^x_j}\left[Q_i(\cW) \Big|(\cW^x_{j'})_{j' \neq j},\cW^z\right]\right] \\
	&\overset{1}{=} \Ex_{(x_i(j'))_{j' \neq j},z_i}\left[{\rm Var}_{x_{i}(j)}\left[h'_i(x_i,z_i)\Big|(x_i(j'))_{j' \neq j},z_i\right]\right] \\
	&= {\sf Inf}_{x_i(j)}\left[h'_i\right] \\
	&\overset{2}{\leq} {\sf Inf}^{(\Omega_i)}_j\left[h_i\right] \leq \tau,
	\end{align*}
	where in step $1$ we again use the observation that for every $\ell \in [R]$, the only variable in $\cW^x_\ell$ which appears in $Q_i(\cW)$  is $\cW^x_\ell[i] = \phi^{(\mu_i)}(x_i(\ell))$ (and similarly for $\cW^z_\ell$),  step $2$ uses Claim \ref{cl:inf-tr}, and the final inequality follows from the assumption on $h_i$ in the setting of the lemma. Using identical arguments, we also have that ${\sf Inf}_{\cW^z_j}\left[Q_i\right] \leq \tau$ for every $i \in [r]$ and $j \in [R]$.	
	\item[$C_4$] Furthermore, for every $i \in [r]$, from the setting of the lemma we have ${\rm Var}[Q^{>d}_i] = {\rm Var}[h^{>d}_i] \leq (1 - \eta)^d$ for every $d \geq \Omega(\log(1/\tau)/\log(1/\gamma^r\eta^r\beta^r))$.
\end{itemize}

Then the conditions $C_1$-$C_4$ from above show that the ensemble sequences $\cW,\cG$ along the $r$-dimensional vector polynomial $Q(\cW)$ meet the conditions required to invoke Theorem \ref{thm:clt}. Therefore, using Theorem \ref{thm:clt} along with the fact that $\Psi(\cdot)$ is $\sqrt{r}$-lipschitz, we can upper bound \eqref{eqn:expr-0} as:
\begin{align}					
	\Ex_{\cW} \left[\Psi(Q(\cW))\right] 				
	&\leq \Ex_{\cG} \left[\Psi(Q(\cG))\right] + C_r\sqrt{r}\tau^{\frac{C\eta \mu \beta}{(r\log(1/\beta\eta\gamma))}} 		\non\\
	&\leq \Ex_{\cG} \left[\Psi(Q(\cG))\right] + C_r\sqrt{r}\tau^{O(\eta \kappa/r^2)},  \non
\end{align}
where the last step uses $\kappa = \beta/\log(1/\gamma)$.
\end{proof}

\subsection{Step 3: Decoupling the $\cG^z$ variables}

In this section, we show that we can move from $\{\cG^z_j\}_j$ to an ensemble of {\em completely} independent Gaussians (while retaining the correlation structure of $\{\cG^x_j\}_j)$, at the cost of multiplicative and (negligible) additive factors, as stated in the following lemma.
\begin{lemma}			\label{lem:prod-2}
	Let $\cG$ be as above, and let $\what{\cG} := (\what{\cG}^x_1,\ldots,\what{\cG}^x_R,\what{\cG}^z_1,\ldots,\what{\cG}^z_R)$ be another independent sequence of Gaussian ensembles defined as follows. Let $\tilde{g}_{z,1},\ldots,\tilde{g}_{z,r} \sim N(0_R,I_R)$ be $r$ independent $R$-dimensional Gaussian vector variables. For every $j \in [R]$, define
	\begin{equation}					\label{eqn:tilde-g}
		\what{\cG}^x_j:= \cG^x_j,
		\ \ \ \ \ \ \textnormal{and} \ \ \ \ \ \ 
		\what{\cG}^z_j[i] :=  
		\begin{cases}
			1 & \mbox{ if } i = 0, \\
			\tilde{g}_{z,i}(j) &\mbox{ if } i \in [r].
		\end{cases}
	\end{equation}
	Then we have that
	\[
	\Ex_{\cG} \left[\Psi(Q(\cG))\right] \leq 2^r\Ex_{\what{\cG}} \left[\Psi(Q(\what{\cG}))\right] + \mu^r.
	\]
\end{lemma}

Before, we prove the above lemma, we need to setup some notation and observations. Denote $\cG = (\cG^x, \cG^z)$ where $\cG^x := (\cG^x_1,\ldots,\cG^x_R)$, and similarly $\cG^z := (\cG^z_1,\ldots,\cG^z_R)$. Then, note that by construction $\cG^x$ and $\cG^z$ satisfy the following properties:
\begin{itemize}
	\item[$(a)$] The ensemble sequences $\cG^x$ and $\cG^z$ are independent w.r.t. each other, and therefore for any fixing  of $\cG^x$, the distribution of $\cG^z$ is unchanged. 
	\item[$(b)$] By construction, the (vector) random variables $g_{z,1},\ldots,g_{z,r}$ from $\cG^z$ are independent $\rho$-correlated copies of a Gaussian vector $g$ (Observation \ref{obs:g-vec}).
\end{itemize}
Also observe that by definition, for every $i \in [r]$, we have that $\Pi_{[0,1]}(Q_i(\cG))$ is actually just a function of the Gaussian random vectors $g_{x,i}$ and $g_{z,i}$ -- this can be formalized as follows. For every $i \in [r]$, we define the function $P_i :\mathbbm{R}^R \times \mathbbm{R}^R \to [0,1]$ as:
\[
P_i(a,b) =  \Pi_{[0,1]}\left(\sum_{S,T \subseteq [R]}\wh{h'_i}(S,T) \prod_{j \in S} a(j)\prod_{j' \in T} b(j')\right).
\]
Then note that by definition of $Q_i$ (from \eqref{eqn:Q-def1}) and $\cG$ (from \eqref{eqn:cg-def}), we have that the following identity:
\[
P_i(g_{x,i},g_{z,i}) = \Pi_{[0,1]}\left(\sum_{S,T \subseteq [R]} \wh{h'_i}(S,T) \prod_{j \in S} \cG^x_j[i]\prod_{j' \in T} \cG^z_{j'}[i]\right) = \Pi_{[0,1]}(Q_i(\cG)).
\]
Next, for any fixing of $g_{x,i}$, let $P^{(i)}_{g_{x,i}}:\mathbbm{R}^R \to [0,1]$ denote the restriction of $P_i$ to the realization of $g_{x,i}$ i.e., $P^{(i)}_{g_{x,i}}(b) = P_i(g_{x,i},b)$. Using these notations and the above observations, we can rewrite:
\begin{align}						
	\Ex_{\cG} \left[\Psi(Q(\cG))\right] = \Ex_{\cG} \left[\prod_{i \in [r]} \Pi_{[0,1]}\left(Q_i(\cG)\right)\right]
	&= \Ex_{\cG^x}\Ex_{\cG^z}\left[\prod_{i \in [r]} P_{i}(g_{x,i},g_{z,i})\right] 		\tag{Using defn. of $\Psi(\cdot)$}		\non \\
	&= \Ex_{\cG^x}\Ex_{\cG^z}\left[\prod_{i \in [r]} P^{(i)}_{g_{x,i}}(g_{z,i})\right]   				\non \\
	&= \Ex_{\cG^x}\Ex_{g \sim N(0,I_R)}\Ex_{g_{z,1},\ldots,g_{z,r} \underset{\rho}{\sim}g}\left[\prod_{i \in [r]} P^{(i)}_{g_{x,i}}(g_{z,i})\right],
	\label{eqn:re-step}
\end{align}
where the second step uses item $(a)$ and the last step uses item $(b)$. Now we state and prove the following key lemma which bounds the inner expectation over $\cG^z$ for each fixing of $\cG^x$.
\begin{lemma}						\label{lem:gauss-avg}
	Fix a realization of $\cG^x$, and let $g_{x,1},\ldots,g_{x,r}$ be the vector-variables corresponding to $\cG^x$. Furthermore, for any $i \in [r]$, let $P^{(i)}_{g_{x,i}}$ be the restriction of $P_i$ to $g_{x,i}$ as defined above. Then,
	\[
	\Ex_{g \sim N(0_R,I_R)}\Ex_{g_{z,1},\ldots,g_{z,r} \underset{\rho}{\sim}g}\left[\prod_{i \in [r]} P^{(i)}_{g_{x,i}}(g_{z,i})\right] \leq 2^{r}\prod_{i \in [r]} \Ex_{\cG^z}\left[P^{(i)}_{g_{x,i}}(g_{z,i})\right] + \mu^r.
	\]
\end{lemma}
\begin{proof}
	We first observe that the vector random variables $g_{z,1},\ldots,g_{z,r}$ satisfy the following properties:
	\begin{itemize}
		\item For every $i \in [r]$, $g_{z,i}$ is marginally distributed as $N(0_R,I_R)$ (Observation \ref{obs:g-vec}).
		\item Using \eqref{eqn:cov-2}, and that $(g_{z,i}(\ell))_{i \in [r]}$ are independent across $\ell \in [R]$, for every $i \neq i'$ we have
		\[
		{\rm Cov}(g_{x,i},g_{x,i'}) = \Ex\left[g_{x,i}\left(g_{x,i'}\right)^\top\right] = \rho^2 I_R.
		\]
	\end{itemize}
	Furthermore, $P^{(1)}_{g_{x,1}},\ldots,P^{(r)}_{g_{x,r}}$ are all $[0,1]$-valued functions. Hence using the Exchangeable Gaussians Theorem (Corollary \ref{corr:egt}), we get that 
	\begin{equation}				\label{eqn:lambda}
	\Ex_{g \sim N(0_R,I_R)}\Ex_{g_{z,1},\ldots,g_{z,r} \underset{\rho}{\sim} g}\left[\prod_{i \in [r]} P^{(i)}_{g_{x,i}}(g_{z,i})\right]
	\leq \Lambda_\rho\Big(\delta_{1},\ldots,\delta_{r}\Big),
	\end{equation}
	where $\delta_i := \Ex_{\cG^z}P^{(i)}_{g_{x,i}}(g_{z,i})$ for every $i \in [r]$. We now proceed to further upper bound the above RHS. Let $\delta_\ell$ be the minimum among $\delta_1,\ldots,\delta_r$. We consider two cases depending on the value of $\delta_\ell$.
	
	{\bf Case (i)} Suppose $\delta_\ell \leq \mu^r$. Then,
	\begin{align*}
		\Lambda_\rho\Big(\delta_{1},\ldots,\delta_{r}\Big)
		&= \Pr_{g \sim N(0,1)} \Pr_{g_1,\ldots,g_r \underset{\rho}{\sim} g}\left[\forall j \in [r], g_j \leq \Phi^{-1}(\delta_j)\right] \\
		&\leq \Pr_{g \sim N(0,1)} \Pr_{g_\ell \underset{\rho}{\sim} g}\left[g_\ell \leq \Phi^{-1}(\delta_\ell)\right] \\
		&= \delta_\ell \leq \mu^r.
	\end{align*}
		
	{\bf Case (ii)} Suppose $\delta_\ell > \mu^r$.  Then our choice of $\rho$ implies that
	\[
	\rho = \frac{1}{4r^2 \log(1/\mu)} = \frac{1}{4r\log(1/\mu^r)} \leq \frac{1}{4r\log(1/\delta_\ell)}. 
	\]
	Then, using Lemma \ref{lem:hspace} we can bound,
	\[
	\Lambda_\rho(\delta_1,\ldots,\delta_r) \leq 2^r \prod_{i \in [r]} \delta_i.
	\] 
	Combining \eqref{eqn:lambda} with the bounds from the two cases we get that,
	\begin{align*}
	\Ex_{g \sim N(0,I_R)}\Ex_{g_{z,1},\ldots,g_{z,r} \underset{\rho}{\sim}g}\left[\prod_{i \in [r]} P^{(i)}_{g_{x,i}}(g_{z,i})\right]
	&\leq 2^r \prod_{i \in [r]} \delta_i + \mu^r \\
	&= 2^r \prod_{i \in [r]} \Ex_{\cG^z}\left[P^{(i)}_{g_{x,i}}(g_{z,i})\right] + \mu^r,
	\end{align*}	
	which completes the proof of the lemma.
\end{proof}

Using the above lemma, we can now prove Lemma \ref{lem:prod-2}
\begin{proof}[Proof of Lemma \ref{lem:prod-2}]
Recall that using the computation from \eqref{eqn:re-step} we have that 
\[
\Ex_{\cG} \left[\prod_{i \in [r]} \Pi_{[0,1]}\left(Q_i(\cG)\right)\right] = \Ex_{\cG^x}\Ex_{g \sim N(0,I_R)}\Ex_{g_{z,1},\ldots,g_{z,r} \underset{\rho}{\sim}g}\left[\prod_{i \in [r]} P^{(i)}_{g_{x,i}}(g_{z,i})\right].
\]
We continue with bounding the RHS from above. Since $\cG^x$ and $\cG^z$ are independent, we can apply Lemma \ref{lem:gauss-avg} to the inner expectation for each fixing of $\cG^x$ and get that:
\begin{align}
	\Ex_{\cG^x}\Ex_{g \sim N(0,I_R)}\Ex_{g_{z,1},\ldots,g_{z,r} \underset{\rho}{\sim}g}\left[\prod_{i \in [r]} P^{(i)}_{g_{x,i}}(g_{z,i})\right]
	&\leq 2^r\Ex_{\cG^x}\left[\prod_{i \in [r]} \Ex_{\cG^z} P^{(i)}_{g_{x,i}}(g_{z,i})\right] + \mu^r    \non\\
	&= 2^r\Ex_{\cG^x}\left[\prod_{i \in [r]}\Ex_{\cG^z} P_i(g_{x,i},g_{z,i}) \right] + \mu^r		\non\\
	&= 2^r\Ex_{\cG^x}\Ex_{\what{\cG}^z}\left[\prod_{i \in [r]} P_i(g_{x,i},\tilde{g}_{z,i}) \right] + \mu^r		\non\\
	&= 2^r \Ex_{\cG^x}\Ex_{\what{\cG}^z}\left[\prod_{i \in [r]}\Pi_{[0,1]}\left(Q_i(\cG^x,\what{\cG}^z)\right)\right] + \mu^r	\non \\
	&= 2^r \Ex_{\what{\cG}}\left[\prod_{i \in [r]} \Pi_{[0,1]}(Q_i(\what{\cG}))\right] + \mu^r  	\non\\
	&= 2^r \Ex_{\what{\cG}}\left[\Psi(Q(\what{\cG}))\right]	+ \mu^r,	 \non
\end{align}
which concludes the proof.
\end{proof}
\subsection{Step 4: Back to Boolean Space} 

Now we shall again use Theorem \ref{thm:clt} to shift the analysis back to the Boolean setting using arguments similar to Step 2. 
\begin{lemma}				\label{lem:prod-3}
	Let $\what{\cG}$ be the Gaussian ensemble sequence defined in \eqref{eqn:tilde-g}. Then,
	\[
	\Ex_{\what{\cG}}\Big[\Psi(\what{\cG})\Big]
	\leq \Ex_{(x_i)^r_{i=1} \sim \theta^{R}_e}\left[\prod_{i \in [r]}\bar{h}_{i}(x_i)\right] + C_r\sqrt{r}\tau^{O(\eta\kappa/r^2)},
	\]
where $\bar{h}_i(x_i) = \Ex_{z \sim \{\bot,\top\}^R_\beta}\left[h'_i(x_i,z) \right]$ for every $i \in [r]$.
\end{lemma}
\begin{proof}
For moving back, we will define a sequence of ensembles $\what{\cW} := (\what{\cW}^x_j,\what{\cW}^z_j)_{j \in [R]}$ as follows. Let $\tilde{z}_1,\ldots,\tilde{z}_r$ be a collection of independent random variables where $\tilde{z}_i \sim \{\bot,\top\}^R_\beta$ for every $i \in [r]$. In other words, $(\tilde{z}_i)_{i \in [r]}$ consists of $rR$ independent random variables, each of which is distributed as $\{\bot,\top\}_\beta$. Then, for every $j \in [R]$ define $\what{\cW}^x_j := \cW^x_j $ and 
\[
\what{\cW}^z_j[i] := 
\begin{cases}
	1 & \mbox{ if } i = 0, \\
	\phi^{(\beta)}(\tilde{z}_i(j)) & \mbox{ if } i \in [r].
\end{cases}
\]
Let $\what{\cG}$ be the Gaussian ensemble constructed in \eqref{eqn:tilde-g}. Then it is straightforward to verify that $\what{\cW}$ matches the covariance structure of $\what{\cG}$. Using identical arguments as in the proof of Lemma \ref{lem:prod-1}, we can see that conditions $C1$-$C4$ from Lemma \ref{lem:prod-1} again hold with respect to polynomials $Q_1,\ldots,Q_r$, and ensemble sequences $\what{\cW}$ and $\what{\cG}$. Therefore, applying Theorem \ref{thm:clt}, we again get that 
\begin{align}
	\Ex_{\what{\cG}}\left[\Psi(Q(\what{\cG}))\right] 
	&\leq \Ex_{\what{\cW}}\left[\Psi(Q(\what{\cW}))\right] + C_r\sqrt{r}\tau^{O(\eta\kappa/r^2)} 				\non\\
	&= \Ex_{\what{\cW}}\left[\prod_{i \in [r]} \Pi_{[0,1]}(Q_i(\what{\cW}))\right] + C_r\sqrt{r}\tau^{O(\eta\kappa/r^2)} 		\non\\
	&\overset{1}{=} \Ex_{(x_i)^r_{i=1} \sim \theta^{R}_e}\Ex_{\tilde{z}_1,\ldots,\tilde{z}_r \sim \{\bot,\top\}^R_\beta}\left[\prod_{i \in [r]} h'_i(x_i,\tilde{z}_i)\right] + C_r\sqrt{r}\tau^{O(\eta\kappa/r^2)}			\non \\
	&\overset{2}{=} \Ex_{(x_i)^r_{i=1} \sim \theta^{R}_e}\left[\prod_{i \in [r]} \Ex_{\tilde{z}_i \sim \{\bot,\top\}^R_\beta} h'_i(x_i,\tilde{z}_i)\right] + C_r\sqrt{r}\tau^{O(\eta\kappa/r^2)} 			\non\\
	&\overset{3}{=} \Ex_{(x_i)^r_{i=1} \sim \theta^{R}_e}\left[\prod_{i \in [r]}\bar{h}_{i}(x_i)\right] + C_r\sqrt{r}\tau^{O(\eta\kappa/r^2)}.	\non	
\end{align}
Here in step $1$ we substitute the expression for $Q_i(\what{\cW})$ and use Observation \ref{obs:h-exp} to get that 
\begin{align*}
Q_i(\what{\cW}) 
&= \sum_{S,T \subseteq [R]} \wh{h'_i}(S,T) \prod_{j \in S} \what{\cW}^x_j[i]\prod_{j' \in T} \what{\cW}^z_{j'}[i]    \\
&= \sum_{S,T \subseteq [R]} \wh{h'_i}(S,T) \prod_{j \in S} \phi^{(\mu_i)}(x_i(j))\prod_{j' \in T} \phi^{(\beta)}(\tilde{z}_{i}(j')) = h'_i(x_i,\tilde{z}_i).
\end{align*}
Step $2$ follows using that $\tilde{z}_1,\ldots,\tilde{z}_r$ are all independent, and step $3$ follows using the definition $\bar{h}_i(x) = \Ex_{z \sim \{\bot,\top\}^R_\beta}\left[h'_i(x,z) \right] = \Ex_{z \sim \{\bot,\top\}^R_\beta}\left[h_i(x,z)\right]$. 
\end{proof}

\subsection{Finishing the proof}

Now we put together the lemmas from the previous sections to finish the proof of Lemma \ref{lem:prod}:
\begin{align*}
	&\Ex_{(x_i,z_i)_{i \in [r]} \sim \cD^{R}_e}\left[\prod_{i \in [r]} h'_i(x_i,z_i)\right] \\
	&=  \Ex_{\cW} \Big[\Psi(Q(\cW))\Big]				\tag{Eq. \eqref{eqn:expr-0}} \\
	&\leq \Ex_{\cG} \Big[\Psi(Q(\cG))\Big] + C_r\sqrt{r}\tau^{O(\eta\kappa/r^2)}
	\tag{Lemma \ref{lem:prod-1}}\\
	&\leq 2^r\Ex_{\what{\cG}}\left[\Psi(Q(\what{\cG}))\right] + \mu^r + C_r\sqrt{r}\tau^{O(\eta\kappa/r^2)} 
	\tag{Lemma \ref{lem:prod-2}} \\
	&\leq 2^r \Ex_{(x_i)^r_{i=1} \sim \theta^{R}_e}\left[\prod_{i \in [r]}\bar{h}_{i}(x_i)\right] + \mu^r + 2C_r\sqrt{r}\tau^{O(\eta\kappa/r^2)}. 	
	\tag{Lemma \ref{lem:prod-3}}
\end{align*}

%% file: decoding.tex
\section{SSEH Decoding Lemmas}				\label{sec:dec}

{\bf Influence Decoding in SSEH}. The following lemma says that if there exists a family of functions defined on the noisy graph $G^{\otimes R}_\eta$ such that a constant fraction of the functions have influential coordinates, then it can be used to decode a small non-expanding set in $G$.

\begin{lemma}[Lemma 10.1 \cite{GL22arXiv},\cite{RST12}]						\label{lem:sse-dec}
	Let $(\Omega^R,\gamma^R)$ be a product probability space. Let $\{f_A\}_{A \in V^R}$ be a set of functions such that $f_A:\Omega^R \to [0,1]$. Furthermore, suppose the class of functions are permutation respecting i.e., for every $A \in V^R$, $\omega \in \Omega^R$ and permutation $\pi:[R] \to [R]$ we have $f_{\pi(A)}(\pi(\omega)) = f_A(\omega)$. For every $A \in V^R$, define the averaged functions $g_A := \Ex_{B \sim G^{\otimes R}_\eta(A)} f_B$. Then if $G$ is a NO instance of $(\epsilon,\delta,M)$-\smallsetexpansion, then 
	\[
	\Pr_{A \sim V^R}\left[\max_{i \in [R]} \Inf{i}{\sT^{(\Omega)}_{1 - \eta}g_A} > \tau\right] \leq \frac{\nu^2}{8r}.
	\]
\end{lemma}
The proof of the above follows using the following lemma from \cite{RST12}.

\begin{lemma}[\cite{RST12}]				\label{lem:rst-dec}			
	There exists a constant $\eta_0 \in (0,1)$ such that the following holds for any $\eta \in (0,\eta_0]$ and $R$ large enough. Suppose $(A,B)$ be a distribution over pairs of vertices defined according to the following process: let $A \sim V^R$ and $B \sim G^{\otimes R}_{\eta}(A)$. Now, suppose there exists a function $F:V^R \to [R]$ such that it satisfies
	\begin{equation}				\label{eqn:match-prob}
	\Pr_{(A,B)}\Pr_{\pi_A,\pi_B \sim \mathbbm{S}_R}\left[\pi^{-1}_A(F(\pi_A(A))) = \pi^{-1}_A(F(\pi_B(B)))\right] \geq \zeta.
	\end{equation}
	Then there exists a subset $S \subseteq V$ such that ${\sf vol}(S) \in \left[\frac{\zeta}{16R}, \frac{3}{\eta R}\right]$ satisfying $\phi_S(G) \leq 1 - \zeta$.
\end{lemma}
Using the above, we can prove Lemma \ref{lem:sse-dec}.

\begin{proof}[Proof of Lemma \ref{lem:sse-dec}]
	Let $\{f_A\}_{A \in V^R}$ be the collection of functions in the setting of the lemma. For every $A \in V^R$, we introduce two sets $L_{A,1},L_{A,2} \subseteq [R]$ which are defined as follows.
	\[
	L_{A,1} := \left\{j \in [R]~\Big|~{\sf Inf}_j\left[\sT_{1 - \eta}f_A\right] \geq \frac{\tau}{2} \right\}
	\ \ \ \ \ \ \textnormal{ and } \ \ \ \ \ \
	L_{A,2} := \left\{j \in [R]~\Big|~{\sf Inf}_j\left[\sT_{1 - \eta}g_A\right] \geq \tau \right\}
	\]
	We can verify that $|L_{A,1}|,|L_{A,2}| \leq 2/\eta\tau$ (for e.g., \cite{KKMO07}). Let $V' \subset V^R$ be the subset of vertices for which 
	\[
	\max_{j \in [R]} {\sf Inf}^{(\Omega)}_j\left[\sT_{1 - \eta}g_A\right] \geq \tau,
	\] 
	and for every such choice of $A \in V'$, let $j_A$ be fixed element of $L_{A,2}$. Note that for any such choice of $A \in V'$ we have 
	\[
	\tau \leq {\sf Inf}^{(\Omega)}_{j_A}\left[\sT_{1 - \eta}g_A\right] =  {\sf Inf}^{(\Omega)}_{j_A}\left[\sT_{1 - \eta}\Ex_{B \sim G^{\otimes R}_{1 - \eta}(A)}f_B\right] \leq \Ex_{B \sim G^{\otimes R}_{1 - \eta}(A)} {\sf Inf}^{(\Omega)}_{j_A}\left[\sT_{1 - \eta}f_B\right],
	\]
	where the inequality follows using the convexity of influences. Therefore, for any such choices of $A$, for at least $\tau/2$ fraction of choices of $B \sim G^{\otimes R}_{\eta}(A)$, we have ${\sf Inf}^{(\Omega)}_{j_A}\left[\sT_{1 - \eta}f_B\right] \geq \tau/2$; call this set of vertices $V(A)$. Now consider the following distribution over functions $F:V^R \to [R]$: for every $A \in V^R$, do the following independently.
	\begin{itemize}
		\item W.p. $1/2$, if $L_{A,1} \neq \emptyset$, assign $F(A) \sim L_{A,1}$ u.a.r., otherwise assign $F(A)$ arbitrarily.
		\item W.p. $1/2$, if $L_{A,2} \neq \emptyset$, assign $F(A) \sim L_{A,2}$ u.a.r., otherwise assign $F(A)$ arbitrarily.
	\end{itemize}
	Now let us analyze the probability of the event in \eqref{eqn:match-prob} w.r.t above distribution over assignments. To that end, observe that for $A \sim V^R$, with probability at least $\nu^2/8r$, we have $A \in V'$. Fixing such a choice of $A \in V'$, note that for $B \sim G^{\otimes R}_{\eta}(A)$, $B \in V(A)$ with probability at least $\tau/2$. Fixing such a choice of $(A,B)$, note that $\{j_A\} \in L_{A,2} \cap L_{B,1}$. Furthermore, since the functions $\{f_{A'}\}$ and $\{g_{A'}\}$ are permutation respecting, for any pair of permutations $\pi_A,\pi_B$, we have that $\pi_{A}(j_A) \in L_{\pi_A(A),2}$ and $\pi_B(j_A) \in L_{\pi_B(B),1}$. Therefore fixing permutations $\pi_A,\pi_B$ , randomizing over the choices of $F$, with probability at least $\eta^2\tau^2/16$, we have $F(\pi_{A}(A)) = \pi_{A}(j_A)$, and $F(\pi_B(B)) = \pi_B(j_A)$, and hence $\pi^{-1}_A(F(\pi_A(A))) = \pi^{-1}_B(F(\pi_B(B)))$. Therefore, there exists a choice of $F$ such that 
	\[
	\Pr_{A,B} \Pr_{\pi_A,\pi_B}\left[\pi^{-1}_A(F(\pi_A(A))) = \pi^{-1}_B(F(\pi_B(B)))\right] \geq \frac{\nu^2\eta^2\tau^3}{64},
	\]
	which using Lemma \ref{lem:rst-dec} implies that there exist a set $S \subset V$ such that ${\sf vol}(S) \in [\nu^2/16R,3/\eta R]$ such that $\phi_G(S) \leq 1 - \nu^2$, which contradicts the fact that $G$ is a NO instance of $(\epsilon,\delta,M)$-SSE.
\end{proof}

\subsection{Proof of Lemma \ref{lem:dec-set0}}				\label{sec:dec-set0}

For contradiction, let us assume that
\[
\Pr_{A \sim V^R}\left[\Pr_{e \sim E_{\rm gap}}\left[\max_{i \in e} \max_{j \in [R]} \Inf{j}{\sT^\Pom_{1 - \eta}g_{A,i}} > \tau\right] > \nu\right] \geq \nu.
\]
Then it follows that
\[
\Pr_{A \sim V^R,e \sim E_{\rm gap}}\left[\max_{i \in e} \max_{j \in [R]} \Inf{j}{\sT^\Pom_{1 - \eta}g_{A,i}} > \tau \right] \geq \nu^2,
\]
which in turn implies that there exists a choice of an edge $e \in E_{\rm gap}$ for which 
\[
\Pr_{A \sim V^R}\left[\max_{i \in e} \max_{j \in [R]} \Inf{j}{\sT^\Pom_{1 - \eta}g_{A,i}} > \tau \right] \geq \nu^2.
\]
Further averaging over the choice of $i \in e$, it follows that there exists a vertex $i \in e$ such that 
\begin{equation}					\label{eqn:lb-1}
	\Pr_{A \sim V^R}\left[ \max_{j \in [R]} \Inf{j}{\sT^\Pom_{1 - \eta}g_{A,i}} > \tau \right] \geq \frac{\nu^2}{r}.
\end{equation}
Now consider the class of functions $\tilde{f}_{A}:\{0,1\}^R \times \{\bot,\top\}^R \to [0,1]$ defined in $L_2(\Omega^R_i)$ as 
\[
\tilde{f}_{A}(x,z) = \Ex_{(A',x') \sim M^{(\mu_i)}_{z}(B,x)}\Ex_{\pi \sim \mathbbm{S}_R}\left[f(\pi(A',x',z))\right]. 
\]	
Then it is easy to see that for every permutation $\pi':[R] \to [R]$ we have 
\begin{align*}
	\tilde{f}_{\pi'(A)}(\pi(x,z)) 
	&= \Ex_{(A',x') \sim M^{(\mu_i)}_{\pi'(z)}(\pi'(A),\pi'(x))}\Ex_{\pi \sim \mathbbm{S}_R}\left[f(\pi(A',x',\pi'(z)))\right]  \\
	&= \Ex_{(A',x') \sim M^{(\mu_i)}_{z}(A,x)}\Ex_{\pi \sim \mathbbm{S}_R}\left[f(\pi(\pi'(A'),\pi'(x'),\pi'(z)))\right]  \\
	&= \Ex_{(A',x') \sim M^{(\mu_i)}_{z}(A,x)}\Ex_{\pi \sim \mathbbm{S}_R}\left[f(\pi(A',x',z))\right]  \\
	&= \tilde{f}_{A}(x,z),
\end{align*}
i.e., the set of functions $\{\tilde{f}_A\}_{A \in V_R}$ (and consequently, the functions $\{\sT^{(\Omega_i)}_{1 - \eta}\tilde{f}_A\}_{A \in V^R}$) are also permutation respecting as in the statement of Lemma \ref{lem:sse-dec}. Furthermore, we have $g_{A,i} := \Ex_{B \sim G^{\otimes R}_\eta(A)}\tilde{f}_B$ for every $A \in V^R$ by definition. Since $G$ is a NO instance, instantiating Lemma \ref{lem:sse-dec} with $\Omega = \Omega^R_i$ and $f_A = \tilde{f}_A$ for every $A \in V^R$ we get that 
\[
\Pr_{A \sim V^R}\left[\max_{i \in [R]} \Inf{i}{\sT^{(\Omega)}_{1 - \eta}g_{A,\mu_i}} > \tau\right] \leq \frac{\nu^2}{8r},
\]	 
which contradicts \eqref{eqn:lb-1}, thus concluding the proof.

\subsection{Proof of Lemma \ref{lem:dec-set1}}				\label{sec:dec-set1}

Suppose for contradiction, 
\[
\Pr_{A \sim V^R}\left[\Pr_{e \sim E_{\rm gap}}\left[\max_{i \in e} \max_{j \in [R]} \Inf{j}{\sT^{(\mu_i)}_{1 - \eta}\bar{g}_{A,i}} > \tau\right] > \nu\right] \geq \nu/2.
\]
Again, following the arguments from the proof of Lemma \ref{lem:dec-set0} (Section \ref{sec:dec-set0}), we can find a choice of $i \in e$ for some $e \in E_{\rm gap}$ for which the following holds:
\begin{equation}				\label{eqn:lb-2}
	\Pr_{A \sim V^R}\left[\max_{j \in [R]} \Inf{j}{\sT^{(\mu_i)}_{1 - \eta}\bar{g}_{A,i}} > \tau\right] \geq \frac{\nu^2}{2r}.
\end{equation}
On the other hand, consider the set of functions $\{\bar{f}_A\}_{A \in V^R}$ defined on $L_2(\{0,1\}^R_{\mu_i})$ as follows. For every $A \in V^R$,
\[
\bar{f}_A(x) 
= \Ex_{z \sim \{\bot,\top\}^R_\beta} \Ex_{(A',x') \sim M^{(\mu_i)}_{z}(A,x)}\Ex_{\pi \sim \mathbbm{S}_R}\left[f(\pi(A',x',z))\right].
\]
Then note that for any permutation $\pi':[R] \to [R]$ we have
\begin{align*}
	\bar{f}_{\pi'(A)}(\pi'(x)) 
	&= \Ex_{z \sim \{\bot,\top\}^R_\beta} \Ex_{(A',x') \sim M^{(\mu_i)}_{z}(\pi'(A),\pi'(x))}\Ex_{\pi \sim \mathbbm{S}_R}\left[f\Big(\pi(A',x',z)\Big)\right] \\
	&= \Ex_{z \sim \{\bot,\top\}^R_\beta} \Ex_{(A',x') \sim M^{(\mu_i)}_{\pi'(z)}(\pi'(A),\pi'(x))}\Ex_{\pi \sim \mathbbm{S}_R}\left[f\Big(\pi(A',x',\pi'(z))\Big)\right] \\
	&= \Ex_{z \sim \{\bot,\top\}^R_\beta} \Ex_{(A',x') \sim M^{(\mu_i)}_{z}(A.x)}
	\Ex_{\pi \sim \mathbbm{S}_R}\left[f\Big(\pi\Big(\pi'(A'),\pi'(x'),\pi'(z)\Big)\Big)\right] \\
	&= \Ex_{z \sim \{\bot,\top\}^R_\beta} \Ex_{(A',x') \sim M^{(\mu_i)}_{z}(A,x)}\Ex_{\pi \sim \mathbbm{S}_R}\left[f\Big(\pi(A',x',z)\Big)\right] \\	
	&= \tilde{f}_{A}(x).
\end{align*}
Therefore, the set of functions $\{\bar{f}_A\}_{A \in V^R}$ is permutation respecting again. Furthermore, by definition, it follows that for every $A \in V^R$, we have $\bar{g}_A = \Ex_{B \sim G^{\otimes R}_\eta(A)}\bar{f}_{B}$. Therefore, instantiating Lemma \ref{lem:sse-dec} with $f_A :=  \cT^{(\mu_i)}_{1 - \eta}\bar{f}_A$, we get that 
\[
\Pr_{A \sim V^R}\left[\max_{j \in [R]} \Inf{j}{\sT^{(\mu_i)}_{1 - \eta}\bar{g}_{A,i}} > \tau\right] < \frac{\nu^2}{8r},
\]
which again contradicts \eqref{eqn:lb-2}, and hence we must have
\[
\Pr_{A \sim V^R}\left[\Pr_{e \sim E_{\rm gap}}\left[\max_{i \in e} \max_{j \in [R]} \Inf{j}{\sT^{(\mu_i)}_{1 - \eta}\bar{g}_{A,i}} > \tau\right] > \nu\right] \leq \nu/2.
\] 
Using identical arguments, we can also prove that 
\[
\Pr_{A \sim V^R}\left[\Pr_{i \sim E_{\rm gap}}\left[\max_{j \in [R]} \Inf{j}{\sT^{(\mu_i)}_{1 - \eta}\bar{g}_{A,i}} > \tau\right] > \nu\right] \leq \nu/2.
\]
The claim now follows by taking a union bound.

%% file: rag-round.tex
\section{Proof of Lemma \ref{lem:round}}				\label{sec:round}

Here we provide a proof of Lemma \ref{lem:round} detailing the changes needed in the proof of soundness analysis in Theorem 6.2 from \cite{RT12}, including the improved analysis of the variance bound (Lemma \ref{lem:r-bias}). To begin with, we describe the rounding scheme from \cite{RT12} adapted to our setting in Figure \ref{fig:round}:

\begin{figure}[ht!]
	\begin{mdframed}
		{\bf Input}: A collection of functions $\{g_i\}_{i \in V_{\rm gap}}$ where $g_i:\{0,1\}^R \to [0,1]$ is defined on the probability space $\{0,1\}^R_{\mu_i}$. Furthermore, the functions satisfy the global constraints:
		\begin{equation}				\label{eqn:r-cond1}
			\Ex_{i \sim V_{\rm gap}}\Ex_{x_i \sim \{0,1\}^R_{\mu_i}}\left[g_i(x_i)\right] = \mu,
		\end{equation}
		\begin{equation}				\label{eqn:r-cond2}
			\Pr_{i \sim G_{\rm gap}} \left[\max_{j \in [R]} \Inf{j}{\sT^{(\mu_i)}_{1 - \eta}g_i} > \tau\right] \leq \nu.
		\end{equation}
		\begin{equation}				\label{eqn:r-cond3}
			\Pr_{e \sim E_{\rm gap}} \left[ \max_{i \in e}\max_{j \in [R]} \Inf{j}{\sT^{(\mu_i)}_{1 - \eta}g_i} > \tau\right] \leq \nu.
		\end{equation}
	
		{\bf Setup}: 
		\begin{itemize}
		\item Let $\{u_i\}_{i \in V_{\rm gap}} \cup \{\uphi\}$ be the $d$-dimensional vector solution corresponding to $\{\theta_e\}_{e \in E_{\rm gap}}$ such that $u_i = \mu_i \uphi + w_i$ for every $i \in V_{\rm gap}$. 
		\item For every $i \in V_{\rm gap}$, let $H_{i}:\mathbbm{R}^R \to \mathbbm{R}$ be the multi-linear polynomial representation of $\sT^{(\mu_i)}_{1 - \eta}g_i$ in the variables $(x_i(j))_{j \in [R]}$ (from Fact \ref{fact:mult})\\
		\item Let $f_{[0,1]}:\mathbbm{R} \to [0,1]$ be the truncation function defined as
		\[
		f_{[0,1]}(x) := 
		\begin{cases}
			0 & \mbox{ if } x < 0 \\
			x & \mbox{ if } x \in [0,1] \\
			1 & \mbox{ if } x > 1
		\end{cases}
		\]
		\end{itemize}
		{\bf Rounding}: \\[1pt]
		\begin{enumerate}
			\item Sample Gaussian matrix ${\bf G} \sim N\left(0,I_{R \times d}\right)$, where $d$ is the ambient dimension of the vector solution. 
			\item For every $i \in V_{\rm gap}$, define $q_i = \mu_i \cdot {\bf 1}_R + {\bf G} \cdot {w}_i$.
			\item For every $i \in V_{\rm gap}$, compute $p_i = f_{[0,1]}\left(H_{i}(q_i)\right)$.
			\item Sample a random assignment $\sigma:V_{\rm gap} \to \{0,1\}$, by sampling $\sigma(i) \sim \{0,1\}_{p_i}$ independently for every $i \in V_{\rm gap}$.   
		\end{enumerate}
	\end{mdframed}
	\caption{Rounding Scheme}
	\label{fig:round}
\end{figure}

We point out a subtle difference between the rounding scheme described above and the one from \cite{Rag08}: here the different vertices $i \in V_{\rm gap}$ can have different rounding functions $H_i$, whereas in \cite{Rag08}, it suffices to work with a single rounding function. However, as we will see, this change does not cause additional issues in the analysis. The proof of Lemma \ref{lem:round} now follows immediately from the following lemmas which we prove in Sections \ref{sec:r-bias} and \ref{sec:r-round}.

\begin{lemma}						\label{lem:r-bias}
	Suppose the set of local distributions $\theta := \{\theta_e\}_e$ has average-correlation at most $\gamma^2$. Then, with probability at least $1 - 2\sqrt{\gamma}$ we have that 
	\[
	\left|\Ex_{i \sim V_{\rm gap}}\big[ \sigma(i)\big] - \mu \right| \leq 2\mu\sqrt{\gamma}.
	\]
\end{lemma}

\begin{lemma}				\label{lem:r-round}
	Suppose the functions $\{g_i\}_{i \in V_{\rm gap}}$ satisfy \eqref{eqn:r-cond3}. Then,
	\[
	\Ex_{\sigma}\Ex_{e \sim E_{\rm gap}}\left[\prod_{i \in e} \sigma(i)\right]
	\geq \Ex_{e \sim E_{\rm gap}}\Ex_{(x_i)_{i \in e} \sim \theta^R_e}\left[\prod_{i \in e} g_i(x_i)\right] - C_r\sqrt{r}\tau^{C\eta\kappa/r^2}.
	\]
\end{lemma}

\begin{proof}[Proof of Lemma \ref{lem:round}]
	The proof follows directly by combining the guarantees of Lemma \ref{lem:r-bias} and Lemma \ref{lem:r-round}.
\end{proof}

\subsection{Proof of Lemma \ref{lem:r-round}}				\label{sec:r-round}

Before we begin, we introduce some additional notation that will be used in the proof. For every edge $e \in E_{\rm gap}$, we introduce an independent sequence of  ensembles $\cX^e := (\cX^e_{j})_{j \in [R]}$, where for every $j \in [R]$, we have $\cX^e_j$ consisting of $\{1\} \cup \{x_i(j)\}_{i \in e}$. As before, we shall use $x_0(j) \equiv 1$ to denote the constant function of the ensemble $\cX^e_j$ for convenience. Analogously, we also define the Gaussian ensemble sequence $\cQ^e := (\cQ^e_j)_{j \in [R]}$, where for every $j \in [R]$ we have that $\cQ^e_j := \{1\} \cup \{q_i(j)\}_{i \in e}$.
The first step of the proof is the following useful observation. 

\begin{observation}			\label{obs:mult}
	For edge $e \in E_{\rm gap}$ and vertex $i \in e$, the polynomial $H_i(x_i)$ is multi-linear in $\cX^e$. Similarly, $H_i(q_i)$ is a multi-linear polynomial in the ensemble sequence $\cQ^e$.
\end{observation}
\begin{proof}
	From Fact \ref{fact:mult} we know that $H_i$ can be expressed in terms of the variables $x_i$ as 
	\[
	H_i(x) = \sum_{S \subseteq [R]} \wh{H_i}(S) \prod_{j \in S}x_i(j) = \sum_{S \subseteq [R]}\wh{H_i}(S) \prod_{j \in S \setminus T} x_0(j) \prod_{j \in S} x_i(j),
	\]
	i.e., each monomial in $H_i$ consists of exactly one variable from $\cX^e_j$ for every $j \in [R]$, and hence $H_i$ is a multi-linear polynomial in $\cX^e$. The second part of the observation follows identically.
\end{proof}
  
Our next observation is that for $e \in E$, the ensemble sequences $\cX^e$ and $\cQ^e$ have up to matching second moments.

\begin{claim}			\label{cl:match}
	For any edge $e \in E$, the ensemble sequences $\cX^e$ and $\cQ^e$ have matching covariance structure. 
\end{claim}
\begin{proof}
	We first establish that the ensembles have matching first moments. Towards that, fix an edge-coordinate index pair $(i,j) \in e \times [R]$. Then,
	\[
	\Ex_{\cQ^e} \left[q_i(j)\right] = \Ex_{\zeta_j \sim N(0,I_d)} \left[\mu_i + \langle \zeta_j,{w}_i \rangle \right] = \mu_i = \Ex_{\cX^e}\left[x_i(j)\right].
	\]	
	Now towards verifying the matching second moments condition, for any choice of $(i,j),(i',j)$ we observe that
	\begin{align*}
		\Ex_{\cQ} \left[q_i(j)q_{i'}(j')\right] 
		&=\Ex_{\zeta_j}\left[\left(\mu_i + \langle \zeta_j, w_i \rangle\right)\left(\mu_{i'} + \langle \zeta_j,w_{i'} \rangle \right)\right] \\
		&=\mu_i \mu_{i'} + \langle w_i,w_{i'} \rangle \\
		&=\Pr_{X_e \sim \theta_e} \left[X_i = 1, X_{i'} = 1\right] 			\tag{Proposition \ref{prop:vector}}\\
		&=\Ex_{\cX} \left[x_i(j) x_{i'}(j)\right].
	\end{align*}
	Since the above holds for any $i,i' \in e$ and $j \in [R]$, along with the fact that $\cX^e_j[0] = \cQ^e_j[0] = 1$, this establishes that $\cX^e_j$ and $\cQ^e_j$ have matching covariances structure, The proof is concluded by arguing the above for all $j \in [R]$.
\end{proof}

Now by assumption on the functions $\{g_i\}_{i \in V_{\rm gap}}$ we have
\[
\Pr_{i \sim G_{\rm gap}} \left[\max_{j \in [R]} \Inf{j}{\sT^{(\mu_i)}_{1 - \eta} g_{i}} > \tau\right] \leq \nu.
\]
Let $E_{\rm nice} \subset E_{\rm gap}$ be the subset of edges for which we have 
\[
\max_{i \in e} \max_{j \in [R]} \Inf{j}{\sT^{(\mu_i)}_{1 - \eta} g_{i}} \leq \tau.
\]
The following is the key technical component of this lemma:
\begin{lemma}				\label{lem:round-nice}
	The following holds every $e=(i_1,\ldots,i_r) \in E_{\rm nice}$:
	\[
	\Ex_{(q_i)_{i \in e}} \left[ \psi\left(f_{[0,1]}\left(H_{i_1}(q_{i_1})\right),\ldots,f_{[0,1]}\left(H_{i_1}(q_{i_1})\right)\right)\right] \geq 
	\sum_{a \in \psi^{-1}(1)}\Ex_{(x_i)_{i \in e} \sim \theta^R_e} \left[\prod_{i \in e} \sT^{(\mu_i)}_{1 - \eta}g^{(a,e,i)}_i(x_i)\right] - C_r2^r\sqrt{r}\tau^{C\eta\kappa/r^2}.
	\]
\end{lemma}
\begin{proof}
	Our proof constitutes showing that for any $S \subseteq e$ we have 
	\begin{equation}				\label{eqn:lem}
		\left|\Ex_{(x_i)_{i \in e}}\left[\prod_{i \in S} \sT^{(\mu_i)}_{1 - \eta}g_i(x_i) \right]
		-\Ex_{\cQ}\left[\prod_{i \in S} f_{[0,1]}\left(H_i(q_i)\right)\right] \right| \leq C_r\sqrt{r}\tau^{C\eta\kappa/r^2}.
	\end{equation}
	We point out that the above immediately implies the lemma. To see this, for any fixed $a \in \psi^{-1}(1)$, let us partition $e = e^+_a \sqcup e^-_a$ where $e^+_a \subseteq e$ consists of the vertices $i \in e$ for which $a_i = 1$. Then using this definition, we have that 
	\begin{align*}
	&\Ex_{(q_i)_{i \in e}} \left[ \psi\left(f_{[0,1]}\left(H_{i_1}(q_{i_1})\right),\ldots,f_{[0,1]}\left(H_{i_1}(q_{i_1})\right)\right)\right] \\
	&= \sum_{a \in \psi^{-1}(1)}\Ex_{(q_i)_{i \in e}} \left[ \prod_{i \in e^+_{a}} f_{[0,1]}\left(H_i(q_i)\right) \prod_{i \in e^-_a}
	\left(1 - f_{[0,1]}\left(H_{i}(q_{i})\right)\right)\right] \\
	&= \sum_{a \in \psi^{-1}}\sum_{S \subseteq e^a_-}(-1)^{|e^a_- \setminus S|}\Ex_{(q_i)_{i \in e}}\left[\prod_{i \in S \cup e^a_+} f_{[0,1]}\left(H_{i}(q_{i})\right)\right] \\
	&\geq  \sum_{a \in \psi^{-1}}\sum_{S \subseteq e^a_-}(-1)^{|e^a_- \setminus S|}\Ex_{(x_i)_{i \in e}}\left[\prod_{i \in S \cup e^a_+} \sT^{(\mu_i)}_{1 - \eta}g_i(x_{i})\right] - 2^rC_r\sqrt{r}\tau^{C\eta\kappa/r^2} \\
	&= \sum_{a \in \psi^{-1}}\Ex_{(x_i)_{i \in e}}\left[\prod_{i \in  e^a_+} \sT^{(\mu_i)}_{1 - \eta}g_i(x_{i})
	\prod_{i \in  e^a_-}\left(1 - \sT^{(\mu_i)}_{1 - \eta}g_i(x_{i})\right)\right] - 2^rC_r\sqrt{r}\tau^{C\eta\kappa/r^2} \\
	&= \sum_{a \in \psi^{-1}(1)}s\Ex_{(x_i)_{i \in e}} \left[ \prod_{i \in e} \sT^{(\mu_i)}_{1 - \eta}g^{(a,e,i)}_i(x_{i})\right] - 2^rC_r\sqrt{r}\tau^{C\eta\kappa/r^2},
	\end{align*}
	which establishes the claim of the lemma.
	
	{\bf Establishing \eqref{eqn:lem}}. Fix a subset $S \subseteq e$. From observation \ref{obs:mult} we know that for every $i \in S$, $H_i$ is a multi-linear polynomial in $\cX^e:= (\cX^e_j)_{j \in [R]}$. Hence for every $i \in S$ and $j \in [R]$ we have that 
	\begin{align}
		{\sf Inf}_{\cX^e_j}\left[H_i\right] 
		&= \Ex_{(\cX^e_{j'})_{j' \neq j}}\left[{\sf Var}_{\cX^e_j}\left[H_i(\cX^e)~\big|~\cX^e_{j' \neq j}\right] \right] 		\non\\
		&= \Ex_{(x_{j'}(i))_{j' \neq j}}\left[{\sf Var}_{x_i(j)}\left[\sT^{(\mu_i)}_{1 - \eta}g_i(x_i)~\big|~(x_i(j'))_{j' \neq j}\right] \right] 		\non\\
		&= {\sf Inf}^{(\mu_i)}_j\left[\sT^{(\mu_i)}_{1 - \eta}g_i\right] 
		\leq \tau,				\label{eqn:inf-1}
	\end{align}
	where the last step follows due the fact that $e \in E_{\rm nice}$.
	
	Next, let $\Psi_S$ be the clipped-product function $\Psi_S(\ell) = \prod_{i \in S}f_{[0,1]}(\ell_i)$.
	
	{\bf Invariance Step}. As in the proof of Lemma \ref{lem:prod}, let us verify the conditions required to invoke Theorem \ref{thm:clt}.
	\begin{itemize}
		\item[$C_1$] From Claim \ref{cl:match}, $\cX^e$ and $\cQ^e$ have matching covariance structure.
		\item[$C_2$] From \eqref{eqn:inf-1}, for every $j \in [R]$ and $i \in S$ we have 
		\[
		{\sf Inf}^{(\mu_i)}_{\cX^e_j}\left[H_i\right] \leq \tau.
		\] 
		\item[$C_3$] Due to $(\gamma\mu)^r$-smoothness of the local distribution $\theta$ (Lemma \ref{lem:smooth}) we have $\Pr\left[\cX^e_j = \omega\right] \geq \gamma^r\mu^r$ for every atom $\omega$ in the event space of $\cX^e_j$
		\item[$C_4$] Since $\|g_i\|_{\infty} \leq 1$, we have ${\rm Var}[H_i(\cX^e)] \leq 1$ and using the Fourier decay property (Fact \ref{fact:decay-2}) we have ${\rm Var}\left[H^{\geq d}_i\right] \leq (1 - \eta/2)^d$ for $d \geq \frac{1}{18}\log(1/\tau)/\log(1/\alpha)$.
	\end{itemize}
	The above checks imply that the vector polynomial $H_S := (H_i)_{i \in S}$ along with the ensemble sequences $\cX^e$, $\cG_e$ satisfy the conditions required to invoke Theorem \ref{thm:clt}. Hence applying Theorem \ref{thm:clt} we get that 
	\begin{align*}
		\Ex_{\cQ^e}\left[\prod_{i \in S} f_{[0,1]}(H_i(Q_i))\right]
		= 	\Ex_{\cQ^e}\Big[\Psi_S(H_S(\cQ^e))\Big] 
		&\geq \Ex_{\cX^e}\Big[\Psi_S(H_S(\cX^e))\Big] - C_r\sqrt{r}\tau^{C\eta\kappa/r^2} 	\\
		&= \Ex_{(x_i)_{i \in e}}\left[\prod_{i \in S}\sT^{(\mu_i)}_{1 - \eta}g_{i}(x_i)\right] - C_r\sqrt{r}\tau^{C\eta\kappa/r^2} 	\\
	\end{align*}
	which establishes \eqref{eqn:lem}. Since we have already established that \eqref{eqn:lem} implies the lemma, this completes the proof. 
\end{proof}

{\bf Cleaning Up}. Now we use Lemma \ref{lem:round-nice} to finish the proof. Recall that in the setting of the lemma we have 
\[
\Pr_{e \sim E_{\rm gap}}\left[\max_{i \in e}\max_{j \in [R]}{\sf Inf}^{(\mu_i)}_{j}\left[\sT^{(\mu_i)}_{1 - \eta}g_{i}\right] > \tau\right] \leq \nu,
\]
i.e., $E_{\rm nice}$ has weight at least $1 - \nu$ in $E_{\rm gap}$. Combining this with Lemma \ref{lem:round-nice} we get that
\begin{align*}
	&\Ex_{\sigma}\Ex_{e = (i_1,\ldots,i_r) \sim E_{\rm gap}} \bigg[\psi\Big(\sigma(i_1),\ldots,\sigma(i_r)\Big)\bigg] \\
	& \geq \Ex_{\sigma}\Ex_{e = (i_1,\ldots,i_r) \sim E_{\rm gap}| e \in E_{\rm nice}} \bigg[\psi\Big(\sigma(i_1),\ldots,\sigma(i_r)\Big)\bigg] - \nu \\
	&= \Ex_{e = (i_1,\ldots,i_r) \sim E_{\rm gap}| e \in E_{\rm nice}}\Ex_{(q_i)_{i \in e}} 
	\bigg[\psi\left( f_{[0,1]}\left( H_{i_1}(q_{i_1})\right),\ldots,f_{[0,1]}\left(H_{i_r}(q_{i_r})\right) \right) \bigg] - \nu \\
	&\geq \Ex_{e \sim E_{\rm gap}|e \in E_{\rm nice}}\Ex_{(x_i)_{i \in e} \sim \theta^R_e} 
	\left[\sum_{a \in \psi^{-1}(1)} \prod_{i \in e} \sT^{(\mu_i)}_{1 - \eta} g^{(a,e,i)}_i(x_i) \right] - \nu - 2^rC_r\sqrt{r}\tau^{C\eta\kappa/r^2} \\
	&\geq \Ex_{e \sim E_{\rm gap}}\Ex_{(x_i)_{i \in e} \sim \theta^R_e} 
	\left[\sum_{a \in \psi^{-1}(1)} \prod_{i \in e} \sT^{(\mu_i)}_{1 - \eta} g^{(a,e,i)}_i(x_i) \right] - 2\cdot 2^r  \nu - 2^rC_r\sqrt{r}\tau^{\eta\kappa/r^2},
\end{align*}
where the second last inequality uses Lemma \ref{lem:round-nice}, and the last inequality uses the fact that for every fixed $a \in \psi^{-1}(1)$, we have 
\[
\Ex_{e \sim E_{\rm nice}|e \in E_{\rm gap}}\left[\Ex_{(x_i)_{i \in e} \sim \theta^R_e}\prod_{i \in e} \sT^{(\mu_i)}_{1 - \eta}g^{(a,e,i)}_i(x_i) \right] \geq  
\Ex_{e \sim E_{\rm gap}}\left[\Ex_{(x_i)_{i \in e} \sim \theta^R_e}\prod_{i \in e} \sT^{(\mu_i)}_{1 - \eta}zg^{(a,e,i)}_i(x_i) \right] - 2\nu,
\]
since $\Pr_{e \sim E_{\rm gap}}\left[e \in E_{\rm nice}\right] \geq 1 - \nu$, and the inner expectation in the above expressions is a $[0,1]$-valued random variable. 

\subsection{Proof of Lemma \ref{lem:r-bias}}			\label{sec:r-bias}

The key technical tool used in the proof is the following lemma that bounds the variance of the weight of rounded solution by average correlation of the gap distribution.

\begin{lemma}			\label{lem:var-bound}
	\[
	\Ex_{{\bf G}}\left[\left(\Ex_i p_i\right)^2 \right] - \left(\Ex_{{\bf G}}\left[\Ex_i\left[ p_i\right]\right]\right)^2 \leq \Ex_{i,j}\left[|\rho_{ij}|\right],
	\]
	where $\rho_{ij} := \langle \bar{w}_i,\bar{w}_j \rangle$ is the pseudo-correlation between the variables $X_i,X_j$ under the pseudo-distribution $\theta$.
\end{lemma}

The above lemma is the key technical improvement on \cite{RT12}'s analysis. In \cite{RT12}, the authors show that that the LHS in the above lemma is upper bounded by $2^{O(R)}\cdot \Ex_{{i,j}}\left[|\langle w_i,w_j \rangle|\right]$ -- however since $R$ depends on the volume parameter $\delta$ of the SSE instance, this bound becomes unusable in this context. Our analysis avoids the $2^{O(R)}$ blow up by controlling the RHS using average correlation instead of average covariance. We defer the prove Lemma \ref{lem:var-bound} to Section \ref{sec:var-bound} and continue with the analysis. 

Next we have the following lemma which states that in expectation, the weight of the rounded solution is $\approx \mu$.

\begin{lemma}			\label{lem:exp-bound}
	Suppose we have that 
	\[
	\Pr_{i \sim G_{\rm gap}}\left[\max_{j \in [R]}{\sf Inf}^{(\mu_i)}_j\left[\sT^{(\mu_i)}_{1 - \eta}g_i\right] > \tau\right] \leq \nu. 
	\]
	Then,
	\[
	\Ex_{{\bf G}}\left[\Ex_{i \sim G_{\rm gap}}p_i\right]  \in \left[\mu \pm \tau^{O(\mu\kappa/r)} \right].
	\]
\end{lemma}
\begin{proof}
	By definition, we have that 
	\begin{align*}
		\Ex_{i \sim G_{\rm gap}}\Ex_{\bf G}\left[p_i\right]
		&= \Ex_{i \sim G_{\rm gap}} \Ex_{\bf G}\left[f_{[0,1]}\left(H_i(q_i)\right)\right] \\
		&= \Ex_{i \sim G_{\rm gap}} \Ex_{x_i \sim \{0,1\}^R_{\mu_i}}\left[f_{[0,1]}\left(H_i(x_i)\right)\right] \pm \left(O(\nu) +  \tau^{O(\eta\kappa/r)}\right) \\
		&= \Ex_{i \sim G_{\rm gap}} \Ex_{x_i \sim \{0,1\}^R_{\mu_i}}\left[f_{[0,1]}\left(\sT^{(\mu_i)}_{1 - \eta}g_i(x_i)\right)\right] \pm \left(O(\nu) +  \tau^{O(\eta\kappa/r)}\right) \\
		&= \Ex_{i \sim G_{\rm gap}} \Ex_{x_i \sim \{0,1\}^R_{\mu_i}}\left[g_i(x_i)\right] \pm \left(O(\nu) + \tau^{O(\eta\kappa/r^2)}\right) \\
		&=\mu \pm \left(O(\nu) +  \tau^{O(\eta\kappa/r^2)}\right),
	\end{align*}
	where the second step follows using the Invariance-principle based arguments from Lemma \ref{lem:round-nice}, and the last step follows from \eqref{eqn:cond-1}.
\end{proof}

By combining the variance and the expectation bounds from Lemmas \ref{lem:var-bound} and \ref{lem:exp-bound}, using Chebyshev's inequality we get that 
\begin{equation}					\label{eqn:wt-1}
\Pr_{{\bf G}}\left[\left|\Ex_{i \sim G_{\rm gap}} p_i - \mu \right| \geq \mu\sqrt{\gamma}\right] \leq \frac{\Ex_{i,j \sim G_{\rm gap}} \left[|\rho_{ij}|\right]}{\mu^2\gamma} \leq \frac{\gamma^2}{\mu^2\gamma} \leq \sqrt{\gamma},
\end{equation}
wherein the second inequality uses the fact that the gap-distribution $\theta:=\{\theta_S\}$ satisfies
\[
\Ex_{i,j \sim G_{\rm gap}}\left[|\rho_{ij}|\right] = \Ex_{i,j \sim G_{\rm gap}}\left[\Big|{\sf Corr}_{\theta}(X_i,X_j)\Big|\right] \leq \gamma^2.
\] 	
Now to finish the proof observe that $\Ex_{i \sim G_{\rm gap}}\left[\sigma(i)\right]$ is just a weighted sum of $|V_{\rm gap}|$ independent $\{0,1\}$-valued random variables where the weights are bounded in $\Theta(1/|V_{\rm gap})$, and 
\[
		\Ex_{\sigma}\left[\Ex_{i \sim G_{\rm gap}} \sigma(i)\right] = \Ex_{i \sim G_{\rm gap}} p_i. 
\]
Hence, using Chernoff Bound we have that 
\begin{equation}			\label{eqn:wt-2}
\Pr_{\sigma}\left[\left|\Ex_{i \sim G_{\rm gap}} \sigma(i) - \Ex_{i \sim G_{\rm gap}}p_i\right| \geq O(\sqrt{n}) \right] \leq \exp\left(-\Omega(\sqrt{n})\right).
\end{equation}
Therefore, combining \eqref{eqn:wt-1} and \eqref{eqn:wt-2}, with probability at least $1 - \sqrt{\gamma} - e^{-\Omega(\sqrt{n})}$, we have that 
\[
\left|\Ex_{i \sim G_{\rm gap}} \sigma(i) - \mu \right| 
\leq \left| \Ex_{i \sim G_{\rm gap}} \sigma(i) - \Ex_{i \sim V_{\rm gap}} p_i \right|
+ \left| \Ex_{i \sim G_{\rm gap}} p_i - \mu \right| \leq \mu\sqrt{\gamma} + O(n^{-1/2}) \leq 2\mu\sqrt{\gamma},
\]
since $|V_{\rm gap}| \gg 1/\mu^2\gamma$ in the setting of the reduction.

\subsection{Proof of Lemma \ref{lem:var-bound}}			\label{sec:var-bound}

We prove the following well known lemma.

\begin{lemma}[Folklore]
For every $i,j \in V_{\rm gap}$ we have 
\[
\left| \Ex_{{\bf G}}\left[p_ip_j\right] - \Ex_{{\bf G}}\left[p_i\right] \Ex_{{\bf G}}\left[p_j\right] \right| \leq  |\langle \bar{w}_i,\bar{w}_j\rangle|  = |{\rm Corr}_{\theta}(X_i,X_j)|.
\]
\end{lemma}

The proof of the above lemma uses elementary Hermite analysis; we point the readers to Section \ref{sec:hermite} for the basic facts used here.

\begin{proof}
Let $F_i:\mathbbm{R}^R \to [0,1]$ be the function 
\[
F_i(a):= f_{[0,1]}\Big(H_i(\mu_i\cdot{1}_R + \|w_i\|\cdot a)\Big),
\]
and define $F_j:\mathbbm{R}^R \to [0,1]$ similarly. Recall that in the rounding scheme, ${\bf G} \sim N(0,I_{R \times d})$ is a $R \times d$-dimensional Gaussian matrix. Let $\zeta_i := {\bf G} \cdot \bar{w}_i$ and $\zeta_j = {\bf G} \cdot \bar{w}_j$. Then by definition 
\[
F_i(\zeta_i) = f_{[0,1]}\Big(H_i(\mu_i\cdot{1}_R + \|w_i\|\cdot {\bf G}\bar{w}_i)\Big) = f_{[0,1]}(H_i(q_i)) = p_i, 
\]
and similarly, $p_j = F_j(\zeta_j)$. Furthermore, we also have the following properties:
\begin{itemize}
	\item $\zeta_i$ and $\zeta_j$ are marginally distributed as $N(0,I_R)$ i.e., $R$-dimensional standard Gaussian random variables.
	\item For every $u \in [R]$, $\zeta_i(u)$ and $\zeta_j(u)$ are $\rho_{ij}$-correlated.
\end{itemize}
Hereon, we will think of $F_i,F_j$ as functions in $L_2(\gamma^R)$, where $L_2(\gamma^R)$ is the space of square-integrable functions under the $R$-dimensional Gaussian measure. In particular, we can write $F_i$ and $F_j$ in the Hermite basis as 
\[
F_i = \sum_{\sigma} \wh{F_i}(\sigma)h_\sigma  
\ \ \ \ \ \ \ \textnormal{and} \ \ \ \ \ \
F_j = \sum_{\sigma} \wh{F_j}(\sigma)h_\sigma.
\]
Then,
\begin{align*}
	\Ex_{{\bf G}}\Big[p_ip_j\Big] 
	&= \Ex_{{\bf G}}\Big[F_i(\zeta_i)F_j(\zeta_j)\Big] \\
	&= \Big\langle F_i, U_{\rho_{ij}} F_j \Big\rangle_{\gamma^R} 				\tag{Inner product w.r.t $\gamma^R$} \\
	&= \sum_{\sigma} \rho^{|\sigma|}_{ij} \wh{F_i}(\sigma)\wh{F_j}(\sigma),			\tag{Plancharel's}		\\
	&= \wh{F_i}(\emptyset)\wh{F_j}(\emptyset) + \sum_{\sigma \neq \emptyset} \rho^{|\sigma|}_{ij} \wh{F_i}(\sigma)\wh{F_j}(\sigma) \\
	&= \Ex_{\bf G}\big[ p_i\big]\Ex_{{\bf G}}\big[p_j\big] + \sum_{\sigma \neq \emptyset} \rho^{|\sigma|}_{ij} \wh{F_i}(\sigma)\wh{F_j}(\sigma)
\end{align*}
which on rearranging and applying the triangle-inequality gets us
\begin{align*}
	\left|\Ex_{{\bf G}}\big[p_ip_j\big] - \Ex_{\bf G}\big[ p_i\big]\Ex_{{\bf G}}\big[p_j\big] \right|
	&\leq \sum_{\sigma \neq \emptyset} |\rho^{|\sigma|}_{ij}| |\wh{F_i}(\sigma)| |\wh{F_j}(\sigma)| \\
	&\leq |\rho_{ij}|\sum_{\sigma} |\wh{F_i}(\sigma)| |\wh{F_j}(\sigma)| \\
	&\leq |\rho_{ij}|\sqrt{\sum_\sigma \wh{F_i}(\sigma)^2}\sqrt{\sum_\sigma \wh{F_j}(\sigma)^2}		\tag{Cauchy-Schwarz} \\
	&= |\rho_{ij}| \|F_i\|_{L_2} \|F_j\|_{L_2} 			\tag{Parseval's}\\
	&\leq |\rho_{ij}| \|F_i\|_{L_\infty} \|F_j\|_{L_\infty} \\
	&\leq |\rho_{ij}|, 
\end{align*}
where the last inequality is due to the fact that $F_i$ and $F_j$ are $[0,1]$-valued functions.
\end{proof}

Using the above, we immediately get the following corollary.

\begin{corollary}						
\[
\Ex_{{\bf G}}\left[\left(\Ex_i p_i\right)^2 \right] - \left(\Ex_{{\bf G}}\left[\Ex_i\left[ p_i\right]\right]\right)^2 \leq \Ex_{i,j}\left[|\rho_{ij}|\right].
\]
\end{corollary}
\begin{proof}
Observe that 
\[
\Ex_{{\bf G}}\left[\left(\Ex_i p_i\right)^2 \right]
= \Ex_{i,j} \Ex_{{\bf G}}\left[p_ip_j\right] \\
\leq \Ex_{i,j}\left[\Ex_{{\bf G}}p_i\Ex_{{\bf G}}p_j\right] + \Ex_{i,j} \left[|\rho_{ij}|\right] \\
\leq \left(\Ex_{{\bf G}}\left[\Ex_{i}p_i\right] \right)^2 + \Ex_{i,j}\left[ |\rho_{ij}|\right] ,
\]
which on rearranging finishes the proof.
\end{proof}

%% file: appendix-lift.tex
\input{halfspace}

\section{Miscellaneous Fourier Analysis Facts}

In this section, we cover some basic facts that are quite well-known in the literature; we still state and prove them in the precise form that is needed in our proofs for completeness.

\subsection{Fourier Decay of $\cT^{(\Omega_i)}_\rho$}

\begin{fact}			\label{fact:decay-2}
	Let $h_i = h(x_i,z_i)$ be a function in the $2R$-dimensional probability space $\{0,1\}^R_{\mu_i} \otimes \{\bot,\top\}^R_\beta$. Then for any $d \geq 2$ we have 
	\[
	\left\|\left(\sT^{\Omega_i}_{1 - \eta}h\right)^{\geq d}\right\|^2_2 \leq (1 - \eta)^d.
	\]
\end{fact}
\begin{proof}
	Recall that we can write $\sT^{(\Omega_i)}_{1 - \eta}h_i$ with respect to the Fourier basis corresponding to the probability space $\{0,1\}^R_{\mu_i} \otimes \{\bot,\top\}^R_\beta$ as.
	\[
	\sT^{(\Omega_i)}_{1 - \eta}h_i(x_i,z_i) = \sum_{S,T \subseteq [R]} \wh{h_i}(S,T)(1 - \eta)^{|S| + |T|} \phi^{(\mu_i)}_S(x_i)\phi^{(\beta)}_T(z_i)
	\]
	Then 
	\[
	\left(\sT^{(\Omega_i)}_{1 - \eta}h_i\right)^{>d}(x_i,z_i) = \sum_{|S| + |T| > d}  \wh{h_i}(S,T)(1 - \eta)^{|S| + |T|} \phi^{(\mu_i)}_S(x_i)\phi^{(\beta)}_T(z_i).
	\]
	Therefore, using Parseval's identity we have
	\[
	\left\|\left(\sT^{(\Omega_i)}_{1 - \eta}h_i\right)^{>d}\right\|^2_2 
	= \sum_{|S| + |T| > d}(1 - \eta)^{|S| + |T|}\wh{h_i}(S,T)^2 \leq (1 - \eta)^d\sum_{S,T} \wh{h_i}(S,T)^2 \leq (1 - \eta)^d.
	\]
\end{proof}

\subsection{Transformations}

\begin{fact}					\label{fact:mult}
	Let $g = g(x_1,\ldots,x_R)$ where $(x_i)_{i \in [R]} \sim \{0,1\}^R_\mu$ for some $\mu \in [0,1]$. Then for any $\rho \in [0,1]$, the function $\sT^{(\mu)}_{\rho}g$ can be expressed as multi-linear polynomial in $x_1,\ldots,x_R$.
\end{fact}
\begin{proof}
	Writing $\sT^{(\mu)}_\rho g$ in the Fourier basis corresponding to $L_2(\{0,1\}^R_\mu)$, we get that 
	\begin{align*}
	\sT^{(\mu)}_\rho g(x) 
	&= \sum_{S \subseteq [R]} \wh{g}(S) \rho^{|S|}\phi^{(\mu)}_S(x) \\
	&= \sum_{S \subseteq [R]} \wh{g}(S) \rho^{|S|}\prod_{i \in S} \frac{x_i - \mu}{\sqrt{\sigma}} \\\
	&= \sum_{S \subseteq [R]} \wh{g}(S)(\rho/\sigma)^{|S|} \sum_{T \subseteq S} (-\mu)^{|T\setminus S|}\prod_{i \in T} x_i,
	\end{align*}
	which is clearly a multi-linear polynomial in $x_1,\ldots,x_R$.
\end{proof}

\subsection{Lasserre Feasibility}

\begin{claim}						\label{cl:feas}	
	Let $\theta$ be a degree-$\ell$ pseudo-distribution. Let $\hat{\theta} := \{\hat{\theta}_S\}_{|S| \leq \ell}$ be a collection of $R$-local distributions which are defined in terms of $\theta$ as follows. For every $S \subseteq [n]$ of size at most $\ell$, define $\hat{\theta}_S$ as
	\begin{itemize}
		\item Sample $X_S \sim \theta_S$.
		\item For every $i \in S$, do the following independently: with probability $1 - \eta$, set $\hat{X}_i = X_i$ and with probability $\eta$ sample $\hat{X}_i \sim \{0,1\}_\mu$.
	\end{itemize}
	Then $\hat{\theta}$ is a valid degree-$\ell$ pseudo-distribution as well.
\end{claim}	
\begin{proof}
	Given $\theta$, define $\theta^{(1)} = \{\theta^{(1)}_S\}_{|S| \leq \ell}$ as follows. For any subset $S$ of size at most $\ell$, define $\theta^{(1)}_S$ as the distribution over partial assignments in $\{0,1\}^S$ generated using the following process: 
	\begin{itemize}
		\item Sample $X_S \sim \theta_S$. 
		\item If $i \in S$, resample $X_i$ with probability $\eta$ from the distribution $\{0,1\}_\mu$. Output $X_S$.
	\end{itemize}
	We claim that $\theta^{(1)}$ is also a valid degree-$\ell$ pseudo-distribution: this follows from the observation that it is mixture of the degree-$\ell$ pseudo-distributions $\theta$ and $\theta'$ where $\theta'$ is the pseudo-distribution which is defined identically as $\theta$ on variables $X_2,\ldots,X_n$, and variable $X_1$ is always independently sampled from $\{0,1\}_\mu$.
	
	Summarizing, given $\theta$, we construct another degree-$\ell$ psuedo-distribution $\theta^{(1)}$ which re-randomizes only the first variable. Applying the above argument iteratively on coordinates $2,\ldots,n$, will yield a sequence of degree-$\ell$ pseudo-distributions $\theta^{(2)},\ldots,\theta^{(n)}$. The proof is concluded by observing that $\hat{\theta} = \theta^{(n)}$.
\end{proof}

\subsection{Hermite Analysis}		\label{sec:hermite}

Let $\gamma^R$ denote the $R$-dimensional Gaussian measure. Then $L_2(\gamma^R)$ denotes the vector space of real-valued functions defined over $\mathbbm{R}^R$ that are square-integrable w.r.t. $\gamma^R$. Then $L_2(\gamma^R)$ can be equipped with the inner product $\langle\cdot,\cdot\rangle_{\gamma^R}$ where for any $f,g \in L_2(\gamma^R)$ we have
\[
\Langle f,g \Rangle_{\gamma^R} = \Ex_{x \sim \gamma^R}\left[f(x)g(x)\right].
\]
We will need the notion of Hermite polynomials.
\begin{definition}[Hermite Polynomials]
	There exists an infinite sequence of polynomials $h_0 \equiv 1, h_1,h_2,\ldots$ where $h_i:\mathbbm{R} \to \mathbbm{R}$ is a polynomial degree $i$ such that $(h_i)_i$ forms an orthonormal basis for functions in $L_2(\gamma)$. In particular, any functions $f \in L_2(\gamma)$ admits a (Hermite) decomposition of the form
	\[
	f = \sum_{\sigma = 0}^{\infty} \wh{f}(\sigma)h_\sigma,
	\]
	where $\wh{f}(\sigma)$ is referred to as the Hermite coefficient corresponding to polynomial $h_\sigma$.
\end{definition}
For $R \geq 2$, by extending the above into the product space $\gamma^R$, for any $f \in L_2(\gamma^R)$, the corresponding Hermite decomposition is given as 
\[
f = \sum_{\sigma \in \mathbbm{Z}^R_{\geq 0}}\wh{f}(\sigma)h_{\sigma},
\]
where for any $\sigma = (i_1,\ldots,i_R) \in \mathbbm{Z}^R_{\geq 0}$ we have $h_\sigma(x) = \prod_{j \in [R]} h_{i_j}(x(j))$. We recall some basic facts about the Hermite decomposition of functions.

\begin{fact}[Plancharel's]
	For any $f,g \in L_2(\gamma^R)$, $\langle f,g \rangle_{\gamma^R} = \sum_{\sigma} \wh{f}(\sigma)\wh{g}(\sigma)$.
\end{fact}

\begin{fact}[Empty Fourier Coefficient]
	For any $f \in L_2(\gamma^R)$, we have $\wh{f}({\bf 0}_R) = \Ex_{x \sim \gamma^R} \left[f(x)\right]$.
\end{fact}

{\bf Noise Operator}. For any $\rho \in [-1,1]$, a pair of jointly distributed Gaussian vectors $g,h$ are said to be $\rho$-correlated if $g$ and $h$ are marginally distributed as $\gamma^R \equiv N(0,1)^R$ and for every $i \in [R]$, $\Ex\left[g(i)h(i)\right] = \rho$. We use the notation $h \underset{\rho}{\sim} g$ to denote a random draw of $\rho$-correlated pair $g$ and $h$. 

\begin{definition}[Gaussian Noise Operator]
	For any $\rho \in [-1,1]$, $U_\rho$ is a stochastic functional on $L_2(\gamma^R)$ which is defined as follows: for any $f \in L_2(\gamma^R)$ we define
	\[
	U_\rho f(g) = \Ex_{g \sim h}\left[f(h)\right]
	\]
\end{definition}

Similar to the noise operator in the finite probability space setting, we have the following fact.
\begin{fact}			\label{fact:herm-noise}
	For any $f \in L_2(\gamma^R)$ and $\rho \in [-1,1]$ the Hermite decomposition of $U_\rho f$ is uniquely given by 
	\[
	U_\rho f = \sum_{|\sigma|} \rho^{|\sigma|} \wh{f}(\sigma)h_{\sigma},
	\]
	where $|\sigma|$ denotes the number of non-zero indices in the multi-index $\sigma$.
\end{fact}

\section{Separating Example}			\label{sec:examp}

Consider the $3$-ary predicate $\psi:\{0,1\}^3 \to \{0,1\}$ whose accepting set is $\psi^{-1}(1) = \{(1,0,0),(0,1,1)\}$. Following \cite{GL22}'s notation, the set of minimal accepting strings of $\psi$, denoted by $\cM(\psi)$ is the accepting set itself. Since $(0,1,1) \in \psi^{-1}(1)$ has Hamming weight $2$, as per \cite{GL22}'s characterization, $\mu$-Biased CSPs on predicate $\psi$ should be at least as hard as the \dksh~problem of arity $2$ (i.e, the Densest-$k$-Subgraph problem with $k = \mu|V|$), and hence it is at least $\Omega(\mu \log(1/\mu))$ hard to approximate in their unbounded weight setting -- we point out that the hard instances in \cite{GL22}'s reduction crucially use the fact that the weights are allowed to be unbounded. On the other hand, for our setting we have the following lemma.

\begin{lemma}
	There exists a $\mu^{3/4}$-approximation algorithm for $\mu$-constrained instances of Max-CSP$(\psi)$.
\end{lemma}

Towards proving the above lemma, we shall need the following useful observation.
\begin{observation}				\label{obs:dks}
	Suppose $H = (V,E)$ is a \dks~instance with optimal value at least $c$. Then there exists an efficient algorithm that returns an assignment with value at least $c^3$.
\end{observation}
\begin{proof}
	The algorithm is the following: solve the basic SDP for \dks~on $H$, and then do independent rounding from the SDP solution. If the optimal value of the instance is $c$, then the optimal SDP value is also at least $c$. By averaging, at least $c/2$ edges would then have SDP value at least $c/2$. Fix any such edge $e \in E$. Then, denoting the local distributions corresponding to the SDP solution as $\theta$, we have 
	\[
	\Pr_{\theta}\Big[X_i = 1\Big], \Pr_{\theta}\left[X_j = 1\right] \geq \Pr_{\theta}\left[X_i = X_j = 1\right] \geq c/2.
	\]
	Therefore, any such edge will be induced by the sampled set with probability at least $c^2/4$, which implies that the expected number of edges induced by the rounded solution will be at least $c^3/8$. Using Chernoff bound, the weight of the rounded set will be $\approx \mu$ w.h.p, from which the claim follows.  
\end{proof}

Now we prove the above lemma.
\begin{proof}
	Let $H = (V,E,w)$ be the Max-CSP($\psi$) instance. Denote $c := {\sf Opt}_\mu(H)$. We consider two cases:
	
	{\bf Case(i)} Suppose $c \leq \mu^{1/4}$. Note that, since the weights are bounded by $1/|V|^{-1/100}$, randomly setting $\mu$-fraction of variables to $1$ will yield a feasible labeling that satisfies $\mu/2$-fraction of constraints using the assignment $(1,0,0)$ w.h.p, and hence this will yield a $\mu/\mu^{1/4} = \mu^{3/4}$-approximation.
	
	{\bf Case (ii)} Suppose $c > \mu^{1/4}$, and let $\sigma^*$ be the corresponding optimal labeling. 	Now suppose $\sigma^*$ satisfies at least $c/2$-fraction of constraints using $(1,0,0)$ assignment. Then solving the LP relaxation corresponding to the accepting string $(1,0,0)$ and rounding it (similar to Max-$k$-Coverage) yields a $\mu$-weight assignment that satisfies at least $(1 - 1/e)c/2$-fraction of constraints.
	
	Otherwise, if $\sigma^*$ satisfies at least $c/2$-constraints using the $(0,1,1)$ assignment, then we can do the following. Let $\tilde{H} = (V,\tilde{E})$ be the graph where for every edge $e = (i,j,k) \in E$, we introduce an edge $(j,k)$ of weight $w(e)$. Then note that there exists a set of weight $\mu$ that induces at least $c/2$ weight of edges in $\tilde{H}$. Then using the algorithm from Observation \ref{obs:dks}, we can find a set of weight $\mu$ that induces at least $c^3/8$-weight of edges in $\tilde{H}$ --  denote the corresponding labeling as $\sigma$. Note that $\sigma$ is not guaranteed to be a good labeling for $H$ as for some of the edges in $\tilde{H}$ induced by set indicated by the labeling, the first vertex from the corresponding constraint in $H$ might also be included in the set i.e., the assignment to the corresponding edge ends up being $(1,1,1)$ instead of $(0,1,1)$. To fix this, we construct another labeling $\sigma'$ from $\sigma$ by setting a random $\mu/2$-fraction of ones to zeros, and $\mu/2$-fraction of zeros to ones. It is easy to see that this is also a feasible labeling that will now satisfy at least $\Omega(c^3)$-constraints in $\tilde{H}$, thus yielding a $\Omega(c^2) = \Omega(\sqrt{\mu})$ approximation.
\end{proof}

%% file: halfspace.tex
\section{Multivariate Gaussian CDF Bound}				\label{sec:hspace}

We begin by establishing by recalling a couple of elementary facts about the Gaussian CDF.

\begin{fact}				\label{fact:f1}
	For every $\epsilon \in (0,1)$, there exists $\delta(\epsilon) \in (0,1)$ such that for every $\delta \leq \delta(\epsilon)$, we have that 
	\[
	(1 - \epsilon)\sqrt{2\log(1/\delta)} \leq |\Phi^{-1}(\delta)| \leq (1 + \epsilon)\sqrt{2\log(1/\delta)}.
	\] 
\end{fact}

\begin{fact}				\label{fact:f2}
	There exists $0<\delta_0 \leq 1/2$ such that the following holds for every $\delta \in (0,\delta_0]$. Let $t = \Phi^{-1}(\delta)$. Then for any $0 \leq \Delta \leq 1/\sqrt{4 \log(1/\delta)}$ we have $\Phi(t + \Delta) \leq 2\delta$.
\end{fact}
\begin{proof}
	Let $\delta_0$ be such that for every $\delta' \leq \delta_0$ we have $|\Phi^{-1}(\delta')| \leq \sqrt{3\log(1/\delta')}$, and fix a $\delta \in (0,\delta_0)$. Then by definition we have that
	\[
	\frac{\phi(t + \Delta)}{\phi(t)} = \frac{e^{-(t + \Delta)^2/2}}{e^{-t^2/2}} = e^{-\Delta^2/2 + \Delta|t|} \leq 1.1.
	\]
	Then,
	\[
	\Phi(t + \Delta) - \Phi(t) = \int^{t + \Delta}_t \phi(x)dx \leq 1.1\int^{t + \Delta}_t \phi(t)dx = 1.1 \Delta \phi(t) 
	\leq  \frac{1.1}{2\sqrt{\log(1/\delta)}} \cdot\delta\sqrt{2.1 \log(1/\delta)} \leq \delta, 
	\]
	where the second inequality uses Fact \ref{fact:f1}. Using the above inequality we can conclude that $\Phi(t + \Delta) \leq \Phi(t) + \delta = 2\delta$ which concludes the proof.
\end{proof}

Now we prove the following noise stability bound.

\subsection{Proof of Lemma \ref{lem:hspace}}

	Without loss of generality, assume that $\delta_r \leq \delta_{r-1} \leq \cdots \leq \delta_1$. Now, we begin by observing that 
	\[
	\Lambda_{\rho}\left(\delta_1,\ldots,\delta_r\right) 
	= \Pr_{g \sim N(0,1)^r} \Pr_{g_1,\ldots,g_r \underset{\rho}{\sim} g}\left[\forall i \in [r], g_i \leq t_i\right]
	\]
	where $t_i := \Phi^{-1}(\delta_i)$ for every $i \in [r]$. Now recall that $g_1,\ldots,g_r$ are generated using the following process:
	\begin{itemize}
		\item Sample $g \sim N(0,1)^r$ and $z_1,\ldots,z_r \sim N(0,1)^r$ independently.
		\item For every $i \in [r]$, set $g_i := \rho \cdot g + \sqrt{1 - \rho^2} \cdot z_i$.
	\end{itemize}
	Now denote $\Delta:= 1/\sqrt{2\log(1/\delta^*)}$. Then, we can proceed by upper bounding:
	\begin{align}
		&\Pr_{g \sim N(0,1)^r} \Pr_{g_1,\ldots,g_r \underset{\rho}{\sim} g}\left[\forall i \in [r], g_i \leq t_i\right]
		&\leq \Pr_{g \sim N(0,1)^r}\left[g \leq \frac{\Delta}{\rho}t_r \right]
		+ \Pr_{(g_i)^r_{i = 1}| g \geq \Delta \rho/t_r}\left[\forall i \in [r], g_i \leq t_i\right].			\label{eqn:pr-1}
	\end{align}
	Then for the first term, we have that  
	\begin{equation}					\label{eqn:pr-2}
	\Pr_{g \sim N(0,1)^r}\left[g \leq \frac{\Delta}{\rho}t_r \right] 
	= \Phi\left(\frac{\Delta t_r}{\rho}\right)
	\leq e^{-\Delta^2 t^2_r/2\rho^2} \leq e^{-2r\log(1/\delta^*)} = (\delta^*)^r/2,
	\end{equation}
	where the penultimate inequality uses $\rho^2 \leq 1/(16r\log(1/\delta^*)) \leq \Delta^2/2r$. This gives us a bound for the first term from \eqref{eqn:pr-1}. Next we proceed to address the second probability term from \eqref{eqn:pr-2}:
	\begin{align*}
		\Pr_{(g_i)^r_{i = 1}| g \geq \Delta t_r/\rho}\left[\forall i \in [r], g_i \leq t_i\right]
		&= \Pr_{z_1,\ldots,z_r| g \geq \Delta t_r/\rho}\left[\forall i \in [r], \rho g + \sqrt{1 - \rho^2} z_i   \leq t_i\right] \\
		&\leq \Pr_{z_1,\ldots,z_r| g \geq \Delta t_r/\rho}\left[\forall i \in [r],  z_i  \leq \frac{t_i - \rho g}{\sqrt{1 - \rho^2}}\right] \\
		&\leq \Pr_{z_1,\ldots,z_r}\left[\forall i \in [r],  z_i  \leq \frac{t_i + |t_r| \rho}{\sqrt{1 - \rho^2}}\right] \\
		&= \Pr_{z_1,\ldots,z_r}\left[\forall i \in [r],  z_i  \leq \frac{t_i + \Delta}{\sqrt{1 - \rho^2}}\right]		 \\
		&\leq \Pr_{z_1,\ldots,z_r}\Big[\forall i \in [r],  z_i  \leq t_i + \Delta\Big] 		\tag{Since $t_i + \Delta \leq 0$}\\
	\end{align*}
	To finish the proof we observe that 
	\[
	\prod_{i \in [r]}\Pr_{z_i \sim N(0,1)}\Big[ z  \leq t_i + \Delta\Big] 
	= \prod_{i \in [r]} \Phi\left(t_i + \Delta\right)
	\overset{1}{\leq} \prod_{i \in [r]} 2 \delta_i 
	= 2^r \prod_{i \in [r]} \delta_i,
	\]
	where inequality $1$ is due to the Fact \ref{fact:f2}.